\documentclass[12pt,dvipsnames]{article}
\usepackage[margin = 1 in]{geometry}
\usepackage{thmtools}
\usepackage{verbatim, booktabs, multirow}
\usepackage{graphicx,dirtytalk}
\usepackage[round]{natbib}
\usepackage{caption}
\usepackage{subcaption}
\usepackage[utf8]{inputenc}
\usepackage{enumerate}
\usepackage{euscript,lscape}
\usepackage{amsmath, amsthm,amsfonts, amssymb}
\usepackage[onehalfspacing]{setspace}
\usepackage{tikz}
\newcounter{mysubequations}
\renewcommand{\themysubequations}{(\roman{mysubequations})}
\newcommand{\mysubnumber}{\refstepcounter{mysubequations}\themysubequations}

\usetikzlibrary{arrows}
\usepackage{environ}
\usetikzlibrary{calc}
\usetikzlibrary{positioning}
\usetikzlibrary{arrows.meta}
\usepackage[normalem]{ulem}
\newtheorem{thm}{Theorem}
\newtheorem{lemma}{Lemma}
\newtheorem{prop}{Proposition}
\newtheorem{cor}{Corollary}
\newtheorem{defn}{Definition}

\newtheorem*{example*}{Example}

\usepackage[colorlinks = true,
            linkcolor = blue,
            urlcolor  = blue,
            citecolor = blue,
            anchorcolor = blue]{hyperref}

\newcommand{\R}{\mathbb{R}}

\newcommand{\gV}{\mathcal{V}}
\newcommand{\gD}{\mathcal{D}}

\usepackage{array}
\renewenvironment{abstract}
 {\small
  \begin{center}
  \bfseries \abstractname\vspace{-.5em}\vspace{0pt}
  \end{center}
  \list{}{%
    \setlength{\leftmargin}{5mm}
    \setlength{\rightmargin}{\leftmargin}%
  }%
  \item\relax}
 {\endlist}
\newcolumntype{H}{>{\setbox0=\hbox\bgroup}c<{\egroup}@{}}
\allowdisplaybreaks

\title{
{Marital Stability With Committed Couples: A Revealed Preference Analysis}\thanks{We thank Laurens Cherchye, Thomas Demuynck, Bram De Rock, and seminar participants in the RES Annual Conference, the Spring Meeting of Young Economists, the York-Durham Workshop, 18th ACEGD at ISI Delhi, University of Antwerp, Bristol, Essex, Glasgow, Leuven, and Louvain-la-Neuve for many helpful suggestions.}}
\author{Mikhail Freer\thanks{Department of Economics, University of Essex, m.freer@essex.ac.uk} 
\and 
Khushboo Surana\thanks{Department of Economics and Related Studies, University of York, khushboo.surana@york.ac.uk}
}
\begin{document}
\sloppy
\maketitle
\thispagestyle{empty}

\begin{abstract}
\noindent
\normalfont
We present a revealed preference characterization of marital stability where some couples are committed.
A couple is committed if they can divorce only with mutual consent.
We provide theoretical insights into the potential of the characterization for identifying intrahousehold consumption patterns.
We demonstrate that without price variation for private goods among potential couples, intrahousehold resource allocations can only be identified for non-committed couples.
We conduct simulations using Dutch household data to support our theoretical findings.
Our results show that with price variation, the empirical implications of marital stability allow for the identification of household consumption allocations for both committed and non-committed couples. 

\end{abstract}

\section{Introduction}
Households consist of multiple decision-makers with potentially different preferences.
Over the past few decades, structural household models have become increasingly popular  for analyzing intrahousehold allocations of time and resources \citep[see][]{chiappori2017static}. In the absence of direct observation on “who gets what” in the household, such frameworks are useful for understanding consumption patterns within households. 
Following \citealp{becker1973theory}, researchers have often combined household bargaining analysis with a marital matching framework to exploit individuals’ outside options as threat points. 
A series of recent studies have used a combination of the two theories of “who marries whom” and “who gets what” as a framework
for empirical work \citep[see, e.g.,][]{cherchye2017household,gousse2017marriage,weber2017collective}.

The existing literature on matching models typically assumes that marriages operate under a no-commitment framework, allowing either spouse to unilaterally end the marriage without the other's consent. However, in reality, it is more common for some or all couples to be ``committed'' in the sense that the right to divorce is held jointly by both partners.
This commitment can stem from various sources, such as legal divorce policies, the presence of children or the accumulation of shared assets. 
For example, countries like the Czech Republic (\citealp{hrusakova2003grounds}), Hungary (\citealp{weiss2003grounds}) and Japan (\citealp{akiba1995marriage,tanase2010post}) require mutual consent for divorce among married couples, while cohabiting couples can separate without formal procedures. 
Even in jurisdictions with no-fault unilateral divorce laws, mutual consent may be needed in special cases. In Russia, for instance, a husband cannot file for divorce without his wife's consent if she is pregnant or within a year of giving birth (\citealp{antokolskaia2003grounds}).
Furthermore, joint custody and the division of shared assets often necessitate mutual agreement to proceed with divorce. Recent research suggests that joint asset ownership, such as a house, can serve as a commitment mechanism (\citealp{lafortune2017tying,lafortune2023collateralized}).

This paper presents a revealed preference characterization of household consumption under the assumption of marital stability with commitment. Unlike most of the previous research which assumes that all couples can divorce unilaterally, we allow couples to be committed.\footnote{    
    A notable exception is \citealp{sun2021efficiency}, who study senior level job matching market with committed agents in an non-transferable utility (NTU) framework.
    However, our study is the first to analyze marriage matching with committed couples in an imperfectly transferable utility (ITU) framework.
    \citealp{browning2014economics} provide an overview of the literature on stable matching on the marriage market. See also \citealp{chiappori2016econometrics} for an overview on the literature on the econometric of matching models.
}
In our framework, committed couples can divorce only if both partners agree, whereas non-committed couples can divorce unilaterally. To explore the testable implications of marital stability, we analyze two scenarios: one where committed couples can make transfers to secure agreement for divorce, and one where they cannot. 
We then provide theoretical insights into the empirical bite of our characterizations.
We begin by investigating the issue of rationalizability of a data set, which involves determining the necessary conditions under which the model could be rejected by the data.
We find that if there is no price variation across counterfactual matches and if individual incomes are the same inside and outside marriages, then violations of the model can only come from non-committed couples.
Next, we investigate whether the characterizations can be of use for identifying intrahousehold allocations.
We show that indentifiability of intrahousehold allocations of committed couples requires the data set to have price variation for private goods across potential matches.

To support our theoretical results, we conduct simulation exercises using data from the Longitudinal Internet Studies for the Social Sciences (LISS) panel, which includes a sample of Dutch households. 
Given our theoretical findings indicating the difficulty in identifying intrahousehold allocations for committed couples, we focus on evaluating the effectiveness of our methods for these households. The simulations reveal that, with some variation in the prices of private goods, our framework for marital stability with transfers provides informative bounds on the intrahousehold allocations of committed couples.

A key ingredient of our method is the use of a revealed preference framework following the tradition of \citealp{afriat1967construction,diewert1973afriat}, and \citealp{varian1982nonparametric}. Our revealed preference characterization of marital stability is nonparametric, meaning it does not need a prior functional specification of individual utilities. As such, this nonparametric approach avoids potentially erroneous conclusions due to wrongly specified parametric forms and allows for individual-level heterogeneity in preferences.
Our results complement the revealed preference method developed by \citealp{cherchye2017household}, who derived the implications of stable marriages for household consumption under the implicit assumption of unilateral divorce. We extend their work by characterizing marital stability in the presence of committed couples. 
We demonstrate that Cherchye et al.'s result is a special case of our general framework where none of the couples are committed.

We emphasize two specific features of our framework. First, our model uses an imperfectly transferable utility (ITU) framework. 
We account for intrahousehold consumption transfers but do not assume perfectly transferable utility (TU). 
The latter assumption (TU) would require strict conditions on utility functions such as generalized quasi-linear form.
While appealing from a theoretical perspective, the TU framework imposes substantial restrictions on individuals' preferences which may not hold in general \citep[see][]{chiappori2020transferable}. 
Our paper contributes to the growing strand of literature that uses an ITU framework to model marriage markets \citep[see][]{legros2007beauty,choo2013collective,chiappori2017matching,galichon2019costly}.\footnote{
    There is also important literature on revealed preferences of matching in NTU and TU frameworks \citep[see][]{echenique2008matchings,echenique2013revealed,demuynck2022revealed}. However, this literature primarily focuses on testing the stability of the matching market rather than on identifying intrahousehold resource allocation.}
Second, our analysis focuses on static equilibrium conditions for marital stability within a competitive, frictionless marriage market \citep[see][]{shapley1971assignment,becker1973theory}.
The static nature of our model is a substantial simplification of real-world marital behavior. Intertemporal considerations and the ease with which one can meet potential partners are particularly relevant when analyzing household decisions with long-term consequences (such as fertility). Nonetheless, the equilibrium concept of marital stability that we consider provides a natural starting point from which to analyze individuals' marital and consumption behavior. It can be used as a building block for more advanced dynamic models \citep[for a review, see][]{chiappori2017static}. 
Additionally, using a static framework is advantageous for our primary goal: investigating model identification under minimal data requirements. This approach allows for analysis using cross-sectional data, rather than necessitating panel data.

The remainder of the paper is organized as follows.
Section \ref{sec:model} introduces the structural components of the model, defines the core and shows that it is non-empty. 
Section \ref{sec:results} presents the revealed preference concepts and our main findings, including characterizations of marital stability with and without transfers, as well as theoretical results concerning the models' empirical content. 
Section \ref{sec:empirical} describes the simulation analysis.
Section \ref{sec:conclusion} provides concluding remarks.
All proofs omitted from the main text are provided in the Appendix.

\section{Stable Marriage with Committed Couples}
\label{sec:model}
\subsection{Preliminaries}

\paragraph*{Marriage Market and Committed Couples.} 
We consider a marriage market formed by a finite set of men $M$ and a finite set of women $W$. 
To ease the notational burden, in the following exposition we will assume that $\vert M \vert =\vert W \vert$.
Importantly, this assumption is not critical for our results.\footnote{
    The number of men and women differs when there are singles in the data. Our formal exposition do not explicitly discuss singles, but they can modelled in the same way. Specifically, single males (females) can be included in the analysis by considering them as (virtual) couples where the female (male) has zero consumption. 
    Moreover, we will include singles in our simulation analysis.}
We refer to a man as $m\in M$, a woman as $w\in W$ and the option of staying alone as $\emptyset$.

Let us denote by $\gV\subseteq M\times W$, the set of {committed couples}. Consider a couple $(m,w)$ formed by man $m$ and woman $w$. If the couple $(m,w)\in \gV$, they can only divorce upon mutual consent.
We assume that the set of committed couples $\gV$ is exogenously defined.\footnote{
    Given the static nature of the model, we abstain from modelling the decision to commit. In principle, the commitment status of couples could be treated as unobserved binary variables; however, for the sake of compactness, we do not pursue this approach in the following analysis. For further discussion, see Appendix \ref{app:LinearProgramming:unknownV}.
}

\paragraph{Consumption.}
For any couple $(m,w)$, we assume that the couple consumes two types of goods. 
The first type consists of goods consumed privately by the individual members.
Let us denote by $q^{m}_{m,w}$ and $q^w_{m,w} \in \R^k_{+}$, the vectors of private consumption of man $m$ and woman $w$, respectively.
The second type of consumption consists of goods consumed publicly within the household.
Let $Q_{m,w} \in \R^K_+$ denote the vector of public consumption.


\paragraph{Incomes and Prices.} 
Consumption decisions are made under (linear) budget constraints defined by prices and income. 
The income of the couple $(m,w)$ is denoted by $y_{m,w} \in \R_{++}$, while the incomes of man $m$ and woman $w$ when single are denoted by $y_{m,\emptyset} \in \R_{++}$ and $y_{\emptyset,w} \in \R_{++}$, respectively.
Next, let $p_{m,w}\in \R^k_{++}$ be the price vector of the private goods and $P_{m,w}\in \R^K_{++}$ be the price vector of the public goods faced by the couple $(m,w)$. 
Similarly, let $p_{m,\emptyset} \; (p_{\emptyset,w}) \in \R^k_{++}$ be the price vector of the private goods and $P_{m,\emptyset} \;(P_{\emptyset,w})\in \R^K_{++}$ be the price vector of the public goods faced by man $m$ (woman $w$) as single. 
For the pair $(m,w)$, the consumption possibilities are determined by the budget set:
$$ B_{m,w} = \lbrace (q^m, q^w, Q) \ \vert \ p_{m,w}(q^m + q^w) + P_{m,w} Q \leq y_{m,w}\rbrace.
$$

\paragraph*{Preferences.}
Individuals derive utility from both their private consumption and public consumption.
We assume that every individual $i\in M\cup W$ is endowed with a continuous, monotone, and concave utility function $u^i: \R^k_+ \times \R^K_+ \rightarrow \R$.

\paragraph*{Marriage Game and Allocation.}
We define the marriage game as a 5-tuple 
$$
\Gamma = (M, W, \gV, B, U),
$$
where $M$ and $W$ are the sets of men and women in the marriage market, $\gV \subseteq M\times W$ is the set of committed couples, $B = \{p_{m,w}, P_{m,w}, y_{m,w}\}$ is the set of prices and income faced by all $(m,w) \in (M \cup \lbrace\emptyset\rbrace) \times (W \cup\lbrace\emptyset\rbrace)$, and $U = \{u^i\}_{i\in M\cup W}$ is the set of utilities for all individuals in the market.

The outcome of the marriage game constitutes an allocation which comprises of a matching function and consumption bundle for each individual in the given matching. 
Specifically, the allocation is given by
$$
\alpha = \Big( \sigma, \{q^m_{m,\sigma(m)}, q^{\sigma(m)}_{m,\sigma(m)}, Q_{m,\sigma(m)}\}_{m \in M} \Big),
$$
where $\sigma: M\cup W \rightarrow M\cup W$ is a matching function describing who is married to whom and satisfying the following properties:
\begin{itemize}
    \item [--] $\sigma(m)\in W$ for every $m\in M$,
    \item [--] $\sigma(w)\in M$ for every $w\in W$,
    \item [--] $w = \sigma(m)$ if and only if $m=\sigma(w)$.
\end{itemize}
\noindent
Additionally, the consumption allocations are such that they satisfy the associated household budget constraint. That is, for any $m \in M$,
$$
p_{m,\sigma(m)} (q^m_{m,\sigma(m)} + q^{\sigma(m)}_{m,\sigma(m)}) + P_{m,\sigma(m)} Q_{m,\sigma(m)} \le y_{m,\sigma(m)}.
$$
We assume that the allocation defines the public consumption as well as the individuals’ private consumption for the matched couples, but  not  for  other  potential  (unmatched)  couples.


\subsection{Core}

The standard definition of the core requires that no coalition of agents can benefit by deviating from the allocation (\say{no blocking coalitions}). 
However, in our setting, the presence of committed couples restricts the set of coalitions that can deviate.
Thus, we say that an allocation belongs to the core if there are ``no permissible blocking coalitions''. 
To formally define the core, we first describe what constitutes a coalition and when it is permissible. We then define what it means for a coalition to be blocking. 
An individual might find preferable outside options compared to their current marriage, while their spouse might be at loss upon divorce. 
Therefore, in a committed relationship, there is an incentive for the partner who prefers divorce to compensate their spouse to obtain mutual consent for the divorce. Consequently, it is essential to consider that committed partners might promise transfers to each other in the event of divorce. 
In our formal exposition, we consider two scenarios of what constitutes a blocking coalition. In the first scenario, we assume committed couples cannot make monetary transfers to their current partners to incentivize them to agree to a divorce. We refer to the core of the marriage game in this scenario as the ``core'' (Definition \ref{def:Core}). 
In the second scenario, we allow committed couples to use transfers to incentivize their current partners to agree to divorce. We refer to the core in this scenario as the ``core with transfers'' (Definition \ref{def:CoreT}).\footnote{
    In standard matching models, which do not consider the commitment status of couples, transfers between ex-partners do not play a role. However, in our setting, allowing for transfers between committed couples upon divorce yields different empirical implications.
}

\paragraph{Permissible Coalition.}
We define a coalition as a tuple $(S,\hat \sigma)$, where $S\subseteq M\cup W\cup \{\emptyset\}$ is the set of members of the coalition and $\hat\sigma: S\rightarrow S$ is a matching function among the members of the coalition.\footnote{
    Note that $\lbrace \emptyset \rbrace$ represents the set of virtual partners, which symbolizes the option of remaining single for the individuals.
    }
Commitment status of the couples imposes restrictions on the feasibility of the coalitions that can block a given allocation.
We say that a coalition $(S,\hat\sigma)$ is \emph{permissible} if for every $i\in S$, if $(i,\sigma(i))\in \gV$, then $\sigma(i)\in S$.
Intuitively, permissibility of a coalition requires that if a member of the set of committed couple is part of the coalition, then the spouse should also be a member of the coalition. This notion of permissible coalition is in close spirit to cooperative games on graphs (see \citealp{myerson1977graphs}).


\paragraph{Blocking Coalition (without Transfers).}
A coalition is said to be blocking if the members of the coalition can improve upon the current allocation. To define this concept more formally, we start by defining blocking pairs.
A potential couple $(m,w)$ is called a {\it weakly blocking pair} if both $m$ and $w$ weakly prefer to marry each other than to stay in their current marriages.
Formally, $(m,w)$ is a weakly blocking pair if there is a consumption bundle $(q_{m,w}^m,q_{m,w}^w, Q_{m,w})$ such that 
$$
p_{m,w}(q_{m,w}^m + q_{m,w}^w) + P_{m,w} Q_{m,w} \leq y_{m,w},
$$
and
\begin{align*}
    & u^m(q^m_{m,w},Q_{m,w}) \geq u^m ( q^m_{m,\sigma(m)}, Q_{m,\sigma(m)}), \\
    & u^w(q^w_{m,w},Q_{m,w}) \geq u^w ( q_{\sigma(w),w}^w, Q_{\sigma(w),w}).
\end{align*}
A pair $(m,w)$ is a {\it blocking pair} if it is weakly blocking and at least one of the partners ($m$ or $w$) is strictly better-off in the new match. That is, at least one of the two inequalities above is strict.
A coalition $(S,\hat\sigma)$ is a {\it blocking coalition} if every rematched couple $(m, \hat\sigma(m)) \in S$ is weakly blocking and at least one is blocking.


\paragraph{Blocking Coalition with Transfers.}
Consider a coalition $(S, \hat{\sigma})$.
For any $m \in M$, let us denote by $t_{m,\sigma(m)}$ the transfer that man $m$ commits to pay to his current match $\sigma(m)$ upon divorce.
Note that while we generally interpret the transfer as going from $m$ to his partner $\sigma(m)$, this is only true when $t_{m,\sigma(m)}$ is positive. If $t_{m,\sigma(m)}$ is negative, $\sigma(m)$ pays money to her partner $m$. 
We assume that $t_{m,\sigma(m)}\ne0$ only if both $m$ and $\sigma(m)\in S$. This assumption implies that transfers between partners are zero when at least one of the partners is not a member of the coalition (e.g., if they are non-committed). This is because if a couple can divorce unilaterally, there is no need to incentivize the partner to divorce. 
This condition guarantees that there is no monetary transfer between the coalition and the rest of the individuals.

To define a blocking coalition with transfers, we extend the definition of a coalition as a triple $(S,\hat\sigma,t)$, where $t=\{t_{m,\sigma(m)}\}_{m \text{ or } \sigma(m) \in S}$ denotes the set of transfers. When transfers between ex-partners are allowed, a (weakly) blocking pair is defined similarly as above. Specifically, given transfers $(t_{m,\sigma(m)}, t_{\sigma(w),w})$, a pair $(m,w)$ is a weakly blocking pair if there is a consumption bundle $(q_{m,w}^m,q_{m,w}^w, Q_{m,w})$ such that 
$$
p_{m,w}(q_{m,w}^m + q_{m,w}^w) + P_{m,w} Q_{m,w} \leq y_{m,w} - t_{m,\sigma(m)} + t_{\sigma(w),w},
$$
and
\begin{align*}
    & u^m(q^m_{m,w},Q_{m,w}) \geq u^m ( q^m_{m,\sigma(m)}, Q_{m,\sigma(m)}), \\
    & u^w(q^w_{m,w},Q_{m,w}) \geq u^w ( q_{\sigma(w),w}^w, Q_{\sigma(w),w}).
\end{align*}
If at least one of the two inequalities above is strict, we say that the pair $(m,w)$ is a {\it blocking pair}.
A coalition $(S,\hat\sigma, t)$ is a {\it blocking coalition with transfers} if every rematched couple $(m, \hat\sigma(m)) \in S$ is weakly blocking and at least one is blocking.

\paragraph*{Core.}
We now turn to defining the core of the marriage game. 
\begin{defn}
\label{def:Core}
An allocation $\alpha$ is in the {\bf core}, $C(\Gamma)$, of marriage game $\Gamma$, if there are no permissible blocking coalitions.
\end{defn}

\begin{defn}
\label{def:CoreT}
An allocation $\alpha$ is in the {\bf core with transfers}, $C_T(\Gamma)$, of a marriage game $\Gamma$, if there are no permissible blocking coalitions with transfers.
\end{defn}

\noindent
If allocation $\alpha$ is in the core, we say that the corresponding allocation is \textit{stable}.
Note that Pareto efficiency of within-household resource allocations follows directly from the fact that the matching is in the core. {This is because the requirement of no permissible blocking coalitions applied on the observed couples implies that there cannot be another feasible household allocation that results in both spouses being better off and at least one spouse being strictly better off than the allocation ($q^m_{m,\sigma(m)}, q^{\sigma(m)}_{m,\sigma(m)}, Q_{m,\sigma(m)}$).}
Indeed, if an observed couple is not Pareto efficient, they would always be capable of forming a permissible coalition and would, therefore, block themselves.
This aligns our model within the collective household literature, which was first introduced by \citealp{chiappori1988rational,chiappori1992collective}.

We conclude this expeditionary section by establishing that the core (with transfers) is non-empty.\footnote{
    For the purposes of this paper, our analysis is confined to proving the existence of a core allocation. We have not investigated the development of algorithms that might generate such allocations. This is primarily because, unlike other matching problems (such as college admission problems) which can be centralized, marital matching is an inherently decentralized process. Additionally, since we are working with an ITU framework within households, defining a partial order to examine the potential lattice structure of core allocations is challenging as there are no natural candidates. For a NTU setting, we refer readers to \citealp{sun2021efficiency}, who study many-to-one matching with contracts in a labor market context. They demonstrate the existence of core allocations, establish that the core contains a lattice structure, and show that top trading cycle and a version of deferred acceptance mechanisms yield core allocations.}
Showing non-emptiness of core allocations is important as, in our following revealed preference characterization, we start by assuming that the observed marriages are stable.
Thus, guaranteeing that there is always a stable allocation is important for validity of the revealed preference analysis.

\begin{prop}
\label{prop:NonEmptyCore}
Given a marriage game $\Gamma$, both $C(\Gamma)$ and $C_T(\Gamma)$ are non-empty.
\end{prop}

Appendix \ref{app:Proofs} presents the proof of Proposition \ref{prop:NonEmptyCore}, which builds on a general result of \citealp{alkan1990core} and \citealp{cherchye2017household}. We note that the existence of core allocations does not necessarily imply a unique stable matching. However, non-uniqueness of core allocations does not undermine the validity of the testable implications and the set identification results we show below. Given the necessary and sufficient nature of the stability conditions, we know that for any utility function rationalizing the observed data, the corresponding value of intrahousehold allocation must lie within the bounds that defined by our method. Importantly, this argument does not depend on the utility function we construct in our sufficiency proof. 


\section{Revealed Preferences}
\label{sec:results}
In this section, we provide a characterization of household consumption under the assumption of marital stability.
Following the revealed preference tradition, our goal is to identify the set of behavioral restrictions that align with core allocations. This means that every core allocation must satisfy these ``revealed preference'' conditions, and any behavior that satisfy these conditions is considered core-compatible.\footnote{
    This contrasts with standard matching problems, which typically assume that utilities/preferences are observed and then seek core allocations and their properties. Instead, ours is a revealed preference problem where we observe the rules of the game and the allocation, and aim to determine if utility functions can be constructed such that the observed data are consistent with a core allocation.} 
We start by defining the type of data that we consider and when a data set is said to be rationalizable as a core allocation.
A data set is rationalizable as a core allocation if there are individual utilities such that the data could constitute a core allocation.
We also present a graph representation of the marriage market, which we use to express the revealed preference characterization. Next, we outline the revealed preference conditions that a data set needs to satisfy to be rationalizable as a core allocation. Finally, we discuss the empirical implications of these revealed preference conditions. Specifically, we examine (1) whether the conditions are effective in detecting violations of marital stability and (2) whether they can be used to infer intrahousehold consumption allocations, which we assume are not directly observed in the data.

\subsection{Set-up}
\paragraph{Data and Rationalization.}
So far, we assumed that an allocation is defined by a matching function ($\sigma$) and household consumption allocations $(q^m_{m,\sigma(m)}, q^{\sigma(m)}_{m,\sigma(m)}, Q_{m, \sigma(m)})$ for all matched couples $(m, \sigma(m))$ with $m \in M$. In practice, most household surveys do not provide information on individual private consumption, but only the aggregate consumption bundles.
Thus, instead of observing $q^m_{m,\sigma(m)}$ and $q^{\sigma(m)}_{m,\sigma(m)}$, we only observe $q_{m,\sigma(m)}$. We will account for this in our set-up by considering individual private consumption bundles as unknowns that satisfy the following adding-up conditions: 
$$
q_{m,\sigma(m)} = q^m_{m,\sigma(m)} + q^{\sigma(m)}_{m,\sigma(m)}.
$$

We assume a data set $\gD$ that contains  (1) the set of men $M$, (2) the set of women $W$, (3) the  set of committed couples $\gV$,\footnote{
    In practice, we only need to know which of the matched couples are committed. This information can be derived from the legal system; for instance, if married couples are governed by mutual consent divorce laws while cohabiting couples can separate without formal procedures, the marital status can indicate committed couples. Additional indicators such as the presence of children or shared assets, commonly available in household surveys, can also be used (recent evidence suggests that owning joint assets, like a house, serves as a commitment device; see \citealp{lafortune2017tying,lafortune2023collateralized}). Our model theoretically allows for the identification of committed couples by treating their commitment status as unknown variables within the rationalizability conditions we describe later. Although we do not explore this approach in the main text, we provide a linear programming formulation in Appendix \ref{app:LinearProgramming:unknownV} that can be used for this identification.}
(4) the prices and income $B = \lbrace p_{m,w}, P_{m,w}, y_{m,w} \rbrace$ for all $ (m,w) \in (M \cup \lbrace \emptyset \rbrace) \times (W \cup \lbrace \emptyset \rbrace)$, (5) the matching function $\sigma$ and (6) the consumption bundles $\lbrace q_{m,\sigma(m)}, Q_{m,\sigma(m)} \rbrace $ for all matched couples $(m, \sigma(m))$ with $m \in M$. Items (1)-(4) correspond to the elements of the marriage game and items (5)-(6) correspond to the allocation.
Thus,
$$
\gD = \Big \lbrace M,W,\gV,B, \sigma, \{q_{m,\sigma(m)}, Q_{m,\sigma(m)}\}_{m \in M} \Big\rbrace.
$$

\noindent
We can now state when a given data set $\gD$ is said to be rationalizable.
Specifically, to test whether the data set is rationalizable we need to check if there exist individual utility functions and private consumption bundles such that the allocation belongs to the core. 

\begin{defn}
\label{def:Rationalizble}
A data set $\gD$ is {\bf rationalizable as a core (with transfers) allocation} if there exist continuous, monotone, and concave utility functions $u^m$ and $u^w$ for every $m\in M$ and $w\in W$, and individual private consumption bundles $q^m_{m,\sigma(m)}, q^{\sigma(m)}_{m,\sigma(m)} \in \R^k_+$, with
$$
q^m_{m,\sigma(m)} + q^{\sigma(m)}_{m,\sigma(m)} = q_{m,\sigma(m)},
$$
such that the allocation $\alpha =\big(\sigma, \{q^m_{m,\sigma(m)}, q^w_{\sigma(w),w}, Q_{m,\sigma(m)}\}_{m\in M, w\in W}\big)$ is in the core (with transfers).
\end{defn}

\paragraph{Graph Representation.}
We introduce a description of the marriage market in terms of graph theory by considering a finite, weighted directed graph of the form $G = (V, E, A)$, where $V$ is the set of vertices, $E$ is the set of edges and $A$ is a weighting function $A: E \rightarrow \R$.\footnote{ 
Our graph theoretic interpretation differs from that used in the matching literature as we start under existent matching and look for potential blocking coalitions.
}
Given a data set $\gD$, the associated marriage market can be presented as a weighted directed graph where each matched couple in the data is a vertex and a directed edge in this graph represents a potential pair where the male of the outgoing vertex is matched with the female of the incoming vertex. For example, the edge from vertex $(m,\sigma(m))$  to vertex $(\sigma(w),w)$ would represent the potential pair $(m,w)$.   
The weighting function, $A(\gD)$, is defined as follows.
The weight of an edge going from $(m,\sigma(m))$ to $(\sigma(w),w)$, which represents the potential pair $(m,w)$, is denoted by $a_{m,w} \in A(\gD)$ and defined as:
\begin{align*}
    a_{m,w} =  p_{m,w} (q^m_{m,\sigma(m)} + q^w_{\sigma(w),w}) + P^m_{m,w}Q_{m,\sigma(m)} + P^w_{m,w}Q_{\sigma(w),w} -y_{m,w},
\end{align*}
with
$$
P^m_{m,w} + P^w_{m,w} = P_{m,w}.
$$

The weight function $A(\gD)$ can be interpreted based on revealed preferences. 
First, the edge weight $a_{m,w}$ defines individual prices $P^m_{m,w}, P^w_{m,w}\in \R^K_{+}$ which reflect the willingness-to-pay of $m$ and $w$, respectively, for the public consumption in the allocation ($q^m_{m,w}, q^w_{m,w}, Q_{m,w}$). Essentially, these prices capture the fraction of market prices of the public goods that are borne by $m$ and $w$.
These prices can be seen as Lindahl prices that support the efficient consumption of the public good: they must add up to the observed market prices ($P^m_{m,w} + P^w_{m,w} = P_{m,w}$), so as to be consistent with Pareto efficiency. This concept of personalized prices complements the concept of personalized quantities of private consumption.

Next, the edge weight $a_{m,w}$ specifies the income that the potential pair $(m,w)$ would be left with if they would buy the bundle they consume in their current matches $(q^m_{m,\sigma(m)}, Q_{m, \sigma(m)})$ and $(q^w_{\sigma(w),w}, Q_{\sigma(w),w})$ with their budget conditions (prices $(p_{m,w}, P_{m,w})$ and income $y_{m,w}$).
As we will show, our revealed preference characterization of core allocation will impose further restrictions on the weights $a_{m,w}$.

We now turn to defining the notion of a path of remarriages.
A \textit{path of remarriages} is defined by a set of agents $S = \{m_1,\ldots, m_{n-1}, \sigma(m_2),\ldots,\sigma(m_{n})\}$ and a matching $\hat\sigma: S\rightarrow S$ such that $\hat\sigma(m_j) = \sigma(m_{j+1})$ for every $1 \leq j < n$.
We can represent this path of remarriages as a path in the above directed graph\footnote{A path in a directed graph is a sequence of edges, which connects a sequence of vertices.} in which the sequence of vertices is $((m_1, \sigma(m_1)),(m_2,\sigma(m_2)),\ldots, (m_{n},\sigma(m_{n})))$ and the sequences of edges is $((m_1, \sigma(m_2)),(m_2,\sigma(m_3)),\ldots, (m_{n-1},\sigma(m_{n})))$. The set of edges in this path specify who is remarrying whom. In what follows, we will denote such a path as $\rho = (m_1, (m_2,\sigma(m_2)), \cdots, (m_{n-1}, \sigma(m_{n-1})), \sigma(m_n))$. Note that this path of remarriages would correspond to a permissible coalition if either $m_1 = m_n$ (the path is a cycle of remarriages) or both $(m_1,\sigma(m_1))$ and $(m_n,\sigma(m_n)) \notin \gV$ (couples at both ends of the path are not committed). Both of these cases ensure that both spouses of committed couples are in the coalition.
Clearly, a permissible path of remarriages corresponds to a permissible coalition.
In the Appendix (Lemma \ref{lemma:Coalition2Path}), we demonstrate that every permissible coalition contains a (permissible) path of remarriages.

\subsection{Revealed Preference Conditions}
We now provide revealed preference characterization of a data set $\gD$ that is rationalizable in the sense of Definition \ref{def:Rationalizble}. We first present the revealed preference conditions for the case in which no transfers between current partners post-divorce are allowed. Next, we allow for post-divorce monetary transfers between partners.


\paragraph{Core (without Transfers).}
Marital stability without transfers requires that there is no permissible blocking coalition. To recall, a blocking coalition is a coalition endowed with a rematching such that every pair is weakly blocking and at least one is blocking. 
Consider a potential pair $(m,w)$ in a coalition and suppose that for a given consumption allocations $q^m_{m,\sigma(m)}$ and $q^w_{\sigma(w),w}$, and individual prices $P^m_{m,w}$ and $P^w_{m,w}$ (with $P^m_{m,w} + P^w_{m,w} = P_{m,w}$), we have that $a_{m,w} < 0$. Then,
\begin{equation*}
    p_{m,w} (q^m_{m,\sigma(m)} + q^w_{\sigma(w),w}) + P^m_{m,w}Q_{m,\sigma(m)} + P^w_{m,w}Q_{\sigma(w),w} < y_{m,w}.
\end{equation*}
The left hand side of this inequality represents the total value of the currently consumed bundles by man $m$ ($p_{m,w}q^m_{m,\sigma(m)} + P^m_{m,w}Q_{m,\sigma(m)}$) and woman $w$ ($p_{m,w}q^w_{\sigma(w),w} + P^w_{m,w}Q_{\sigma(w),w}$), evaluated at the prices they will face in the new match, using personalized prices to evaluate the public goods. If their income $y_{m,w}$ exceeds this sum, they can reallocate their income so that both individuals are better off than in their current marriages (with at least one strictly better off). This will make $(m,w)$ a blocking pair.  

Consider the coalition corresponding to a path of remarriages $\rho = (m_1, (m_2,\sigma(m_2)), \cdots, (m_{n-1}, \sigma(m_{n-1})), \sigma(m_n))$. If $a_{m_r,\sigma(m_{r+1})}\le 0$ for all $ 1 \leq r < n$, with at least one inequality being strict, then the path of remarriages $\rho$ would specify a  blocking coalition. 
Therefore, given a data set $\gD$, if the observed marriages are stable, there must exist feasible intrahousehold allocations within current marriages and individual prices such that if the edge weights along any path of remarriages $\rho$ are negative and at least one is strictly negative, it must not be permissible.
Following this logic, Definition \ref{def:PathConsistency} defines the concept of \emph{path consistency}, and Theorem \ref{thm:PathConsistency} states that path consistency is a necessary and sufficient condition for rationalizability with core.

\begin{defn}
\label{def:PathConsistency}
Given a data set $\gD$, $A(\gD)$ satisfies {\bf path consistency} if, {for all $m \in M$ and $w \in W$}, there are
$$
q^m_{m,\sigma(m)}, q^w_{\sigma(w),w} \in \R^k_+ \text{ and } P^m_{m,w}, P^w_{m,w} \in \R^K_{++}
$$
with
$$
q^m_{m,\sigma(m)} + q^{\sigma(m)}_{m,\sigma(m)} = q_{m,\sigma(m)} \text{ and }  P^m_{m,w} + P^w_{m,w} = P_{m,w},
$$
such that for every path of remarriages $\rho = (m_1, (m_2,\sigma(m_2))$, $\ldots$, $(m_{n-1},\sigma(m_{n-1})), \sigma(m_n))$ if 
$$
a_{m_r,\sigma(m_{r+1})} \le 0 \; \text{ for all } \; 1 \leq r < n
$$
with at least one inequality being strict then,
\begin{itemize}
    \item [(i)] $(m_1,\sigma(m_1)) \in \gV$ or $(m_n,\sigma(m_n))\in \gV$, and
    \item [(ii)] $m_1\ne m_n$.
\end{itemize} 
\end{defn}

\begin{thm}
\label{thm:PathConsistency}
A data set $\gD$ is rationalizable as a core allocation if and only if $A(\gD)$ satisfies path consistency.
\end{thm}

\noindent
Intuitively, the path consistency condition ensures that the observed data is rationalizable by a stable matching by requiring any path of remarriages that is blocking to be non-permissible.
Condition $(i)$ in Definition \ref{def:PathConsistency} guarantees that the path of remarriages $\rho$ cannot start and end with non-committed couples while condition $(ii)$ guarantees that $\rho$ cannot be a cycle. If either of the conditions are violated, $\rho$ would correspond to a permissible blocking coalition.

\paragraph{Core with Transfers.}
Next, we consider the situation where individuals can commit to transferring money to their current partners, thereby providing incentives to consent to divorce.
Recall that $t_{m, \sigma(m)} \in \R$ denotes the transfer from $m$ to $\sigma(m)$ upon divorce. 
If $t_{m, \sigma(m)}$ is positive (negative), $m$ pays (receives) money to (from) his current partner $\sigma(m)$ after divorce. 
Moreover, we assume that transfers (if any) are made only between spouses who belong to the potentially blocking coalition. This implies that for any permissible coalition $S$, if $m$ or $\sigma(m)\notin S$, then $t_{m,\sigma(m)} = 0$.

Fix a set of consumption allocations and individual prices. 
Consider a path of remarriages $\rho= ( m_1, (m_2,\sigma(m_2)) \cdots, (m_{n-1}, \sigma(m_{n-1})), \sigma(m_n))$ and a set of transfers $t_{m_r, \sigma(m_r)}$ for all $1 \leq r \leq n$. 
If the following inequality holds for all $1 \leq r < n$, with at least one being strict,
$$
a_{m_r, \sigma(m_{r+1})} - t_{m_r, \sigma(m_r)} + t_{m_{r+1}, \sigma(m_{r+1})} \le 0,
$$
then $\rho$ corresponds to a blocking coalition. Further, if either $m_1=m_n$ or $(m_1,\sigma(m_1))$ and $(m_n,\sigma(m_n)) \notin \gV$, $\rho$ would correspond to a permissible blocking coalition. 
Summing the above inequality along the permissible path implies
$$
\sum_{r = 1}^{n-1} a_{m_r,\sigma(m_{r+1})} < 0.
$$
Note that the transfers cancel out because for every agent (who isn't first or last) in the path their partner is also in the path.\footnote{
    That is, if $-t_{m_r,\sigma(m_r)}$ appears in the sum corresponding to the potential pair $(m_r,\sigma(m_{r+1}))$, then $+t_{m_r,\sigma(m_r)}$ appears in the sum corresponding to the potential pair $(m_{r-1},\sigma(m_{r}))$.
} 
Moreover, since the coalition is permissible, we know that the incoming transfer to the first agent and the outgoing transfer from the last agent are equal to each other: if the path is a cycle then the first and last vertex agents are matched to each other, otherwise both agents are non-committed and their transfers are equal to zero.


Therefore, given a data set $\gD$, if the observed marriages are stable, there must exist feasible intrahousehold allocations and individual prices such that if the sum of edge weights along any path of remarriages $\rho$ is strictly negative, it must not be permissible. Definition \ref{def:PathMonotonicity} defines this condition formally as  \emph{path monotonicity}. Clearly, path monotonicity is a necessary condition for rationalizability with core with transfers. Theorem \ref{thm:PathMonotonicity} states that it is also sufficient.

\begin{defn}
\label{def:PathMonotonicity}
Given a data set $\gD$, $A(\gD)$ satisfies {\bf path monotonicity} if, {for all $m \in M$ and $w \in W$}, there are
$$
q^m_{m,\sigma(m)}, q^w_{\sigma(w),w} \in \R^k_+ \text{ and } P^m_{m,w}, P^w_{m,w} \in \R^K_{++}
$$
with
$$
q^m_{m,\sigma(m)} + q^{\sigma(m)}_{m,\sigma(m)} = q_{m,\sigma(m)} \text{ and }  P^m_{m,w} + P^w_{m,w} = P_{m,w},
$$
such that for every path of remarriages $\rho = (m_1, (m_2,\sigma(m_2)), \ldots, (m_{n-1},\sigma(m_{n-1})), \sigma(m_n))$ if 
$$
\sum_{r=1}^{n-1} a_{m_r,\sigma(m_{r+1})} < 0
$$
then,
\begin{itemize}
    \item [(i)] $(m_1,\sigma(m_1)) \in \gV$ or $(m_n,\sigma(m_n))\in \gV$, and
    \item [(ii)] $m_1\ne m_n$.
\end{itemize} 
\end{defn}

\begin{thm}
\label{thm:PathMonotonicity}
A data set $\gD$ is rationalizable as a core with transfers allocation if and only if $A(\gD)$ satisfies path monotonicity.
\end{thm}

\paragraph{Remarks.}
Before moving on to the empirical content of the rationalizability conditions, some remarks are in order. First, the conditions presented in Definitions \ref{def:PathConsistency} and \ref{def:PathMonotonicity} are the path counterparts of two well-known conditions in the revealed preference literature: cyclical consistency \citep{afriat1967construction} and cyclical monotonicity \citep{rockafellar1970convex}.
In particular, the path consistency and path monotonicity conditions collapse to their cyclical counterparts if $\gV = M\times W$ (i.e., when all couples committed).
Our second remark concerns the nested structure of the two sets of revealed preference conditions. If the observed behavior is consistent with the stability conditions with transfers, it would also be consistent with the stability conditions without transfers. This means that data consistency with path consistency is a necessary condition for consistency with path monotonicity.
Third, the results of \citealp{cherchye2017household} can be obtained as a corollary of Theorems \ref{thm:PathConsistency} and \ref{thm:PathMonotonicity} by setting $\gV=\emptyset$ (i.e., when all individuals can divorce unilaterally).
In that case, both path consistency and path monotonicity boil down to the requirement that there is no pair $(m,w)$ such that $a_{m,w}<0$.
%
Fourth, the stability conditions and corresponding revealed preference conditions can be trivially modified if one wants to limit the size of blocking coalitions (or the length of the path of remarriages). 
Finally, in Appendix \ref{app:LinearProgramming}, we show that the path consistency and path monotonicity conditions can be reformulated in terms of inequality constraints that are linear in unknowns. These linear inequality constraints are convenient from a practical point of view as they can be easily operationalized.

\subsection{Empirical Content}
We now focus on the empirical tractability of the models. In particular, we discuss the empirical bite of the revealed preference conditions to detect violations of marital stability and to identify the intrahousehold consumption allocations. 
As our aim is to determine the data setting under which marital stability loses all empirical power, in what follows, we concentrate on the more stringent path monotonicity condition. Due to the nested structure of these conditions, the same results will apply to the path consistency condition as well.

\paragraph{Rationalizability.}

So far, we assumed that the data set contains the matching function, household consumption for all matched couples, and prices and income for all potential pairs.
While it is easy to observe the matching function and household consumption of the matched couples in household surveys, prices and income faced by potential couples are typically unknown.
Therefore, empirical applications require making some assumptions about the prices and income that individuals would face in counterfactual matches. 
Existing studies assume that prices for private and public goods are the same across all potential matches and that household income is the sum of individual incomes, which are independent of their partner's income (see, e.g., \citealp[][]{cherchye2017household,cherchye2020marital,browning2021stable}).
We show that, under this assumption, violation of marital stability cannot be due to coalitions formed solely by committed couples.

To present the formal result, we define the notion of a \emph{blocking cycle}. 
A path of remarriages $\rho = (m_1, (m_2,\sigma(m_2)), \ldots, (m_{n-1},\sigma(m_{n-1})), \sigma(m_n))$ forms a blocking cycle if $m_1 = m_n$ and every pair is weakly blocking and at least one is blocking. 
When transfers between committed spouses are allowed, a blocking cycle corresponds to case $(ii)$ of Definition \ref{def:PathMonotonicity}.
Corollary \ref{cor:MutualConsentTransfersNoEmpiricalContent} shows that, if data does not contain price or income variation, then no violation of marital stability can be due to a blocking cycle of remarriages. 
Thus, if the observed household behavior is not rationalizable, it will be due to a presence of permissible blocking coalition involving at least one non-committed couple.

\begin{cor}
\label{cor:MutualConsentTransfersNoEmpiricalContent}
Suppose post-divorce transfers between partners are allowed. If, for every $m,m'\in M$ and $w,w'\in W$,
\begin{itemize}
    \item [(i)] $p_{m,w} = p_{m',w'}$,
    \item [(ii)] $P_{m,w} = P_{m',w'}$, and
    \item [(iii)] $y_{m,w} = y_m + y_w$,
\end{itemize}
then no data set $\gD$ can contain a blocking cycle.
\end{cor}

\paragraph{Identifiability.}
The revealed preference conditions can be useful in the identification of intrahousehold allocation.  
As the rationalizability conditions are linear in nature, they can be used to identify the unobserved parameters of household allocation (such as individual private consumption or Lindahl prices).
Such an identification is usually in the form of set identification where the identified set contains all the feasible values of the unobserved parameter that are consistent with the stability conditions. 
For example, we can identify female's private consumption by defining an upper (lower) bound by maximizing (minimizing) the linear function ($q^w_{\sigma(w),w}$) subject to the rationalizability conditions.
\citealp{cherchye2017household} have demonstrated that the stability conditions with $\gV=\emptyset$ can tightly identify household allocations for non-committed couples.
Here we discuss the identifying power of the stability conditions in more general cases.

As previously discussed, if there is variation in prices and/or income across potential matches, it is possible for a data set to violate the revealed preference conditions by generating blocking cycles. 
Interestingly, however, this assumption is not sufficient to ensure that the model has identifying power for all type of couples. 
Corollary \ref{cor:MutualConsentTransfersNoIdentification} shows that if prices of private goods are the same across outside options, then the stability conditions with transfers cannot identify intrahousehold allocations for committed couples.
In other words, we can assign any possible intrahousehold allocation for committed couples without violating the rationalizability restrictions.
The nested structure of the revealed preference conditions implies that whenever the stability conditions with transfers lacks identifying power, so does the stability conditions without transfers.

\begin{cor}
\label{cor:MutualConsentTransfersNoIdentification}
If  $p_{m,w} = p_{m',w'}$ for every $m,m'\in M$ and $w,w'\in W$ and a data set $\mathcal{D}$ is rationalizable as a core with transfers allocation, then the model has no identifying power for committed couples.
Equivalently, if there exist $q^m_{m,\sigma(m)}, q^{\sigma(m)}_{m,\sigma(m)} \in \R^k_+$ and $P^m_{m,w}, P^w_{m,w} \in \R^K_{++}$
with
$$
q^m_{m,\sigma(m)} + q^{\sigma(m)}_{m,\sigma(m)} = q_{m,\sigma(m)} \text{ and } P^m_{m,w} + P^w_{m,w} = P_{m,w},
$$
such that the data set is rationalizable as a core with transfers allocation, then for any committed couple $(i, \sigma(i)) \in \gV$ and any allocation $\bar{q}^i_{i,\sigma(i)}, \bar{q}^{\sigma(i)}_{i,\sigma(i)} \in \R^k_+$  with $\bar{q}^i_{i,\sigma(i)} + \bar{q}^{\sigma(i)}_{i,\sigma(i)} = q_{i,\sigma(i)}
$
there exist
$$
\bar q^m_{m,\sigma(m)}, \bar q^{\sigma(m)}_{m,\sigma(m)} \in \R^k_+ \text{ for all $m \neq i$ and } \bar P^m_{m,w}, \bar P^w_{m,w} \in \R^K_{++} \text{ for all $(m, w)$}
$$
with
$$
\bar q^m_{m,\sigma(m)} + \bar q^{\sigma(m)}_{m,\sigma(m)} = q_{m,\sigma(m)} \; \text{ and } \; \bar{P}^m_{m,w} + \bar{P}^w_{m,w} = P_{m,w},
$$
such that the data set is rationalizable as a core with transfers allocation.
\end{cor}

\noindent
Corollary \ref{cor:MutualConsentTransfersNoIdentification} indicates that using marital stability as an identifying assumption for committed couples necessitates price variation for private goods across potential matches.  In this context, we remark that price variation is typically required in revealed preference methods to gain identifying power (see, e.g., \citealp{varian1982nonparametric,beatty2011demanding,cherchye2015sharing}). Similarly, in the collective household literature, many traditional parametric methods also rely on price variation for identification (see, e.g., \citealp{browning1998efficient,chiappori2009microeconomics,browning2013estimating}).
In the context of marriage markets, identifying intrahousehold allocation for committed couples requires price variation across potential matches. In Section \ref{sec:empirical}, we reference empirical literature suggesting that such price variation is likely to be present in real-life situations.

\section{Simulation Analysis}
\label{sec:empirical}
We demonstrate the empirical performance of the stability conditions through a simulation exercise. 
Corollaries \ref{cor:MutualConsentTransfersNoEmpiricalContent} and \ref{cor:MutualConsentTransfersNoIdentification} outline the minimal data requirements for our model to be refutable or identifiable. 
However, the mere presence of these requirements does not guarantee that the model can be refuted or identified.
As there is no clear theoretical approach to show that these minimal requirements ensure empirical power for the stability conditions, we rely on a simulation exercise. This exercise allows us to illustrate realistic data settings in which marital stability can be tested and intrahousehold allocations can be identified for committed couples. 
In this simulation, we introduce random variation in the counterfactual prices and income that individuals face in their outside options. 
We consider the goodness-of-fit and the tightness of the bounds that the stability conditions recover for the within-household sharing pattern. 
We conclude that with small variations in prices, the path monotonicity condition generates very tight bounds.  

\subsection{Data and Setup}
\paragraph{Data.}Our analysis is based on a sample of households drawn from the LISS panel, a representative survey of households in the Netherlands conducted by CentERdata.\footnote{
While the Netherlands has a policy of allowing unilateral divorce, we choose to use this data set for the illustration as our main aim is to explore the empirical tractability of the rationalizability conditions. Furthermore, this data set has been utilized by numerous studies that employ the collective household model. In particular, \citealp{cherchye2017household} use the same data set to identify intrahousehold resource allocation in the context of unilateral divorce.}
The survey collects rich data on economics and sociodemographic variables at both individual and household levels. 
We consider a sample of 632 households comprising of 264 couples, 170 single males, and 198 single females. This sample is formed by individuals with or without children, working at least 10 hours per week in the labor market, and aged between 25 and 65. 
This data set was also studied by \citealp{cherchye2017household}; we refer to this paper for further details on the data and sample selection procedure.  

In terms of the set-up, we consider a labor supply setting in which a household spends its full potential income on leisure, private consumption, and public consumption. In the data, we observe both aggregated household expenditure as well as some assignable expenditure. Following \citealp{cherchye2017household}, we use all expenditure information to form a Hicksian good with price normalized to one.
We assume that the non-assignable expenditure is equally divided between public and private consumption. For leisure, we take the price to be the individual's hourly wage.

To deal with large sample size, we use subsampling to bring the stability conditions to data (similar to \citealp{browning2021stable}). We randomly draw 100 subsamples of 50 households from the original sample.\footnote{
    We have conducted robustness checks to assess the sensitivity of our results to the subsampling procedure by considering alternative subsample sizes of 70 and 100. Increasing the size of the subsamples leads to higher empirical content of the rationalizability constraints. This leads to lower stability indices and sharper upper and lower stable bounds (i.e., tighter identification). Our main qualitative conclusions remain intact; see Appendix \ref{app:simulation:robustness}.
}
For each subsample, we consider 11 different marriage market scenarios by varying the share of committed couples to take one of the following values $\{$0\%, 10\%, 20\%, $\cdots$ 100\%$\}$, each time assigning the commitment status of couples randomly. Then for each subsample-marriage market scenario, we apply the revealed preference methods separately. In what follows, we will report the summary results for these 100 subsamples.

We note that our subsamples include both committed and non-committed couples, as well as singles. Including singles allows for the possibility that married individuals consider a single person of the opposite gender as a potential mate, bringing our simulation setting closer to real-world marriage markets.\footnote{
    \citealp{cherchye2017household} also include singles in their sample.}
However, in this context, our theoretical result in Corollary \ref{cor:MutualConsentTransfersNoEmpiricalContent} has limited availability, as the data will always contain some non-committed couples, either directly or in the form of  singles.
We present a simulation exercise without singles in Appendix \ref{app:simulation:couples}, which allows us to demonstrate that the result of Corollary \ref{cor:MutualConsentTransfersNoEmpiricalContent} holds.

\paragraph{Simulation Set-up.}
Recall that while defining the rationalizability of a data set, we assumed that the data set contains a matching function, a set of committed couples,  household consumption for all matched couples, and price and income faced by all potential pairs. All of these, except prices and income faced by potential pairs, are present or can be inferred from the information recorded in typical household surveys. Existing empirical applications assume that prices are the same within and outside marriage and household incomes are the sum of individual incomes (see \citealp{cherchye2017household,cherchye2020marital,browning2021stable,bostyn2023time}). We examine the empirical content of our rationalizability conditions under two simulation scenarios that deviate from these assumptions. For each subsample-marriage market scenario, we consider two simulation set-ups with: (1) variation in prices and (2) variation in income. Below we outline how we introduce variation in prices and income across potential matches and briefly discuss why we might expect such variation to be present in reality.

{\it Price variation.} Recent empirical research has documented the existence of substantial price dispersion even for relatively homogeneous goods both across space (e.g., because of price differences across regions or differences across stores within a given geographical location) and across time within space (e.g., because of high-frequency sales) (see \citealp{aguiar2007life,griffith2009consumer,kaplan2015morphology,kaplan2019relative,nevo2019elasticity}). In our setting, both sources of price variation become relevant. If individuals relocate for the purpose of marriage, they may encounter different prices, for example, due to differences in housing costs or local market structure. Even within a local market, potential couples may experience different prices because they may shop at cheaper stores or at the same store but on different days.

We simulate variation in prices of private and public goods across potential matches through an additive random component. Specifically, we assume that prices for both private and public Hicksian consumption are equal to one in the current marriages and are defined as follows in potential marriages:
    \begin{equation*}
    \begin{split}
        p_{m,w} = 1 + \alpha \epsilon^p \; \text{where} \; \epsilon^p \sim U[-1,1] \; \text{and}  \; \alpha \in \lbrace 0\%, 1\%, 3\%, 5\% \rbrace, \\
        P_{m,w} = 1 +  \alpha \epsilon^P\; \text{where} \; \epsilon^P \sim U[-1,1] \; \text{and}  \; \alpha \in  \lbrace 0\%, 1\%, 3\%, 5\% \rbrace.
        \end{split}
    \end{equation*}
When $\alpha = 0\%$, there is no variation in prices across counterfactual matches, this is the base scenario. Increasing $\alpha$ increases the price variation across the outside options. For example, when $\alpha = 1\%$, prices in potential matches lie between 0.99 and 1.01, while $\alpha = 5\%$ implies prices ranging from 0.95 to 1.05. 

{\it Income variation.} 
 The existence of marital wage premium - the fact that married individuals earn more than their single counterparts - is a well-established empirical result (\citealp{korenman1991does}). It has been shown that this marital wage premium is not because of selection (that productive men are more likely to marry) but because marriage causes individual wages to rise (\citealp{antonovics2004all}). Additionally, recent evidence has shown that the marital wage premium depends on the partners' characteristics (\citealp{wang2013marriage,pilossoph2021household}). These findings provide support to our simulation set-up that household incomes may be different than the sum of individual members' current income depending on the marital status (if the exit option is becoming single) or the new partners' characteristics (if the exit option is marrying someone else).
    
We simulate variation in household incomes through a multiplicative random component. Specifically, we assume that household incomes in the current marriages are the sum of individual labor incomes, while household incomes in outside options are defined as follows:
    \begin{equation*}
        \begin{split}    
        y_{m,f} &= (y_m  + y_f) (1 + \alpha \epsilon^y) \; \text{where} \; \epsilon^y \sim U[-1,1] \; \text{and}  \;  \alpha \in \lbrace 0\%, 1\%, 3\%, 5\% \rbrace.
        \end{split}        
    \end{equation*}
The random component ($\alpha \epsilon^y$) can be interpreted as the gain/loss associated with the new match. When $\alpha = 0\%$, there is no gain or loss associated with the outside options, and the household income is the sum of individual incomes, that is the base scenario. Increasing $\alpha$ increases the random variation in income across counterfactual marriages. For instance, when $\alpha = 1\%$, there can be at most a 1\% gain or loss of household income in a new marriage, while when $\alpha = 5\%$, couples can experience at most a 5\% loss or gain in household total income.

\subsection{Estimation Procedure}
\paragraph{Goodness-of-fit.}
The stability conditions in Definitions \ref{def:PathConsistency} and \ref{def:PathMonotonicity} can be checked by simple linear programming techniques (see Appendix \ref{app:LinearProgramming}). This defines sharp tests for rationalizable consumption behavior: either the data satisfy the stability conditions or they do not. When the data do not satisfy the exact stability conditions, it is interesting to evaluate the degree of violation as the data may not be exactly rationalizable but close to satisfying the exact conditions.\footnote{
    In reality, household consumption behavior may not be exactly consistent with the model if, for example, the data contain measurement errors, there are frictions in the marriage market, or other factors, such as match quality, affect marital behavior. If the data are close to satisfying the exact conditions, we may want to include such data in the empirical analysis. The procedure we use can be applied to model such almost rationalizable behavior.}
To this end, we follow \citealp{cherchye2017household} by evaluating the goodness-of-fit of a model by introducing stability indices. These stability indices allow us to quantify the degree to which the data is consistent with the rationalizable behavior.

Formally, we include a stability index $s_{m,w}$ in each edge weight $a_{m,w}$ associated with the potential pair $(m,w)$. Specifically, we redefine $a_{m,w}$ as
$$
a_{m,w} = p_{m,w} ( q_{m,\sigma(m)}^m +  q_{\sigma(w),w}^w) + P^m_{m,w} Q_{m,\sigma(m)} + P^w_{m,w} Q_{\sigma(w),w} - y_{m,w} s_{m,w}
$$
and add the restriction $0 \leq s_{m,w} \leq 1$. Imposing $s_{m,w} = 1$ gives the original rationalizability restrictions, while imposing $s_{m,w}= 0$ would rationalize any data. 
Intuitively, the stability indices measure the loss of post-divorce income needed to represent the observed marriages as stable.
Generally, a lower stability index corresponds to a higher income loss associated with a particular outside option and can be interpreted as a greater violation of the underlying model assumptions. In other words, a higher stability index signals a better fit of the stability conditions. 

For a given data set $\mathcal{D}$, we compute the values of stability indices by computing
\begin{equation*}
    \max \sum_{m} \sum_{w} s_{m,w}
\end{equation*}
subject to the linear rationalizability conditions. Intuitively, maximizing this objective function boils down to introducing minimum adjustment in the data such that the conditions are satisfied. The solution to this optimization problem gives a stability index $s_{m,w}$ for each potentially blocking pair $(m,w)$.
If the original constraints are satisfied, then there is no need for adjustment, and all stability indices will be equal to one. 
Otherwise, a strictly smaller index will be required to rationalize the behavior.
Thus, the data is rationalizable if and only if the indices equal one. 

\paragraph{Identification.}
By using the computed stability indices, we can address identification of intrahousehold sharing patterns. Specifically, by multiplying the solution values of $s_{m,w}$ with the original income levels $y_{m,w}$, we obtain an adjusted dataset that is rationalizable by a stable matching. These stability indices define minimally adjusted data sets which then allows us to ``set identify'' the unknowns in the stability conditions (e.g., individual private consumption or personalized prices) from Definitions \ref{def:PathConsistency} and \ref{def:PathMonotonicity}. In practice, set identification requires computing upper and lower bounds that define an interval containing all parameter values that are consistent with the conditions.

In our simulation analysis, we focus on identifying the private consumption shares of the females. 
For any woman $w$, her private consumption share $q^{w}_{\sigma(w),w}/q_{\sigma(w),w}$ is linear in unknowns and thus we can define an upper (lower) bound by maximizing (minimizing) this objective function subject to the linear rationalizability conditions.
This effectively set identifies her private consumption share. 
The size of the interval indicates the identifying power of the stability conditions:  closer upper and lower bounds indicate a more informative identification.

\subsection{Estimation Results} 
\paragraph{Goodness-of-fit.} We begin by examining the extent to which the simulated data meet the rationalizability restrictions. 
Using the procedure outlined above, we compute the minimal adjustments required in the data to satisfy the rationalizability conditions.
These adjustments, quantified by stability indices, represent the minimal adjustments in post-divorce income needed to obtain consistency with the rationalizability conditions. As previously mentioned, the stability indices serve as indicators of the empirical content of the models. A higher stability index implies that the simulated data are closer to satisfying the exact conditions.

Figure \ref{fig_dcavg} presents the goodness-of-fit results. Sub-figures (a) and (b) correspond to the simulation scenarios with variations in prices and incomes, respectively.  
In each figure, plots with filled markers represent path consistency results, while plots with hollow markers represent path monotonicity results. The marker shapes indicate different $\alpha$ values, denoting the degree of price/income variation. 
Each plot shows the mean value of the average stability indices across 100 subsamples at a given share of committed couples in the market. 

\begin{figure}[htbp]
\caption{Goodness-of-fit}%
    \label{fig_dcavg}%
    \centering
    \subfloat[\centering Price variation]{{\includegraphics[scale=0.25]{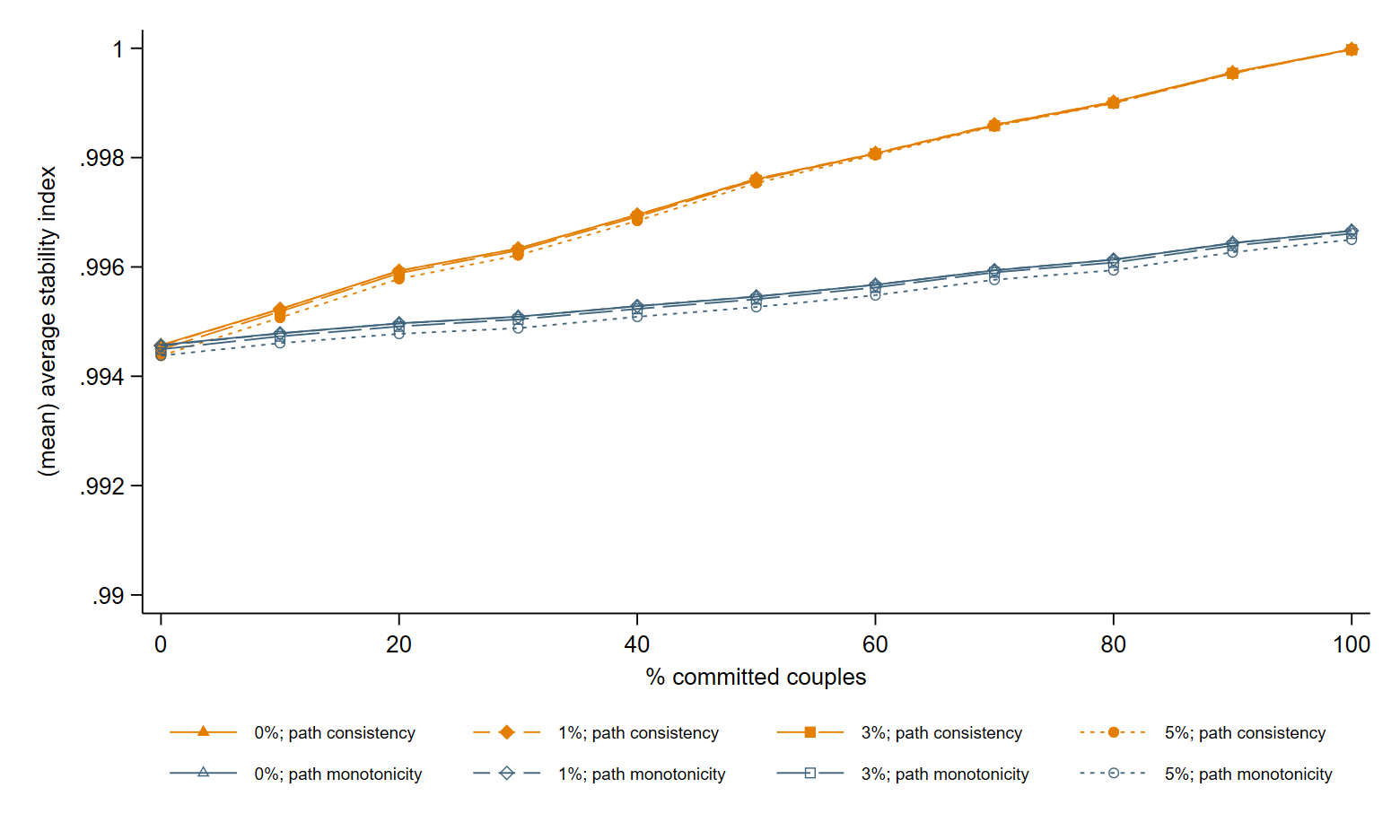} }}%
    \qquad
    \subfloat[\centering Income variation]{{\includegraphics[scale=0.25]{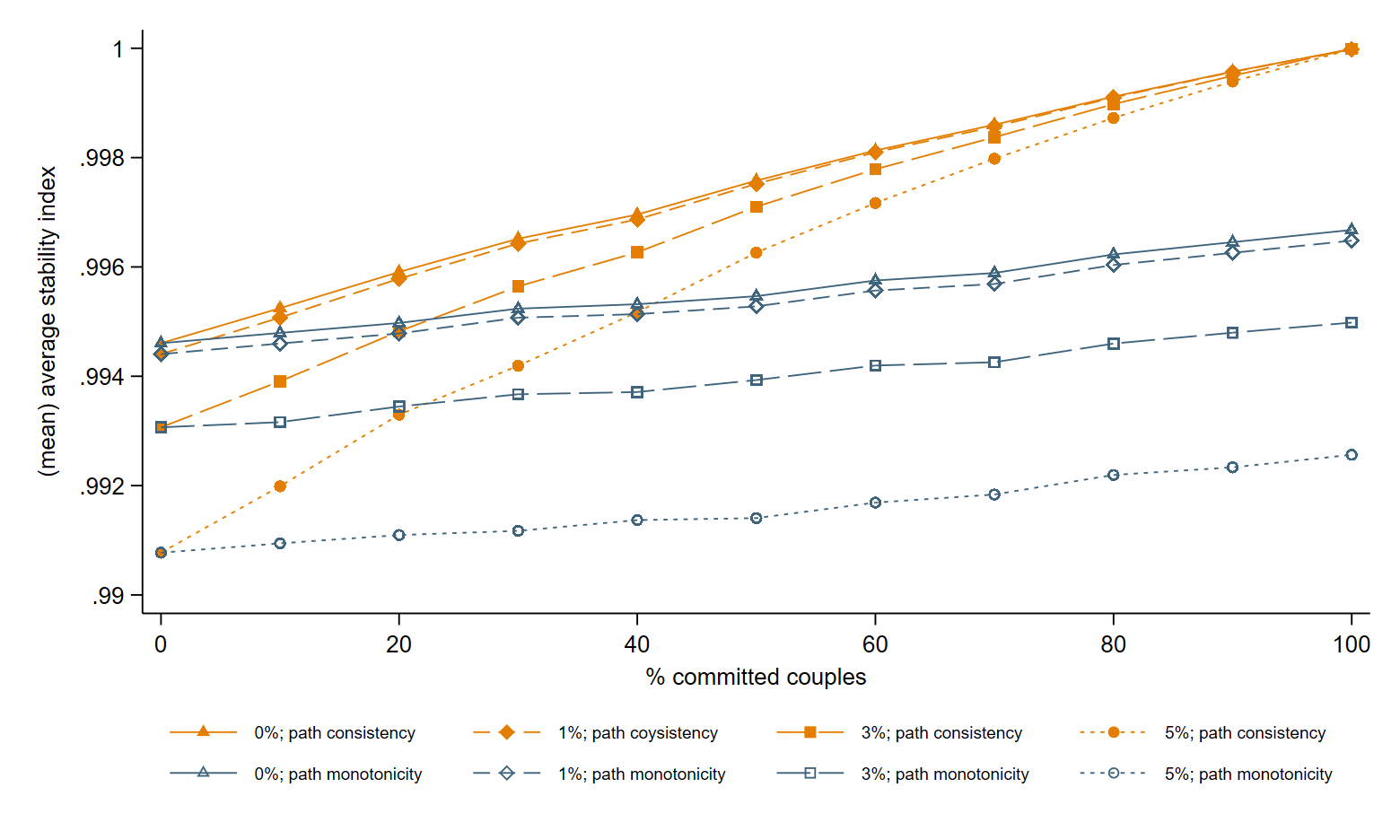} }}%
\end{figure}

We observe several interesting patterns. 
First, in most cases, the average stability index is strictly below one. This indicates that both path consistency and path monotonicity conditions exhibit empirical content, as the simulated data do not fully satisfy the exact condition. The only exception is when path consistency is imposed and all couples are committed. 
Second, there is substantial variation in the average stability indices across different marriage market and simulation scenarios. This suggests that the empirical power of the conditions may vary depending on the level of commitment and degree of price/income variation. We find that the empirical power of the conditions is higher when there is greater degree of price/income variation or smaller share of committed couples in the market. This indicates that the conditions lose some empirical bite as couples become committed.  
Third, given the nested structure of the models, the empirical power is generally higher  under path monotonicity than under path consistency. 
The testable implications of the two stability conditions coincide when none of the couples are committed, resulting in identical stability indices for both rationalizability conditions. 
This result is the consequence of the theoretical fact that if no couples are committed, then both path consistency and path monotonicity collapse to the result from \citealp{cherchye2017household}.

\paragraph{Identification.} 
We use the identified stability indices summarized above to construct new simulated data sets that are rationalizable as core allocations. These adjusted data allow us to set identify the intrahousehold sharing patterns that correspond to stable marriage behavior. We focus on committed couples and estimate the upper and lower bounds for the female share of private consumption in these households.

To demonstrate the identifying power of the stability conditions, we compare our estimated upper and lower bounds with so-called naive bounds.
These naive bounds are defined using the assignable information present in the data and do not make use of the theoretical restrictions associated with marital stability. They are defined as follows. For any woman, the naive lower bound equals the fraction of her assignable private consumption in the total private consumption. The naive upper bound adds the nonassignable private consumption share to the lower bound. Thus, the lower (upper) bound corresponds to the scenario when all nonassignable private consumption is allocated to the male (female). 

We compare the tightness of the ``stable'' bounds, which we obtain through our method, with  naive bounds. The tightness of the stable bounds ($\Delta^s$) are defined as the difference between our estimated upper and lower bounds. Similarly, the tightness of the naive bounds ($\Delta^n$) are defined as the difference between the upper and lower naive bounds.
We compute the relative difference between $\Delta^s$ and $\Delta^n$  (i.e., $\frac{\Delta^{n} - \Delta^{s}}{\Delta^{n}}$) which measures the extent to which our stable bounds are tighter than the naive bounds. In a sense, it quantifies the identifying power that follows from the stability assumptions. If our identification method yields point identification, the relative difference between the stable and naive bounds will be one. On the other hand, if the stability conditions have no identifying power, the relative difference between the two bounds will be zero. Thus, higher relative difference implies higher identification power of the stability restrictions. 

The results of this comparison are summarized in Figure \ref{fig_shares}. Similar to Figure \ref{fig_dcavg}, sub-figures (a) and (b) correspond to the simulation scenarios with variations in prices and income, respectively. Plots with filled markers represent path consistency results, while plots with hollow markers represent path monotonicity results. The marker shapes indicate different $\alpha$ values, denoting the degree of price/income variation. Each plot shows the mean value of the relative difference between the stable and naive bounds among the committed couples in the market across 100 subsamples. 

\begin{figure}[htbp]
    \centering
    \caption{Relative difference in bounds for committed couples}%
    \label{fig_shares}%
    \subfloat[\centering Price variation]{{\includegraphics[scale=0.25]{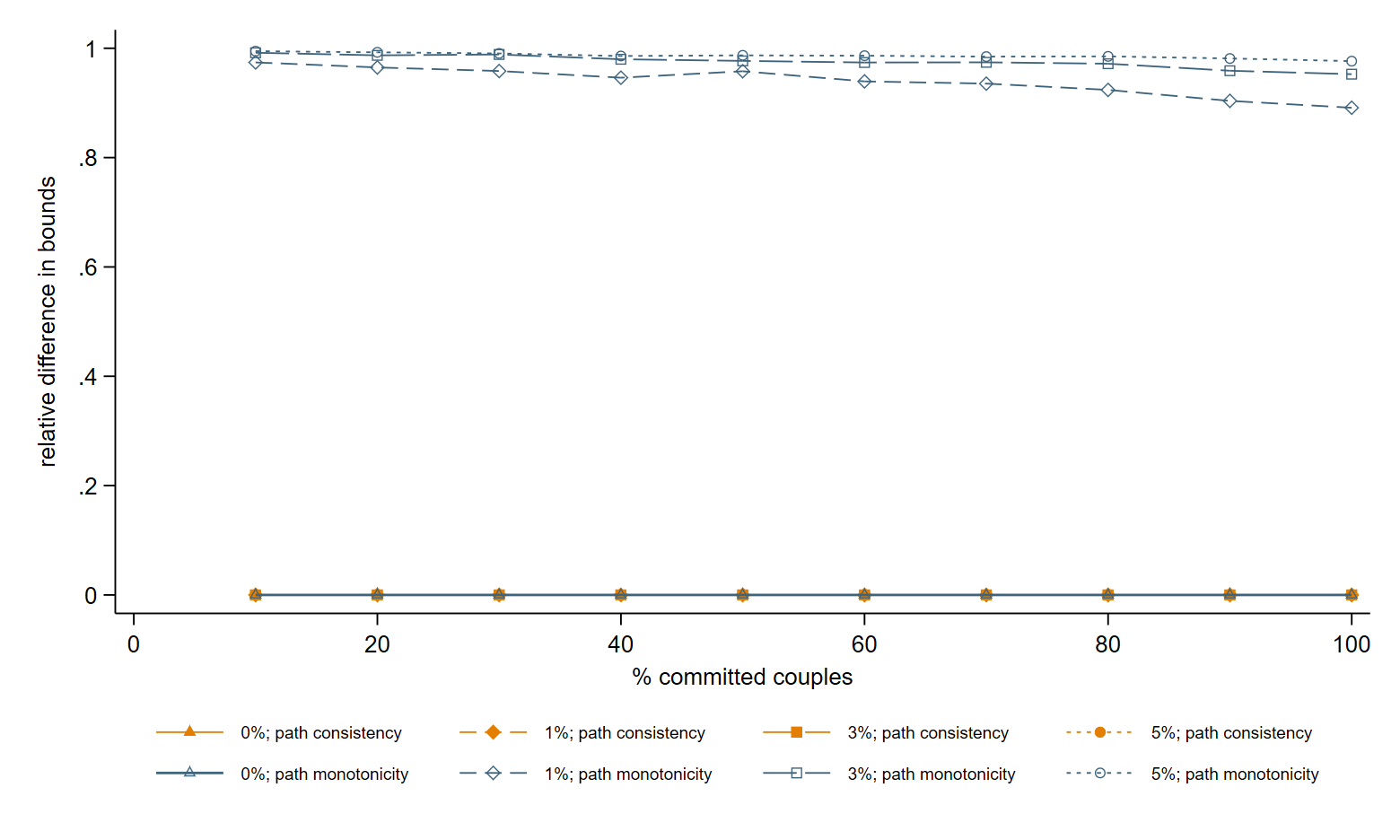} }}%
    \qquad
    \subfloat[\centering Income variation]{{\includegraphics[scale=0.25]{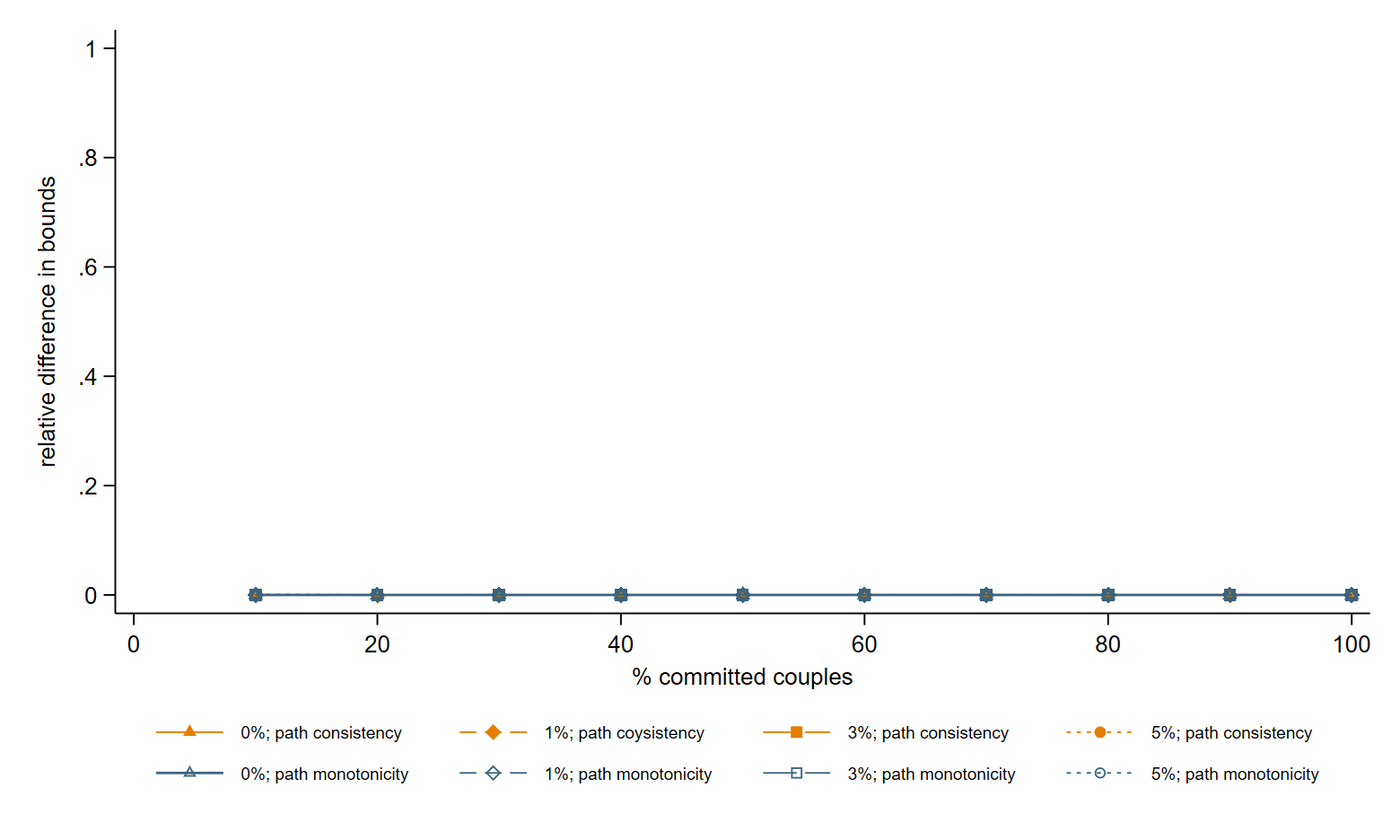} }}%
\end{figure}

From the two sub-figures, we learn that the path consistency condition lacks identifying power for committed couples, even when there is variation in prices or income faced by potential pairs. 
The same conclusion applies to the path monotonicity condition when there is no price variation for potential pairs. Sub-figure (b) shows that adding income variation does not enhance the identifying power of path monotonicity. This observation is in line with our theoretical result in Corollary 2 - if prices are the same across outside options and transfers are allowed, marital stability has no empirical bite for committed couples.\footnote{
    Figure \ref{fig_shares_noncommitted} in Appendix \ref{app:simulation:noncommitted} shows that the stable bounds for non-committed couples are very informative in all scenarios.} 
Interestingly, comparing Figures \ref{fig_dcavg} and \ref{fig_shares} reveals that model being refutable is necessary but not sufficient for identifying intrahousehold allocation. 
We observe several scenarios where the average stability index is strictly less than one, yet the stable bounds are as wide as the naive bounds.
On a more positive note, we find that path monotonicity exhibits substantial identifying power when there is price variation across marriages. Even a small degree of price variation ($\alpha = 1\%$) leads to almost point identification of female private consumption shares. Additionally, the tightness of the path monotonicity-based stable bounds increases with larger price variation or a higher share of non-committed couples in the market.

\paragraph{Summary.} 
Our simulation analysis yields two key conclusions. 
First, we demonstrate that both stability conditions are generally empirically significant, meaning not all observed behaviors can be explained as outcomes of stable marriages. 
The only situation where these conditions lack empirical content (non-refutable) is when the data consists solely of committed couples with no variation in price or income among the potential pairs. 
Second, while these conditions are quite powerful for identifying intra-household allocations among non-committed couples, they have limited ability to identify sharing patterns among committed couples. 
However, if the path monotonicity condition is applied and there is variation in the prices faced by potential pairs, it is possible to estimate informative bounds even for committed couples.

\section{Conclusion}
\label{sec:conclusion}
We presented a novel framework for analyzing rational household consumption under the assumption of marital stability with committed couples. 
The existing literature mostly considers a non-commitment framework, where individuals can divorce their partners at will. 
We extended the framework by generalizing the setting to allow some couple to divorce only with the consent of their current partner.
We began by showing that the core of the game is non-empty.
Next, we provided revealed preference characterization for cases when ex-partners are and aren't allowed to transfer money to each other upon divorce.
The stability conditions are linear in nature and can be used to (partially) identify the intrahousehold resource allocation. The key features of such identification are that it (i) does not make parametric assumptions about utility functions, (ii) allows for individual heterogeneity of preferences, and (iii) uses cross-sectional household data.
We presented a theoretical analysis of the empirical content (whether a data set can reject the model) and identifiability (whether intrahousehold consumption patterns can be identified) of the models. 
Along these lines, we identified the necessary conditions on data under which the models have empirical bite for rationalizability and identification of intrahousehold allocation for committed couples.
Finally, using a simulation exercise, we provided empirical support for our theoretical results.
We showed that in the presence of price variation, the testable implications of marital stability with transfers exhibit significant empirical bite for committed couples.

\appendix
\section{Proofs}
\label{app:Proofs}

\subsection{Proof of Proposition \ref{prop:NonEmptyCore}}

We begin by providing two supplementary lemmas.
Lemma \ref{lemma:Coalition2Path} shows that a permissible coalition can be represented as a permissible path or cycle of remarriages on a graph. 
Lemma \ref{lemma:NegativeNetTransfer} shows that in a permissible coalition with transfers, there exists at least one pair with non-positive net transfers.

For a given data set $\gD$, consider the graph representation of the marriage market and the notion of path of remarriages we defined in Section \ref{sec:results}. To recall, a path of remarriages $\rho = (m_1, (m_2,\sigma(m_2)), \cdots, (m_{n-1}, \sigma(m_{n-1})), \sigma(m_n))$ constitutes a set of agents $S = \{m_1,\ldots, m_{n-1}, \sigma(m_2),\ldots,\sigma(m_{n})\}$ and a matching $\hat\sigma: S\rightarrow S$ such that $\hat\sigma(m_j) = \sigma(m_{j+1})$ for every $1 \leq j < n$.
In the marriage market graph, $\rho$ is a path formed by a sequence of vertices $((m_1, \sigma(m_1)),(m_2,\sigma(m_2)),\ldots, (m_{n},\sigma(m_{n})))$ and a sequences of edges $((m_1, \sigma(m_2)),(m_2,\sigma(m_3)),\ldots, (m_{n-1},\sigma(m_{n})))$. 
The edges in this path specify who is remarrying whom. 
Note that $\rho$ forms a permissible coalition if either $m_n = m_1$ or $(m_1,\sigma(m_1))$ and $(m_n,\sigma(m_n)) \notin \gV$.
We show in Lemma \ref{lemma:Coalition2Path} that a permissible coalition implies existence of a permissible subcoalition that can be represented as a permissible path of remarriages.

\begin{lemma}
\label{lemma:Coalition2Path}
Let $S\subseteq M\cup W$ be a permissible coalition endowed with rematching $\hat \sigma$, then there is a path of remarriages, 
$$
\hat \rho = (m_1, (m_2,\sigma(m_2)), \ldots, (m_{n-1},\sigma(m_{n-1})), \sigma(m_n))
$$
such that $\hat\sigma(m_j) = \sigma(m_{j+1})$  for all $ 1 \leq j < n$, and either $m_n = m_1$ or $(m_1,\sigma(m_1))$ and $(m_n,\sigma(m_n))\notin \gV$.

\end{lemma}

\begin{proof}
We prove the Lemma by construction.
First, note that since $\hat \sigma$ is a matching, $\vert (M \cup \{\emptyset\}) \cap S\vert = \vert (W\cup \{\emptyset\}) \cap S\vert$.
Recall that $\lbrace \emptyset \rbrace$ corresponds to the set of virtual partners symbolizing the staying alone option.
Given that $S$ is a permissible coalition, an individual without their spouse in $S$ cannot be a committed couple.
In our construction, we consider two cases, (1) $\exists m\in S$ such that $\sigma(m)\notin S$ and (2) $\neg({\exists} m\in S$ such that $\sigma(m)\notin S$).
The first case would yield a path of remarriages (in a pure sense) and the second one would yield a cycle of remarriages. 
For both cases, we show an algorithm to construct $\hat \rho$.

\medskip
\noindent
\textbf{Case 1:} $\exists m\in S$ such that $\sigma(m)\notin S$

\noindent
\begin{itemize}
    \item [] Step 1:
    \begin{itemize}
        \item [] Let $m_1 = m\in S$ for some $m\in S$ such that $\sigma(m)\notin S$.\\
        (if there are multiple such $m$, pick one at random)
        
        \item [] Let $\sigma(m_2) = \hat \sigma(m_1)\in S$.\\
        ($\sigma(m_2)$ is always in $S$, otherwise $\vert M \cup \{\emptyset\}\cap S\vert \neq \vert W\cup \{\emptyset\} \cap S\vert$)
        
        \item [] Let $\hat \rho_1 = (m_1,\sigma(m_2))$.        
    \end{itemize}
    
    \item [] Step $i$ for $i \geq 2$:
    \begin{itemize}
        \item [] if $m_i\in S$, then
        \begin{itemize}
            \item [] Let $\sigma(m_{i+1}) = \hat \sigma(m_i)$
            \\
            ($\sigma(m_{i+1})$ is always in $S$, otherwise $\vert M \cup \{\emptyset\}\cap S\vert \neq \vert W\cup \{\emptyset\} \cap S\vert$)
            \item [] Let $\hat \rho_i = (\hat \rho_{i-1},m_i,\sigma(m_{i+1}))$.
        \end{itemize}
        \item [] else if $m_i\notin S$
        \begin{itemize}
            \item [] Let $\hat \rho = \hat \rho_{i-1}$.
            \item [] Let $n= i$.
        \end{itemize}
    \end{itemize}
\end{itemize}
The algorithm takes finite time (with a finite number of agents) and returns a sequence $\hat \rho$ that satisfies all the conditions for being a path of remarriages.
Note that because $m_n \notin S$ and $\sigma(m_1) \notin S$, it must be the case that $(m_1,\sigma(m_1))$ and $(m_n,\sigma(m_n))\notin \gV$.

\medskip
\noindent
\textbf{Case 2:} $\neg (\exists m\in S$ such that $\sigma(m)\notin S)$

\noindent
\begin{itemize}
    \item [] Step 1:
    \begin{itemize}
        \item [] Let $m_1 = m\in S$ for some $M$.
        \item [] Let $\sigma(m_2) = \hat \sigma(m_1)$.\\
        ($\sigma(m_{2}) \in S$, otherwise $\vert M \cup \{\emptyset\}\cap S\vert \neq \vert W\cup \{\emptyset\} \cap S\vert$)
        \item [] Let $\hat \rho_1 = (m_1,\sigma(m_2))$.       
        \end{itemize}
    \item [] Step $i$ for $i \geq 2$:
    \begin{itemize}
        \item [] if $m_i \neq m_1$, then
        \begin{itemize}
            \item [] Let $\sigma(m_{i+1}) = \hat \sigma(m_i)$
            \\
            ($\sigma(m_{i+1})$ is always in $S$, otherwise $\vert M \cup \{\emptyset\}\cap S\vert \neq \vert W\cup \{\emptyset\} \cap S\vert$)
            \item [] Let $\hat \rho_i = (\hat \rho_{i-1},m_i,\sigma(m_{i+1}))$.
        \end{itemize}
        \item [] else ($m_i = m_1$)
        \begin{itemize}
            \item [] Let $\hat \rho = \hat \rho_{i-1}$.
            \item [] Let $n= i$.
        \end{itemize}
    \end{itemize}
\end{itemize}
The algorithm takes finite time (with a finite number of agents) and returns a sequence $\hat \rho$ that satisfies all the conditions for being a path of remarriages. In particular, by construction we have $m_1 = m_n$.

\end{proof}

\noindent
Next, we consider a permissible coalition with transfers $(S,\hat\sigma,t)$. We show that if there exist at least one pair in the coalition which receives a positive net transfer, then there exists another pair that receives a negative net transfer. Recall that $t_{m,\sigma(m)}$ denotes the transfer that man $m$ commits to pay to his current match $\sigma(m)$ upon divorce. 
For a (potential) pair $(m,w) \in S$, the budget constraint is defined as 
$$
 p_{m,w} (q^m + q^w) + P_{m,w}Q \le  y_{m,w}  - t_{m,\sigma(m)}+t_{\sigma(w),w},
$$
and $- t_{m,\sigma(m)}+t_{\sigma(w),w}$ is the net transfers for the pair.

\begin{lemma}
\label{lemma:NegativeNetTransfer}
Let $(S,\hat\sigma,t)$ be a permissible coalition with transfers, if there is $(m,w)\in S\times S$ such that $\hat \sigma (m) = w$ and $-t_{m,\sigma(m)}+t_{\sigma(w),w)} > 0$, then there is $(m',w')\in S\times S$ such that $\hat \sigma (m') = w'$ and $-t_{m',\sigma(m')}+t_{\sigma(w'),w'} < 0$.
\end{lemma}

\begin{proof}
Given Lemma \ref{lemma:Coalition2Path} and $(S,\hat\sigma,t)$ we can focus on a path of remarriages $\rho = (m_1, (m_2,\sigma(m_2)), \ldots, (m_{n-1},\sigma(m_{n-1})), \sigma(m_n))
$ such that $\hat\sigma(m_j) = \sigma(m_{j+1})$  for all $1 \leq j < n$, and either $m_n = m_1$ or $(m_1,\sigma(m_1)),(m_n,\sigma(m_n))\notin \gV$.
To simplify the notation, let us denote $t_{m_i,\sigma(m_i)}$ as $t_i$.
We prove the Lemma by contradiction. Suppose there is a pair $t_{i+1}-t_i>0$ and no pair such that $t_{j+1}-t_j<0$.

\medskip
\noindent
\textbf{Case 1:} $(m_1,\sigma(m_1))$ and $(m_n,\sigma(m_n))\notin \gV$.

\noindent
Since both $(m_1,\sigma(m_1))$ and $(m_n,\sigma(m_n))$ are non-committed couples, we know that $t_1 = t_{n} = 0$. Given that $t_1\ge 0$ and $t_{j+1}-t_j\ge 0$ for all $j \le n-1$, it must be the case $t_{j+1} \ge t_j$.
Moreover, since there is a $i$ such that $t_{i+1} > t_i$, it must be that $t_n > t_1 \ge 0$. This is a contradiction as $t_n = 0$.

\medskip
\noindent
\textbf{Case 2:} $m_n = m_1$

\noindent
Since $m_1 = m_n$, we know that $t_1 = t_n$. Given that $t_{j+1}-t_j\ge 0$ for all $j \le n-1$, it must be the case $t_{j+1} \ge t_j$.
Moreover, since there is a $i$ such that $t_{i+1} > t_i$, it must be that $t_n > t_1$. This is a contradiction as $t_1 = t_n$.

\end{proof}

\paragraph{Proof of Proposition \ref{prop:NonEmptyCore}}

\begin{proof}
For a couple $(m,w)$, the household problem is given by
\begin{align*}
    \max_{q^m, q^w, Q} u^m(q^m, Q) &+ \mu u^w(q^w, Q) \text{ such that } \\
    p_{m,w}(q^m, q^w) &+ P_{m,w}Q \leq y_{m,w}
\end{align*}
where $\mu>0$ is the Pareto weight for $w$. The Pareto frontier of the couple's decision problem is continuous and strictly decreasing (see \citealp{cherchye2017household} for proof).
Building on this result, for a marriage game $\Gamma^{\emptyset} = (M, W, \emptyset, B, U)$ with no committed couples in the market (i.e., when $\gV = \emptyset$), a general result of \citealp{alkan1990core} shows that the core is nonempty. Therefore, we know that $C(\Gamma^{\emptyset}) \ne \emptyset$.
Let 
$
\Gamma = (M, W, \gV, B, U),
$
with $\gV \ne \emptyset$, be a game that is different from $\Gamma^{\emptyset}$ only by requiring some couples to be committed.
To complete the proof, we need to show that $C(\Gamma^{\emptyset}) \subseteq C_T (\Gamma) \subseteq C (\Gamma)$.

\medskip
\noindent
{\bf $C_T (\Gamma) \subseteq C (\Gamma)$.} 

\noindent
Suppose that there exist an allocation $\alpha$ such that $\alpha \in C_T(\Gamma)$ but $\alpha \notin C(\Gamma)$.
Then, there exists a permissible coalition $(S,\hat\sigma)$ that blocks $\alpha$.
Letting $t_{m,\sigma(m)}=0$ for every $m\in S$, we obtain a blocking coalition with transfers $(S,\hat\sigma,t)$ that blocks $\alpha$. This implies that $\alpha \notin C_T(\Gamma)$. 

\medskip
\noindent
{\bf $C(\Gamma^{\emptyset}) \subseteq C_T (\Gamma) $.}

\noindent
Consider an allocation $\alpha$ such that $\alpha \in C(\Gamma^{\emptyset})$ and assume that $\alpha \notin C_T(\Gamma)$.
Then there is a blocking coalition with transfers $(S,\hat\sigma,t)$.
By Lemma \ref{lemma:NegativeNetTransfer}, we know that either $t_{m,\sigma(m)}=0$ for every $m\in S$ or this coalition contains a pair $(m,w)\in S \times S$ such that $m=\hat\sigma(w)$ and $t_{\sigma(w),w}-t_{m,\sigma(m)}<0$.
In the first case, all rematched couples in the coalition are weakly blocking and at least one is blocking without any transfer of money. This immediately implies that $\alpha \notin C(\Gamma^{\emptyset})$, a contradiction.
In the second case, consider the pair $(m,w)$ with negative net transfer. Given that this pair is weakly blocking with reduced income after transfers and utilities are monotone then, without transfer, the couple's income will strictly increase allowing both $m$ and $w$ to obtain strictly higher utilities. 
Hence, if $\gV = \emptyset$, the pair $(m,w)$ will be a blocking pair. 

\end{proof}

\subsection{Proof of Theorem \ref{thm:PathConsistency}}

\begin{proof}
{\bf Necessity.}
Assume that the data set is rationalizable as a core allocation.
Then, for any permissible path of remarriages $\rho = (m_1, (m_2,\sigma(m_2))$, $\ldots$, $(m_{n-1},\sigma(m_{n-1})), \sigma(m_n))$, the corresponding coalition cannot be blocking. This means that either $(i)$ none of the pairs in $\rho$ are blocking or $(ii)$ there is at least one pair in $\rho$ which is not weakly blocking.

Consider the first case and any potential pair in the coalition, denoted as $(m,w)$. For this couple, consider a Pareto efficient allocation $(q^m_{m,w}, q^w_{m,w}, Q_{m,w}) \in B_{m,w}$ such that the man $m$ is indifferent between the consumption bundle $(q^m_{m,w}, Q_{m,w})$ and the bundle $(q^m_{m,\sigma(m)}, Q_{m,w})$ in his current marriage. The slope of $m$'s indifference curve at $(q^m_{m,w}, Q_{m,w})$ is given by the price vectors $(p_{m,w}, P^m_{m,w})$ where $P^m_{m,w}$ represents man $m$'s willingness-to-pay for the public goods. Given convex preferences, a revealed preference argument implies that $m$'s indifference curve between $(q^m_{m,w}, Q_{m,w})$ and $(q^m_{m,\sigma(m)}, Q_{m,\sigma(m)})$ must satisfy
\begin{equation}
    p_{m,w}q^m_{m,w} + P^m_{m,w}Q_{m,w} \leq p_{m,w}q^m_{m,\sigma(m)} + P^m_{m,w}Q_{m,\sigma(m)},
\end{equation}
as the hyperplane with slope $(p_{m,w}, P^m_{m,w})$  through $(q^m_{m,w}, Q_{m,w})$  must be situated below the bundle $(q^m_{m,\sigma(m)}, Q_{m,\sigma(m)})$.

Next, consider the indifference curve of woman $w$ at the allocation $(q^m_{m,w}, q^w_{m,w}, Q_{m,w})$. The slope of $w$'s indifference curve at $(q^w_{m,w}, Q_{m,w})$ is given by the price vectors $(p_{m,w}, P^w_{m,w})$, where $P^w_{m,w}$ is her willingness-to-pay for the public consumption. By Pareto efficiency, we have $P^m_{m,w} + P^w_{m,w} = P_{m,w}$. For $(m,w)$ to not be a blocking pair, we must have that woman $w$ prefers the bundle $(q^w_{\sigma(w),w}, Q_{\sigma(w),w})$ to  $(q^w_{m,w}, Q_{m,w})$. Then, $w$'s indifference curve through $(q^w_{m,w}, Q_{m,w})$ must lie below the indifference curve through $(q^w_{\sigma(w),w}, Q_{\sigma(w),w})$. This implies
\begin{equation}
    p_{m,w}q^w_{m,w} + P^w_{m,w}Q_{m,w} \leq p_{m,w}q^w_{\sigma(w),w} + P^w_{m,w}Q_{\sigma(w),w}.
\end{equation}

Combining inequalities (1) and (2) and using the budget constraint, we get
\begin{equation*}
    y_{m,w} \leq p_{m,w}(q^m_{m,\sigma(m)} + q^w_{\sigma(w),w}) + P^m_{m,w}Q_{m,\sigma(m)} + P^w_{m,w}Q_{\sigma(w),w},    
\end{equation*}
which simplifies to 
\begin{equation*}
   a_{m,w} \geq 0. 
\end{equation*}

Consider the second scenario and let $(m,w)$ denote the potential pair which is not weakly blocking. We apply the same argument as before for man $m$ and derive inequality (1). Since $(m,w)$ is not weakly blocking, it must be that woman $w$ strictly prefers the bundle $(q^w_{\sigma(w),w}, Q_{\sigma(w),w})$ over $(q^w_{m,w}, Q_{m,w})$. Then, inequality (2) holds as strict inequality. We have that 
\begin{equation*}
    y_{m,w} < p_{m,w}(q^m_{m,\sigma(m)} + q^w_{\sigma(w),w}) + P^m_{m,w}Q_{m,\sigma(m)} + P^w_{m,w}Q_{\sigma(w),w},    
\end{equation*}
which implies
\begin{equation*}
   a_{m,w} > 0. 
\end{equation*}

We have shown that if the data is rationalizable, we can find individual consumption bundles and price vectors such that for any permissible path $\rho = (m_1, (m_2,\sigma(m_2))$, $\ldots$, $(m_{n-1},\sigma(m_{n-1}), \sigma(m_n))$ either $a_{m_k, \sigma(m_{k+1})} \geq 0$ for all $1 \leq k < n$, or there exists a $1 \leq k < n$, such that $a_{m_k, \sigma (m_{k+1}) } > 0$. Therefore, path consistency is a necessary condition for the rationalizability of the data as a core allocation.

\medskip
\noindent
{\bf Sufficiency.}
Suppose $A(\mathcal{D})$ satisfies path consistency. 
For a vector $x$, let $[x]_j$ denote its $j$-th component.
Given the Lindahl prices and individual consumption values that satisfy path consistency, let $o$ and $O$ be such that,
\begin{align*}
O > \max_{m\in M, w\in W, j \le k, J\le K} \lbrace [p_{m,w}]_j; [P^m_{m,w}]_J; [P^w_{m,w}]_J \rbrace & \text{ and } \\
o < \min_{m\in M, w\in W, j \le k, J\le K} \lbrace [p_{m,w}]_j; [P^m_{m,w}]_J; [P^w_{m,w}]_J \rbrace. &
\end{align*}
Next, define a piece-wise linear function to construct individual utilities.
Let $v$ be defined as
\begin{equation*}
    v(x) = 
    \begin{cases}
    Ox & \text{ if } x\le 0, \\
    ox & \text{ if } x > 0.
    \end{cases}
\end{equation*}
Following \citealp{cherchye2017household,browning2021stable}, for any individual $i\in M\cup W$, consider the utility function
$$
u^i(q,Q) = \sum_{j=1}^k v([q]_j - [q^i_{i,\sigma(i)}]_j) + \sum_{J=1}^K v([Q]_J - [Q_{i,\sigma(i)}]_J).
$$
For this specification of utilities, we have
\begin{equation*}
     u^m(q^m_{m,\sigma(m)}, Q_{m,\sigma(m)}) = 0 \quad \text{and} \quad u^w(q^w_{\sigma(w),w}, Q_{\sigma(w),w}) = 0.
\end{equation*}

\noindent
Assume that, for these utility functions, the dataset is not rationalizable as a core allocation. This means that there is a permissible coalition $\rho = (m_1, (m_2,\sigma(m_2))$, $\ldots$, $(m_{n-1},\sigma(m_{n-1})), \sigma(m_n))$ such that every pair is weakly blocking and at least one pair is blocking. Consider any potential pair in the path and denote the couple as $(m,w)$. Since $(m,w)$ is weakly blocking, there exists an allocation $(q^m, q^w, Q) \in B_{m,w}$ such that 
\begin{equation*}
    \begin{split}
        u^m(q^m, Q) & \geq u^m(q^m_{m,\sigma(m)}, Q_{m,\sigma(m)}) = 0, \\
        u^w(q^w, Q) & \geq u^w(q^w_{\sigma(w),w}, Q_{\sigma(w),w}) = 0.
    \end{split}
\end{equation*}
\\
For $m$, if $[q^m]_j > [q^m_{m,\sigma(m)}]_j$ (or $[Q]_J > [Q_{m,\sigma(m)}]_J$) then by construction of $o$ we know that $ o([q^m]_j - [q^m_{m,\sigma(m)}]_j) < p_{m,w} ([q^m]_j - [q^m_{m,\sigma(m)}]_j)$ (or, $o ([Q]_J - [Q_{m,\sigma(m)}]_J) < P^m_{m,w} ([Q]_J - [Q_{m,\sigma(m)}]_J)$). On the other hand, if $[q^m]_j \leq [q^m_{m,\sigma(m)}]_j$ (or, $[Q]_J \leq [Q_{m,\sigma(m)}]_J$), then by construction of $O$ we know that $ O([q^m]_j - [q^m_{m,\sigma(m)}]_j) < p_{m,w} ([q^m]_j - [q^m_{m,\sigma(m)}]_j)$ (or, $O([Q]_J - [Q_{m,\sigma(m)}]_J) < P^m_{m,w} ([Q]_J - [Q_{m,\sigma(m)}]_J)$). Summing these inequalities across all goods, we get:
\begin{equation}
p_{m,w} (q^m_{m,w}-q^m_{m,\sigma(m)}) + P^m_{m,w}(Q-Q_{m,\sigma(m)}) \geq 0.    
\end{equation}
Similarly, for $w$, we obtain:
\begin{equation}
p_{m,w} (q^w_{m,w}-q^w_{\sigma(w),w}) + P^w_{m,w}(Q-Q_{\sigma(w),w}) \geq 0.    
\end{equation}
Adding inequalities (3) and (4) for $m$ and $w$, we obtain
$$
p_{m,w} (q^m_{m,w}-q^m_{m,\sigma(m)}) + P^m_{m,w}(Q-Q_{m,\sigma(m)}) + p_{m,w} (q^w_{m,w}-q^w_{\sigma(w),w}) + P^w_{m,w}(Q-Q_{\sigma(w),w}) \geq 0.
$$
Given that the Lindahl prices add up to the market price of the public good ($P^m_{m,w} + P^w_{m,w} = P_{m,w}$) and that $p_{m,w} q^m_{m,w} + P_{m,w} Q + p_{m,w} q^w_{m,w} = y_{m,w}$, we can simplify the inequality above to obtain
$$
y_{m,w} - (p_{m,w} q^m_{m,\sigma(m)} + P^m_{m,w}Q_{m,\sigma(m)} + p_{m,w} q^w_{\sigma(w),w} + P^w_{m,w} Q_{\sigma(w),w}) \geq 0.
$$
Thus,
$$
a_{m,w} \leq 0.
$$
\\
Now, consider the blocking pair in the path and denote the pair as $(m,w)$. This implies that there exists an allocation $(q^m, q^w, Q) \in B_{m,w}$ such that 
\begin{equation*}
    \begin{split}
        u^m(q^m, Q) & \geq u^m(q^m_{m,\sigma(m)}, Q_{m,\sigma(m)}) = 0, \\
        u^w(q^w, Q) & \geq u^w(q^w_{\sigma(w),w}, Q_{\sigma(w),w}) = 0.
    \end{split}
\end{equation*}
with at least one strict inequality. Applying the same arguments as before, we derive that inequalities (3) and (4) hold, with at least one being strict. Adding these two inequalities yields:
$$
a_{m,w} < 0.
$$
This obtains a violation of path consistency.

\end{proof}

\subsection{Proof of Theorem \ref{thm:PathMonotonicity}}

\begin{proof}
{\bf Necessity.}
Consider a data set $\gD$ that is rationalizable as a core with transfers allocation. For any  permissible path of remarriages $\rho = (m_1, (m_2,\sigma(m_2))$, $\cdots$, $(m_{n-1},\sigma(m_{n-1})), \sigma(m_n))$ and transfers $t_{m_k, \sigma(m_k)}$ for $ 1 \leq k \leq n$, the corresponding coalition with transfers cannot be blocking. This requires that either (i) none of the pairs in $\rho$ are blocking, or (ii) there exist at least one pair in $\rho$ which is not weakly blocking.

Consider the first scenario and any potential pair in the coalition. For simplicity, denote the pair as $(m,w)$. For this couple, consider a Pareto efficient allocation $(q^m_{m,w}, q^w_{m,w}, Q_{m,w})$ in the budget set such that the man $m$ is indifferent between the consumption bundle $(q^m_{m,w}, Q_{m,w})$ and $(q^m_{m,\sigma(m)}, Q_{m,w})$ in his current marriage. The slope of $m$'s indifference curve at $(q^m_{m,w}, Q_{m,w})$ is given by the price vectors $(p_{m,w}, P^m_{m,w})$ where $P^m_{m,w}$ represents man $m$'s willingness-to-pay for the public goods. Given convex preferences, a revealed preference argument implies that $m$'s indifference curve between $(q^m_{m,w}, Q_{m,w})$ and $(q^m_{m,\sigma(m)}, Q_{m,\sigma(m)})$ must satisfy
\begin{equation}
    p_{m,w}q^m_{m,w} + P^m_{m,w}Q_{m,w} \leq p_{m,w}q^m_{m,\sigma(m)} + P^m_{m,w}Q_{m,\sigma(m)},
\end{equation}
because the hyperplane with slope $(p_{m,w}, P^m_{m,w})$  through $(q^m_{m,w}, Q_{m,w})$ must be situated below the bundle $(q^m_{m,\sigma(m)}, Q_{m,\sigma(m)})$.

Next, consider the indifference curve of woman $w$ at the allocation $(q^m_{m,w}, q^w_{m,w}, Q_{m,w})$. The slope of $w$'s indifference curve at $(q^w_{m,w}, Q_{m,w})$ is $(p_{m,w}, P^w_{m,w})$, where $P^w_{m,w}$ is her willingness-to-pay for the public consumption. By Pareto efficiency, we have $P^m_{m,w} + P^w_{m,w} = P_{m,w}$. For $(m,w)$ to not be a blocking pair, we must have that woman $w$ prefers the bundle $(q^w_{\sigma(w),w}, Q_{\sigma(w),w})$ over $(q^w_{m,w}, Q_{m,w})$. Then, $w$'s indifference curve $(q^w_{m,w}, Q_{m,w})$ must lie below the indifference curve through $(q^w_{\sigma(w),w}, Q_{\sigma(w),w})$. This implies
\begin{equation}
    p_{m,w}q^w_{m,w} + P^w_{m,w}Q_{m,w} \leq p_{m,w}q^w_{\sigma(w),w} + P^w_{m,w}Q_{\sigma(w),w}.
\end{equation}

Combining inequalities (5) and (6) and using the budget constraint, we get
\begin{equation*}
    y_{m,w} - t_{m, \sigma(m)} + t_{\sigma(w),w} \leq p_{m,w}(q^m_{m,\sigma(m)} + q^w_{\sigma(w),w}) + P^m_{m,w}Q_{m,\sigma(m)} + P^w_{m,w}Q_{\sigma(w),w},    
\end{equation*}
or 
\begin{equation}
   a_{m,w} + t_{m, \sigma(m)} - t_{\sigma(w),w} \geq 0 \text{ for all } (m,w). 
\end{equation}

Summing inequality (7) over the permissible path, we obtain
$$
\sum_{j=1}^{n-1} a_{m_j, \sigma(m_{j+1})} + t_{m_1, \sigma(m_1)} - t_{m_n, \sigma(m_n)} \geq 0 .
$$
If $(m_1,\sigma(m_1))$ and $(m_n,\sigma(m_n))\notin \gV$, then $t_{m_1, \sigma(m_1)} = t_{m_n, \sigma(m_n)} = 0$  since there are no transfers going to and coming from individuals who do not belong to the coalition. Otherwise, if $m_1 = m_n$, then $t_{m_1, \sigma(m_1)} = t_{m_n, \sigma(m_n)}$. Thus, 
$$
\sum_{j=1}^{n-1} a_{m_j, \sigma(m_{j+1})} \geq 0 .
$$

Let us now consider the second scenario. We begin by showing that for any set of transfers, the sum of net-transfers over all pairs in a permissible path is zero. For a permissible path $\rho = (m_1, (m_2,\sigma(m_2))$, $\ldots$, $(m_{n-1},\sigma(m_{n-1})), \sigma(m_n))$, the sum of net transfers over all pairs in the path is given by
$$
\sum_{j=1}^{n-1} - t_{m_j, \sigma(m_{j})} + t_{m_{j+1}, \sigma(m_{j+1})}, 
$$
which simplifies to 
$$
- t_{m_1, \sigma(m_{1})} + t_{m_{n}, \sigma(m_{n})}.
$$
As $\rho$ is permissible, either $(m_1,\sigma(m_1))$ and $(m_n,\sigma(m_n))\notin \gV$ or $m_1 = m_n$. In the first case, $t_{m_1, \sigma(m_1)} = t_{m_n, \sigma(m_n)} = 0$, because no transfers are made to or from individuals outside the coalition. In the second case, $t_{m_1, \sigma(m_1)} = t_{m_n, \sigma(m_n)}$. This proves our claim.

Next, consider a permissible path of remarriages $\rho = (m_1, (m_2,\sigma(m_2))$, $\ldots$, $(m_{n-1},\sigma(m_{n-1})), \sigma(m_n))$ with transfer $t_{m_k, \sigma(m_k)}$ for $ 1 \leq k \leq n$ such that there exist at least one pair in $\rho$ which is not weakly blocking. 
If none of the other pairs in $\rho$ are blocking, we can apply the same reasoning as in the first scenario.  Therefore, assume that some pairs in $\rho$ are blocking and at least one pair is not weakly blocking. Suppose we reduce the net transfers of the blocking pairs in the path such that they become weakly blocking pairs.
That is, we adjust the net transfer such that all previously blocking pairs are indifferent between remarrying and staying with their current partners.\footnote{This follows from the fact that household Pareto frontiers are continuous and increasing in income.} 
If we redistribute the collected sum of money among the pairs which are not weakly blocking, we claim that this sum will be insufficient to make all non-weakly blocking pairs weakly blocking and at least one pair blocking. This is because, if such transfers did exist then the corresponding coalition would block the allocation and hence, the allocation cannot belong to the core with transfers. Therefore, we can find an adjusted set of transfers $\hat t_{m_k, \sigma(m_k)}$ for $ 1 \leq k \leq n$ such that none of the pairs are blocking. This brings us back to scenario 1, which shows that 
$$
\sum_{j=1}^{n-1} a_{m_j, \sigma(m_{j+1})} \geq 0.
$$

We have shown that if the data is rationalizable, we can find individual consumption bundles and price vectors such that for any permissible path $\rho = (m_1, (m_2,\sigma(m_2))$, $\ldots$, $(m_{n-1},\sigma(m_{n-1})), \sigma(m_n))$,  $\sum_{k=1}^{n-1} a_{m_k, \sigma(m_{k+1})} \geq 0$. Thus, path monotonicity is a necessary requirement for data rationalizability as a core with transfers allocation.

\medskip
\noindent
{\bf Sufficiency.}
Suppose that $A(\mathcal{D})$ satisfies path monotonicity. 
Given the Lindahl prices and individual consumption values that satisfy path monotonicity, we construct individual utility functions using the same approach as in the proof of Theorem \ref{thm:PathConsistency}. For this specification of utilities, recall that we have
\begin{equation*}
     u^m(q^m_{m,\sigma(m)}, Q_{m,\sigma(m)}) = 0, \; u^w(q^w_{\sigma(w),w}, Q_{\sigma(w),w}) = 0
\end{equation*}

Now, suppose that for these utility functions, the data set cannot be rationalized as a core with transfers allocation. This means there exists a permissible coalition $\rho = (m_1, (m_2,\sigma(m_2))$, $\ldots$, $(m_{n-1},\sigma(m_{n-1})), \sigma(m_n))$ and transfers $t_{m_k, \sigma(m_k)}$ for $ 1 \leq k \leq n$ such that every pair in the coalition is weakly blocking and at least one pair is blocking. Consider any potential pair in the path and denote this pair as $(m,w)$. As $(m,w)$ is weakly blocking, this implies that there exists an allocation $(q^m, q^w, Q) \in B_{m,w}$ such that 
\begin{equation*}
    \begin{split}
        u^m(q^m, Q) & \geq u^m(q^m_{m,\sigma(m)}, Q_{m,\sigma(m)}) = 0, \\
        u^w(q^w, Q) & \geq u^w(q^w_{\sigma(w),w}, Q_{\sigma(w),w}) = 0.
    \end{split}
\end{equation*}
\\
Repeating the arguments from the proof in Theorem \ref{thm:PathConsistency}, we obtain the following inequalities,
\begin{equation*}
\begin{split}
p_{m,w} (q^m_{m,w}-q^m_{m,\sigma(m)}) + P^m_{m,w}(Q-Q_{m,\sigma(m)}) \geq 0, \\
p_{m,w} (q^w_{m,w}-q^w_{\sigma(w),w}) + P^w_{m,w}(Q-Q_{\sigma(w),w}) \geq 0.   
\end{split}
\end{equation*}
Combining the inequalities for $m$ and $w$, we obtain
$$
p_{m,w} (q^m_{m,w}-q^m_{m,\sigma(m)}) + P^m_{m,w}(Q-Q_{m,\sigma(m)}) + p_{m,w} (q^w_{m,w}-q^w_{\sigma(w),w}) + P^w_{m,w}(Q-Q_{\sigma(w),w}) \geq 0.
$$
Given that the Lindahl prices add up to the market price of the public good ($P^m_{m,w} + P^w_{m,w} = P_{m,w}$) and that $p_{m,w} q^m_{m,w} + P_{m,w} Q + p_{m,w} q^w_{m,w} = y_{m,w} - t_{m,\sigma(m)} + t_{\sigma(w),w}$, we can simplify the inequality above to obtain
$$
y_{m,w} - t_{m,\sigma(m)} + t_{\sigma(w),w} - (p_{m,w} q^m_{m,\sigma(m)} + P^m_{m,w}Q_{m,\sigma(m)} + p_{m,w} q^w_{\sigma(w),w} + P^w_{m,w} Q_{\sigma(w),w}) \geq 0,
$$
which simplifies to
\begin{equation}
    a_{m,w} + t_{m,\sigma(m)} - t_{\sigma(w),w} \leq 0 \text{ for all } (m,w).    
\end{equation}
\\
Next, consider the blocking pair in the path and denote the pair as $(m,w)$. This implies that there exists an allocation $(q^m, q^w, Q) \in B_{m,w}$ such that 
\begin{equation*}
    \begin{split}
        u^m(q^m, Q) & \geq u^m(q^m_{m,\sigma(m)}, Q_{m,\sigma(m)}) = 0, \\
        u^w(q^w, Q) & \geq u^w(q^w_{\sigma(w),w}, Q_{\sigma(w),w}) = 0.
    \end{split}
\end{equation*}
with at least one strict inequality. Repeating the same reasoning as above, we obtain
\begin{equation}
    a_{m,w} + t_{m,\sigma(m)} - t_{\sigma(w),w} < 0.    
\end{equation}
Summing inequalities (8) and (9) over all pairs along the permissible path $\rho$ yields
$$
\sum_{j=1}^{n-1} a_{m_j, \sigma(m_{j+1})} + t_{m_1, \sigma(m_{1})} - t_{m_{n}, \sigma(m_{n})} < 0 .
$$
If $(m_1,\sigma(m_1))$ and $(m_n,\sigma(m_n))\notin \gV$, then $t_{m_1, \sigma(m_{1})} = t_{m_{n}, \sigma(m_{n})} = 0$  since there are no transfers to and from individuals outside the coalition. Otherwise, if $m_1 = m_n$, then $t_{m_1, \sigma(m_{1})} = t_{m_{n}, \sigma(m_{n})}$. Thus, 
$$
\sum_{j=1}^{n-1} a_{m_j, \sigma(m_{j+1})} < 0 .
$$
This results in a violation of path monotonicty.
\end{proof}

\subsection{Proof of Corollary \ref{cor:MutualConsentTransfersNoEmpiricalContent}}
\begin{proof}
Let $p_{m,w} = p$ and $P_{m,w}=P$ as they are assumed to be the same across all couples.
Recall that we have freedom in determining the Lindahl prices.
Let $P^m_{m,w} = P^{m}$ for all $m\in M$ and $P^w = P-P^m$ for every $w\in W$.
On the contrary, assume that there is a blocking cycle $S$.
Consider the cycle of remarriages and focus on the potential pairs formed by $m$ and $\sigma(m)$ for some $m \in S$. Let us denote the rematches as $(m, w')$ and  $(m', \sigma(m))$.
\begin{align*}
\sum_{m \in S} a_{m, \hat \sigma(m)} &= \sum_{m* \ne m, m'} a_{m*, \hat \sigma(m*)} + a_{m, w'} + a_{m', \sigma(m)} \\
    &= \sum_{m* \ne m, m'} a_{m*, \hat \sigma(m*)} + \\
    & \quad \quad p (q^m_{m,\sigma(m)} + q^{w'}_{\sigma(w'),w'} ) + P^m Q_{m,\sigma(m)} + P^w Q_{\sigma(w'),w'} -  y_{m,w'}  + \\ 
    & \quad \quad p (q^{m'}_{m',\sigma(m')} + q^{\sigma(m)}_{m,\sigma(m)} ) + P^m Q_{m',\sigma(m')} + P^w Q_{m,\sigma(m)} -  y_{m',\sigma(m)}
\end{align*}
Recall that $y_{m,w} = y_m+y_w$, then we can further rearrange the terms
\begin{align*}
\sum_{m \in S} a_{m, \hat \sigma(m)} &= \sum_{m* \ne m, m'} a_{m*, \hat \sigma(m*)} + \\
    & \quad \quad  p (q^m_{m,\sigma(m)} + q^{w'}_{\sigma(w'),w'} ) + P^m Q_{m,\sigma(m)} + P^w Q_{\sigma(w'),w'} -  y_{m} - y_{w'} + 
    \\ 
    & \quad \quad  p (q^{m'}_{m',\sigma(m')} + q^{\sigma(m)}_{m,\sigma(m)} ) + P^m Q_{m',\sigma(m')} + P^w Q_{m,\sigma(m)} -  y_{m'} - y_{\sigma(m)} \\ 
    &= \sum_{m* \ne m, m'} a_{m*, \hat \sigma(m*)} + \\
    & \quad \quad  p (q^m_{m,\sigma(m)} + q^{\sigma(m)}_{m,\sigma(m)}  ) + P^m Q_{m,\sigma(m)} + P^w Q_{m,\sigma(m)}  -  y_{m} -  y_{\sigma(m)} + 
    \\ 
    & \quad \quad  p (q^{m'}_{m',\sigma(m')} + q^{w'}_{\sigma(w'),w'} ) + P^m Q_{m',\sigma(m')} +  P^w Q_{\sigma(w'),w'} -  y_{m'} - y_{w'}  
    \\ 
    &= \sum_{m* \ne m, m'} a_{m*, \hat \sigma(m*)} + \\
    & \quad \quad  p (q^m_{m,\sigma(m)} + q^{\sigma(m)}_{m,\sigma(m)}  ) + PQ_{m,\sigma(m)} -   y_{m,\sigma(m)} + 
    \\ 
    & \quad \quad  p (q^{m'}_{m',\sigma(m')} + q^{w'}_{\sigma(w'),w'} ) + P^m Q_{m',\sigma(m')} +  P^w Q_{\sigma(w'),w'} -  y_{m'} - y_{w'}  
     \\ 
    &= \sum_{m* \ne m, m'} a_{m*, \hat \sigma(m*)} + \\
    & \quad \quad  p (q^{m'}_{m',\sigma(m')} + q^{w'}_{\sigma(w'),w'} ) + P^m Q_{m',\sigma(m')} +  P^w Q_{\sigma(w'),w'} -  y_{m'} - y_{w'} 
\end{align*}

\noindent
Thus, we have removed the terms corresponding to the observed couple $(m,\sigma(m))$ from the cycle of remarriages.
Repeating the same procedure for other observed matches in $S$, we can obtain that the sum of edge weights along the cycle equals zero for every cycle of remarriages.
Thus, path monotonicity along the cycle of remarriages is trivially satisfied, implying that the data cannot contain the blocking cycle.

\end{proof}

\subsection{Proof of Corollary \ref{cor:MutualConsentTransfersNoIdentification}}

\begin{proof}

Consider a data set $\gD$ which is rationalizable with core with transfers. Suppose $p_{m,w} = p$ for every $(m,w ) \in M\cup \{\emptyset\} \times W\cup \{\emptyset\}$. 
As the data set is rationalizable, there are individual consumption bundles $(q^m_{m,\sigma(m)},q^{\sigma(m)}_{m,\sigma(m)})$  and Lindahl prices  $(P^m_{m,w},P^{w}_{m,w})$ such that there is no blocking path of remarriages.
Consider a committed couple $(m,\sigma(m))$. 
Suppose there is another allocation ($\hat q^m_{m,\sigma(m)}$, $\hat q^{\sigma(m)}_{m,\sigma(m)}$) which makes the matching unstable.
Then, according to Theorem \ref{thm:PathMonotonicity}, there is a violation of path monotonicity over a path of remarriages $\rho$ that includes $m$ and $\sigma(m)$. Consider this path of remarriages and focus on the potential pairs formed by $m$ and $\sigma(m)$. Let us denote the rematches as $(m, w')$ and  $(m', \sigma(m))$. Violation of path monotonicity over this path implies,
\begin{align*}
    0 &> \sum_{m \in S} a_{m, \hat \sigma(m)} 
    \\
    &> \sum_{m* \ne m, m'} a_{m*,w*} + \\
    & \quad \quad p (\hat q^m_{m,\sigma(m)} + q^{w'}_{\sigma(w'),w'} ) + P^m_{m,w'}Q_{m,\sigma(m)} + P^{w'}_{m,w'} Q_{\sigma(w'),w'} - y_{m,w'} + 
    \\ 
    & \quad \quad p (q^{m'}_{m',\sigma(m')} + \hat q^{\sigma(m)}_{m,\sigma(m)} ) + P^{m'}_{m',\sigma(m)}Q_{m',\sigma(m')} + P^{\sigma(m)}_{m',\sigma(m)} Q_{m,\sigma(m)} -  y_{m',\sigma(m)} 
    \\
    &> \sum_{m* \ne m, m'} a_{m*,w*} + \\
    & \quad \quad p(\hat q^m_{m,\sigma(m)} + \hat q^{\sigma(m)}_{m,\sigma(m)}) + P^m_{m,w'}Q_{m,\sigma(m)} + P^{\sigma(m)}_{m',\sigma(m)} Q_{m,\sigma(m)}  - y_{m,w'} + 
    \\
    & \quad \quad p(q^{m'}_{m',\sigma(m')} + q^{w'}_{\sigma(w'),w'}) +  P^{m'}_{m',\sigma(m)}Q_{m',\sigma(m')} + P^{w'}_{m,w'} Q_{\sigma(w'),w'} - y_{m',\sigma(m)}  
    \\
    &> \sum_{m* \ne m, m'} a_{m*,w*} + \\
    & \quad \quad p q_{m,\sigma(m)} + P^m_{m,w'}Q_{m,\sigma(m)} + P^{\sigma(m)}_{m',\sigma(m)} Q_{m,\sigma(m)}  - y_{m,w'} + 
    \\
    & \quad \quad p(q^{m'}_{m',\sigma(m')} + q^{w'}_{\sigma(w'),w'}) +  P^{m'}_{m',\sigma(m)}Q_{m',\sigma(m')} + P^{w'}_{m,w'} Q_{\sigma(w'),w'} - y_{m',\sigma(m)}  
\end{align*}
\noindent
At the same time, rationalizability of the data with $q^m_{m,\sigma(m)}$ and $q^{\sigma(m)}_{m,\sigma(m)}$ implies that
\begin{align*}
    0 &\leq \sum_{m \in S} a_{m, \hat \sigma(m)} 
    \\
    &\leq \sum_{m* \ne m, m'} a_{m*,w*} + \\
    & \quad \quad p (q^m_{m,\sigma(m)} + q^{w'}_{\sigma(w'),w'} ) + P^m_{m,w'}Q_{m,\sigma(m)} + P^{w'}_{m,w'} Q_{\sigma(w'),w'} - y_{m,w'} + 
    \\ 
    & \quad \quad p (q^{m'}_{m',\sigma(m')} + q^{\sigma(m)}_{m,\sigma(m)} ) + P^{m'}_{m',\sigma(m)}Q_{m',\sigma(m')} + P^{\sigma(m)}_{m',\sigma(m)} Q_{m,\sigma(m)} -  y_{m',\sigma(m)} 
    \\
    &\leq \sum_{m* \ne m, m'} a_{m*,w*} + \\
    & \quad \quad p(q^m_{m,\sigma(m)} + q^{\sigma(m)}_{m,\sigma(m)}) + P^m_{m,w'}Q_{m,\sigma(m)} + P^{\sigma(m)}_{m',\sigma(m)} Q_{m,\sigma(m)}  - y_{m,w'} + 
    \\
    & \quad \quad p(q^{m'}_{m',\sigma(m')} + q^{w'}_{\sigma(w'),w'}) +  P^{m'}_{m',\sigma(m)}Q_{m',\sigma(m')} + P^{w'}_{m,w'} Q_{\sigma(w'),w'} - y_{m',\sigma(m)}  
    \\
    &\leq \sum_{m* \ne m, m'} a_{m*,w*} + \\
    & \quad \quad p q_{m,\sigma(m)} + P^m_{m,w'}Q_{m,\sigma(m)} + P^{\sigma(m)}_{m',\sigma(m)} Q_{m,\sigma(m)}  - y_{m,w'} + 
    \\
    & \quad \quad p(q^{m'}_{m',\sigma(m')} + q^{w'}_{\sigma(w'),w'}) +  P^{m'}_{m',\sigma(m)}Q_{m',\sigma(m')} + P^{w'}_{m,w'} Q_{\sigma(w'),w'} - y_{m',\sigma(m)}  
\end{align*}

\noindent
That is a contradiction.
Thus, the data set is also rationalizable with the allocation $\hat q^m_{m,\sigma(m)}$ and $\hat q^{\sigma(m)}_{m,\sigma(m)}$ for $(m, \sigma(m))$ and  $\hat q^m_{m,\sigma(m)} = q^m_{m,\sigma(m)}$ and $\hat q^{\sigma(m)}_{m,\sigma(m)} = q^{\sigma(m)}_{m,\sigma(m)}$ for all other couples. We can repeat this logic for the remaining committed couples in the path.

\end{proof}

\section{Practical Implementation}
\label{app:LinearProgramming}


\subsection{Path Consistency}

The path consistency condition in Definition \ref{def:PathConsistency} can be reformulated in terms of linear inequalities characterized by binary integer variables. 

\begin{prop}
\label{prop:pc}
Given a data set $\gD$, $A(\gD)$ satisfies path consistency if and only if there exist vectors $q^m_{m,\sigma(m)}, q^{\sigma(m)}_{m,\sigma(m)} \in \R^k_+$ and numbers $z_m \in \R$ for all $m \in M$ and vectors $P^m_{m,w}, P^w_{m,w} \in \R^K_{++}$, numbers $a_{m,w} \in \R$ and integer variables $\delta_{m,w}, \xi_{m,w} \in\{0,1\}$ for all $(m,w) \in (M \cup \lbrace \emptyset \rbrace) \times (W \cup \lbrace \emptyset \rbrace)$ such that
\begin{equation*}
\label{eq:PC}
\begin{aligned} 
\mysubnumber\quad & q^m_{m,\sigma(m)}+ q^{\sigma(m)}_{m,\sigma(m)}  = q_{m,\sigma(m)}, \\
\mysubnumber\quad & P^m_{m,w} + P^w_{m,w}  = P_{m,w}, \\
\mysubnumber\quad & a_{m,w} = p_{m,w} q^m_{m,\sigma(m)} + P^m_{m,w} Q_{m,\sigma(m)} + p_{m,w} q^w_{\sigma(w),w}  + P^w_{m,w} Q_{\sigma(w),w} - y_{m,w}, \\
\mysubnumber\quad & a_{m,w} + \delta_{m,w} M \ge 0, \\
\mysubnumber\quad & a_{m,w} - (1 - \delta_{m,w})M < 0, \\
\mysubnumber\quad & a_{m,w} + \xi_{m,w} M > 0, \\
\mysubnumber\quad & a_{m,w} - (1 - \xi_{m,w})M \le 0, \\
\mysubnumber\quad & v_m z_m  <  v_{\sigma(w)} z_{\sigma(w)} + (1-\delta_{m,w})M, \\ 
\mysubnumber\quad & v_m z_m  \le  v_{\sigma(w)} z_{\sigma(w)} + (1-\xi_{m,w})M,
\end{aligned}
\end{equation*}
where $v_m = 1 \text{ if } (m,\sigma(m)) \in \gV$ and $0$ otherwise and  $M$ is a fixed number greater than any possible values $a_{m,w}$.
\end{prop}

\noindent
The numbers $v_m$ are indicators that take the value 1 if the couple $(m,\sigma(m))$ is committed and 0 otherwise. The integer variables $\delta_{m,w}$ are indicator variables that equal 1 if and only if $a_{m,w} < 0$ and 0 otherwise. Similarly, the variables $\xi_{m,w}$ are indicators that take the value 1 if and only if $a_{m,w} \leq 0$ and 0 otherwise.
Equations $(i)-(iii)$ set the stage by balancing private consumption, Lindahl prices, and defining the edge weights.
Equations $(iv)-(vii)$ define the values of $\delta_{m,w}$ and $\xi_{m,w}$ by check whether $a_{m,w}$ is positive or strictly positive, respectively. 
Equations $(viii)-(ix)$ generate numbers $z_m$ that implement path consistency.\footnote{
 We note that the strict inequalities in conditions $(v)$, $(vi)$ and $(viii)$ of  Proposition \ref{prop:pc} are not desirable from numerical point of view. To address this, we need to replace the strict inequalities with weak ones by adding a small positive $\epsilon$, ensuring that the conditions are still feasible (see \citealp{nobibon2016revealed}). Additionally, our formulation employs the Big-M method which uses large constants to enforce conditions $(iv)-(ix)$ under specific scenarios. We note that choosing an appropriate value for $M$ is crucial because the Big-M formulation is known to suffer from numerical issues (\citealp{schrijver1998theory}).
} 
To prove Proposition \ref{prop:pc}, we first present the following lemma.

\begin{lemma}
\label{lemma:PathConsistency2MultiplierConsistency}
Given a data set $\gD$, if $A(\gD)$ satisfies path consistency, then there are numbers $z_m$ and $\mu_{m,w}>0$ for every $m\in M$ and $w\in W$ such that 
\begin{equation*}
\begin{cases}
z_{m}-z_{\sigma(w)} \le  \mu_{m,w} a_{m,w} 
& \text{ if } (m,\sigma(m)), (\sigma(w),w)\in\mathcal V, \\ 
z_m \le  \mu_{m,w} a_{m,w} 
& \text{ if } (m,\sigma(m)) \in \mathcal{V}, (\sigma(w),w)\notin\mathcal V, \\ 
- z_{\sigma(w)} \le  \mu_{m,w} a_{m,w} 
& \text{ if } (m,\sigma(m))\notin \mathcal V, (\sigma(w),w)\in\mathcal V.
\end{cases}
\end{equation*}
\end{lemma}
\begin{proof}
Let $o>0$ be a sufficiently small number such that
$$
o < \min \left\{1, \min_{m\in M; w\in W} \{ |a_{m,w}| \} \right\},
$$
and $O>0$ be a sufficiently large number such that
$$
O > \frac{\sum_{m\in M, w\in W: a_{m,w}<0} |a_{m,w}|}{o}
$$
Then, let
\begin{equation*}
\mu_{m,w} = 
\begin{cases}
o & \text{ if } a_{m,w}\le 0, \\
O & \text{ if } a_{m,w} > 0.
\end{cases}
\end{equation*}
We can show that for every permissible path $\rho$, 
$$
\sum_{(m,w)\in \rho} a_{m,w} \mu_{m,w} \ge 0.
$$
The reasoning can be divided into two cases:
First, if the path contains strictly negative $a_{m,w}$, then it should also include at least one strictly positive element. The weight $\mu_{m,w}$ for positive $a_{m,w}$
is chosen to be large enough to outweigh the sum of possible negative values.
Second, if the path contains only zero edge weights, then the sum will be zero by construction.

Next, given the weights defined above, let
$$
z_m = \min\{\mu_{m,w} a_{m,w} + \mu_{\sigma(w),w'} a_{\sigma(w),w'} + \ldots \}
$$
where the minimization is over (i) all sequences starting with $m\in M$ if $\gV = M\times W$; or (ii) all sequences starting with $m\in M$ and ending with a non-committed couple if  $\gV \subset M\times W$. 
With this construction, we show that the inequalities in Lemma \ref{lemma:PathConsistency2MultiplierConsistency} are satisfied.

\medskip
\noindent
\textbf{Case 1: $(m,\sigma(m)), (\sigma(w),w)\in\mathcal V$.}
\\
Since $z_{\sigma(w)}$ is defined as a weighted sum
$$
z_{\sigma(w)} = \mu_{\sigma(w),w'} a_{\sigma(w),w'} + \ldots, 
$$
over a sequence that starts with $\sigma(w)$, then the following sum
$$
\mu_{m,w} a_{m,w} + \mu_{\sigma(w),w'} a_{\sigma(w),w'} + \ldots
$$
is over a sequence that starts with $m$ and follows the sequence formed by $z_{\sigma(w)}$. Then, by construction of $z_m$, we know that 
$$
z_m \le \mu_{m,w} a_{m,w} + \mu_{\sigma(w),w'} a_{\sigma(w),w'} + \ldots = \mu_{m,w}a_{m,w} + z_{\sigma(w)}.
$$

\medskip
\noindent
\textbf{Case 2: $(m,\sigma(m)) \in \mathcal{V}, (\sigma(w),w)\notin\mathcal V$.}
\\
In this case, $z_m \le  \mu_{m,w} a_{m,w}$ by construction, because $(\sigma(w),w)$ is a non-committed couple.

\medskip
\noindent
\textbf{Case 3: $(m,\sigma(m)) \notin \mathcal{V}, (\sigma(w),w)\in\mathcal V$.}
\\
Here, showing $- z_{\sigma(w)} \le  \mu_{m,w} a_{m,w}$
is equivalent to showing $0 \le \mu_{m,w} a_{m,w} + z_{\sigma(w)}$. 
By construction of $z_{\sigma(w)}$, we know that any sequence starting from $m$ and following the sequence formed by $z_{\sigma(w)}$, would start and finish with non-committed couples. Thus, it forms a permissible path. 
Then, by construction of the weights $\mu_{m,w}$, we have
$$
\mu_{m,w} a_{m,w} + \mu_{\sigma(w),w'} a_{\sigma(w),w'} + \ldots \geq 0,
$$
which obtains $- z_{\sigma(w)} \le  \mu_{m,w} a_{m,w}$.

\end{proof}

\paragraph{Proof of Proposition \ref{prop:pc}}
\begin{proof}

{\bf Necessity.}
Assume that the path consistency condition from Definition \ref{def:PathConsistency} holds. Let us use the same solution to define $q^m_{m,\sigma(m)}, q^{\sigma(m)}_{m,\sigma(m)}$, $P^m_{m,w}, P^w_{m,w}$ and $a_{m,w}$. The first three conditions are satisfied by default. Let $\delta_{m,w} = 1$ if $a_{m,w} < 0$ and zero otherwise, and let $\xi_{m,w} = 1$ if $a_{m,w} \leq 0$ and zero otherwise. Conditions $(iv)-(vii)$ are then satisfied by construction. We now need to construct numbers $z_m$ for all $m$ such that conditions $(viii)-(ix)$ of Proposition \ref{prop:pc} are satisfied.

By Lemma \ref{lemma:PathConsistency2MultiplierConsistency}, there exist numbers $z_m$ and $\mu_{m,w}>0$ for every $m\in M$ and $w\in W$ such that 
\begin{equation*}
\begin{cases}
z_{m}-z_{\sigma(w)} \le  \mu_{m,w} a_{m,w} & \text{ if } (m,\sigma(m)), (\sigma(w),w)\in\mathcal V, \\ 
z_m \le  \mu_{m,w} a_{m,w} & \text{ if } (m,\sigma(m)) \in \mathcal{V}, (\sigma(w),w)\notin\mathcal V, \\ 
- z_{\sigma(w)} \le \mu_{m,w} a_{m,w} & \text{ if } (m,\sigma(m))\notin \mathcal V, (\sigma(w),w)\in\mathcal V.
\end{cases}
\end{equation*}
We will show that these constructed $z_m$ values satisfy conditions $(viii)-(ix)$ of Proposition \ref{prop:pc}. First note that if $\delta_{m,w}$ and $\xi_{m,w} = 0$, conditions $(viii)-(ix)$ are satisfied by default.
Consider inequality $(viii)$ which is relevant only if $\delta_{m,w} = 1$, or equivalently if $a_{m,w} < 0$. We have four cases:

\medskip
\noindent
\textbf{Case 1:} $(m,\sigma(m)), (\sigma(w),w)\in\mathcal V$.
If both $m$ and $w$ are committed, by the first inequality we have $z_m - z_{\sigma(w)} \le \mu_{m,w} a_{m,w}$. As $\mu_{m,w}>0$, we have
\begin{equation*}
    \begin{split}
        z_m - z_{\sigma(w)} &\le \mu_{m,w} a_{m,w} < 0, \text{ or,} \\
        v_mz_m &< v_{\sigma(w)}z_{\sigma(w)}.
    \end{split}
\end{equation*}

\medskip
\noindent
\textbf{Case 2:} $(m, \sigma(m)) \in \gV$, $(\sigma(w),w)\notin \mathcal V$. By the second inequality above, we have $z_m \le \mu_{m,w} a_{m,w}$. As $\mu_{m,w} >0$, we have
\begin{equation*}
    \begin{split}
        z_m &< 0, \text{ or,} \\
        v_mz_m  &< v_{\sigma(w)}z_{\sigma(w)}.
    \end{split}
\end{equation*}

\medskip
\noindent
\textbf{Case 3:}  $(m,\sigma(m)) \notin \mathcal V$, $(\sigma(w),w) \in \mathcal V$. By the third inequality above, we have $z_{\sigma(w)} + \mu_{m,w} a_{m,w} \ge 0$. Once again, as $\mu_{m,w} >0$ and $a_{m,w} < 0$, we have,
\begin{equation*}
    \begin{split}
        0 &< z_{\sigma(w)}, \text{ or,} \\
        v_mz_m  &< v_{\sigma(m)}z_{\sigma(w)}.
    \end{split}
\end{equation*}

\medskip
\noindent
\textbf{Case 4:} $(m, \sigma(m)), (\sigma(w),w) \notin \mathcal V$. Then we know that $a_{m,w} \geq 0$ (as path consistency is satisfied) and $\delta_{m,w} = 0$. 

Therefore, equation $(viii)$ is satisfied for all $(m,w)$. Similar arguments can be made to show that condition $(ix)$ is also satisfied. We can thus conclude that the conditions of Proposition \ref{prop:pc} are feasible whenever Definition \ref{def:PathConsistency} is satisfied.

\medskip
\noindent
{\bf Sufficiency.}
Assume that there exists a solution for the conditions in Proposition \ref{prop:pc} but path consistency is violated. 
Then, there is a path $\rho = (m_1, (m_2,\sigma(m_2))$, $\ldots$, $(m_{n-1},\sigma(m_{n-1})), \sigma(m_n))$ with either $m_1 = m_n$ or $(m_1,\sigma(m_1))$ and $(m_n,\sigma(m_n)) \notin \gV$,  such that for every $1 \leq i < n$, $a_{m_i,\sigma(m_{i+1})}\le 0$ with at least one inequality being strict.
This implies that for these pairs $\xi_{m_i,\sigma(m_{i+1})}=1$ for every $1 \leq i < n$ and for at least one of them $\delta_{m_i,\sigma(m_{i+1})}=1$.
Then we can write down the corresponding set of the last inequalities.
\begin{equation*}
\begin{split}
& v_{m_1} z_{m_1}  \le  v_{m_2} z_{m_2} \\
& v_{m_2} z_{m_2}  \le  v_{m_3} z_{m_3} \\
& \ldots \ \ldots \ \ldots \ \ldots \ \ldots \\
& v_{m_{n-1}} z_{m_{n-1}}  \le  v_{m_{n}} z_{m_n} \\
\end{split}
\end{equation*}
with at least one inequality being strict.
Summing these inequalities gives,
$$
v_{m_1} z_{m_1} < v_{m_{n}} z_{m_n}
$$
If $(m_1,\sigma(m_1))$ and $(m_n,\sigma(m_n)) \notin \gV$, we know that $v_{m_1} = v_{m_n} =0$, then the inequality leads to a contradiction. 
Otherwise if $m_1 = m_n$, the inequality can only be satisfied if $0<0$.
Hence, there cannot be a solution to the system of inequalities if path consistency is violated.
\end{proof}

\subsection{Path Monotonicity}
The path monotonicity condition in Definition \ref{def:PathMonotonicity} can be reformulated in terms of linear inequalities.

\begin{prop}
\label{prop:pm}
Given a data set $\gD$, $A(\gD)$ satisfies path monotonicity if and only if there exist vectors $q^m_{m,\sigma(m)}, q^{\sigma(m)}_{m,\sigma(m)} \in \R^k_+$ and numbers $z_m \in \R$ for all $m \in M$ and vectors $P^m_{m,w}, P^w_{m,w} \in \R^K_{++}$ and numbers $a_{m,w} \in \R$ for all $(m,w) \in (M \cup \lbrace \emptyset \rbrace) \times (W \cup \lbrace \emptyset \rbrace)$ such that
\begin{equation*}
\label{eq:PC}
\begin{aligned} \setcounter{mysubequations}{0}
\mysubnumber\quad & q^m_{m,\sigma(m)}+ q^{\sigma(m)}_{m,\sigma(m)}  = q_{m,\sigma(m)}, \\
\mysubnumber\quad & P^m_{m,w} + P^w_{m,w}  = P_{m,w}, \\
\mysubnumber\quad & a_{m,w} = p_{m,w} q^m_{m,\sigma(m)} + P^m_{m,w} Q_{m,\sigma(m)} + p_{m,w} q^w_{\sigma(w),w}  + P^w_{m,w} Q_{\sigma(w),w} - y_{m,w}, \\
\mysubnumber\quad & v_m z_m - v_{\sigma(w)} z_{\sigma(w)} \le a_{m,w},
\end{aligned}
\end{equation*}
where $v_m = 1 \text{ if } (m,\sigma(m)) \in \gV$ and $0$ otherwise.
\end{prop}

\noindent
Equations $(i)-(iii)$ set the stage by balancing private consumption, Lindahl prices, and defining the edge weights.
Equation $(iv)$ implements path monotonicity using numbers $z_m$ by ensuring that the sum of edge weights along every permissible path is non-negative.
We remark that equation $(iv)$ boils down to the standard cyclical monotonicity condition when all couples are committed (i.e., when $v_m = 1$ for all $m$) \citep[see][]{castillo2019general}.

\begin{proof}
{\bf Necessity.}
Assume that the path monotonicty condition in Definition \ref{def:PathMonotonicity} is satisfied. Let us use the same solution to define $q^m_{m,\sigma(m)}, q^{\sigma(m)}_{m,\sigma(m)}$, $P^m_{m,w}, P^w_{m,w}$ and $a_{m,w}$. The first three conditions are satisfied by default. 
Next, we construct numbers $z_m$ such that condition $(iv)$  is satisfied. Let 
$$
z_m = \min\{ a_{m,w} + a_{\sigma(w),w'} + \ldots \}
$$
where the minimization is defined over (i) all sequence starting with $m\in M$ if $\gV = M\times W$; or (ii) all sequences starting with $m\in M$ and ending with a non-committed couple if  $\gV \subset M\times W$. Consider the four cases:

\medskip
\noindent
\textbf{Case 1:} $(m,\sigma(m)), (\sigma(w),w)\in\mathcal V$. By construction we know that $z_m \le a_{m,w} + z_{\sigma(w)}$. Then
$v_mz_m - v_{\sigma(w)}z_{\sigma(w)} \le  a_{m,w}.
$

\medskip
\noindent
\textbf{Case 2:} $(m,\sigma(m))\in \mathcal V$, $(\sigma(w),w)\notin \mathcal V$.
By construction we know that $z_m \le a_{m,w}$. Therefore, $v_mz_m - v_{\sigma(w)}z_{\sigma(w)} \le  a_{m,w}.$

\medskip
\noindent
\textbf{Case 3:} $(m,\sigma(m))\notin \mathcal V$, $(\sigma(w),w)\in \mathcal V$. We know that adding $(m,w)$ to the sequence that defines $z_{\sigma(w)}$ constitute a permissible path as it starts and ends with non-committed couples. 
As path monotonicity is satisfied, we know that the sum of edge weights along this permissible path ($a_{m,w} +  a_{\sigma(w),w'} + \ldots$) is non-negative. Therefore,
$a_{m,w} + z_{\sigma(w)}\ge 0$. 

\medskip
\noindent
\textbf{Case 4:} $(m,\sigma(m))$, $(\sigma(w),w)\notin \mathcal V$. Then $a_{m,w} \geq 0$ by path monotonicity. Therefore, $v_mz_m - v_{\sigma(w)}z_{\sigma(w)} \le  a_{m,w}.$


\medskip
\noindent
{\bf Sufficiency.}
Assume that there exists a solution for the conditions in Proposition \ref{prop:pm} but path monotonicity is violated. 
Then, there is a path $\rho = (m_1, (m_2,\sigma(m_2))$, $\ldots$, $(m_{n-1},\sigma(m_{n-1})), \sigma(m_n))$ with either $m_1 = m_n$ or $(m_1,\sigma(m_1))$ and $(m_n,\sigma(m_n)) \notin \gV$,  such that $ \sum_{r=1}^{n-1} a_{m_r,\sigma(m_{r+1})} < 0$.
For every couple in this path, condition $(iv)$ of Proposition \ref{prop:pm} implies
$$
v_{m_r} z_{m_r} - v_{m_{r+1}} z_{m_{r+1}} \le a_{m_r,\sigma(m_{r+1})}
$$
Adding up the inequalities along this path gives
$$
v_{m_1} z_{m_1} - v_{m_n} z_{m_n} \le \sum_{r=1}^{n-1} a_{m_r,\sigma(m_{r+1})}
$$
Moreover, either $m_1 = m_n$ or $(m_1,\sigma(m_1))$ and $(m_n,\sigma(m_n)) \notin \gV$, in which case $v_{m_1} = v_{m_n} = 0$. 
In both cases, we obtain 
$$
0 \le \sum_{r=1}^{n-1} a_{m_r,\sigma(m_{r+1})}.
$$
Hence, there cannot be a solution to the system of inequalities in Proposition \ref{prop:pm} if path monotonicity is violated.

\end{proof}

\subsection{Unobserved Commitment Status}
\label{app:LinearProgramming:unknownV}
We can modify the inequalities in Propositions \ref{prop:pc} and \ref{prop:pm} to endogenize the commitment status of couples.
For path consistency conditions, we replace inequalities $(viii)$ and $(ix)$ in Proposition \ref{prop:pc} with the following inequalities,
\begin{align*}
 & z_m  <  z_{\sigma(w)} + (1-\delta_{m,w})M, \\ 
 & z_m  \le  z_{\sigma(w)} + (1-\xi_{m,w})M, \\
 & -v_m M/4 \le z_m \le v_m M/4, 
\end{align*}
where $v_m$ are binary integer variables indicating the unknown commitment status of couples.
If a couple $(m, \sigma(m))$ is not committed ($v_m = 0$), the third inequality enforces $z_m = 0$. Otherwise, it bounds the value of $z_m$ between $-M/4$ and $M/4$.

For path monotonicity conditions, we replace inequality $(iv)$ in Proposition \ref{prop:pm} with the following two inequalities,
\begin{align*}
& z_m - z_{\sigma(w)} \le a_{m,w},  \\ 
& -v_m M/4 \le z_m \le v_m M/4,
\end{align*}
where $v_m$ are binary integer variables representing the unknown commitment status of couples. As with path consistency case, if a couple if committed, the second inequality ensures that $z_m=0$. Otherwise, it bounds $z_m$ between $-M/4$ and $M/4$.

    


\section{Additional Empirical Results}

\label{app:simulation}

\subsection{Excluding Singles}
\label{app:simulation:couples}
In our baseline simulation, we included singles to account for the possibility that married couples consider singles of the opposite gender as potential mates. This approach, however, prevented us from validating the theoretical result presented in Corollary \ref{cor:MutualConsentTransfersNoEmpiricalContent}. In this section, we focus solely on households formed by couples.

Figures \ref{fig_dcavg_withoutsingle} and \ref{fig_shares_withoutsingle} present our main findings, which can be compared to Figures \ref{fig_dcavg} and \ref{fig_shares} in the main text. We see two main effects when comparing these results to the ones in the main text. 
First, the empirical content of the stability conditions decreases when only considering households formed by couples. This is because the presence of singles, who are inherently non-committed, facilitates the formation of blocking paths by committed couples as there are more permissible paths. In the extreme case where all couples are committed and prices are the same within and outside marriages, with household incomes being the sum of individual incomes (i.e., when $\alpha = 0\%$ and the share of committed couples is 100\%), the conditions lose all empirical content. This outcome confirms the theoretical result in Corollary 1, which states that no data set can contain a blocking cycle if transfers are allowed, prices are uniform, and incomes are additive across outside options. 
Second, intrahousehold allocations for committed couples can be identified if path monotonicity is applied and there is variation in prices. However, our empirical rationalizability conditions become less restrictive when excluding singles, leading to wider bound estimates. This widening makes intuitive sense: including singles implies more permissible coalitions and, thus, increased marital competition, which in turn yields more precise identification of intrahousehold allocation patterns.

\begin{figure}[htbp]
\caption{Goodness-of-fit with singles excluded}%
    \label{fig_dcavg_withoutsingle}%
    \centering
    \subfloat[\centering Price variation]{{\includegraphics[scale=0.135]{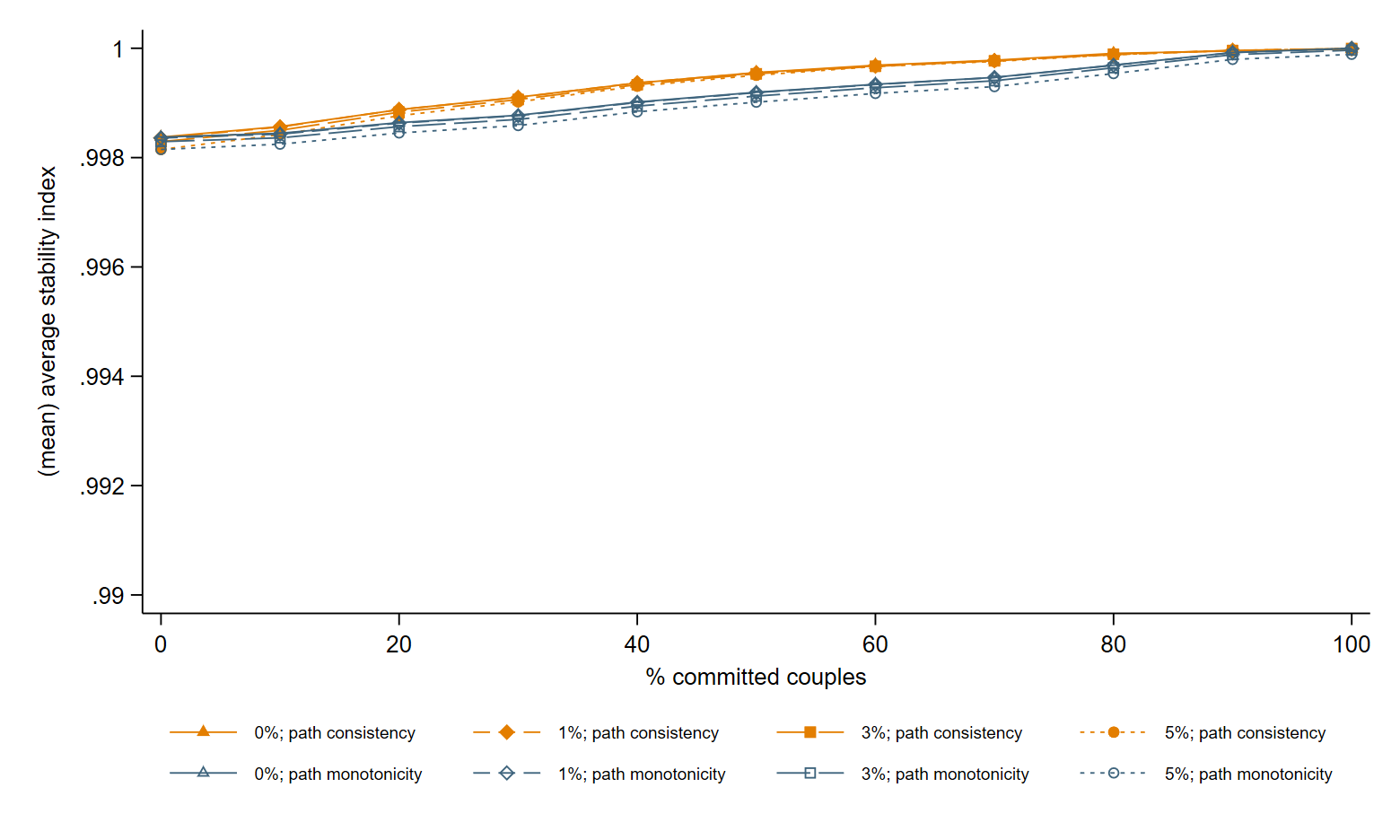} }}%
    \qquad
    \subfloat[\centering Income variation]{{\includegraphics[scale=0.135]{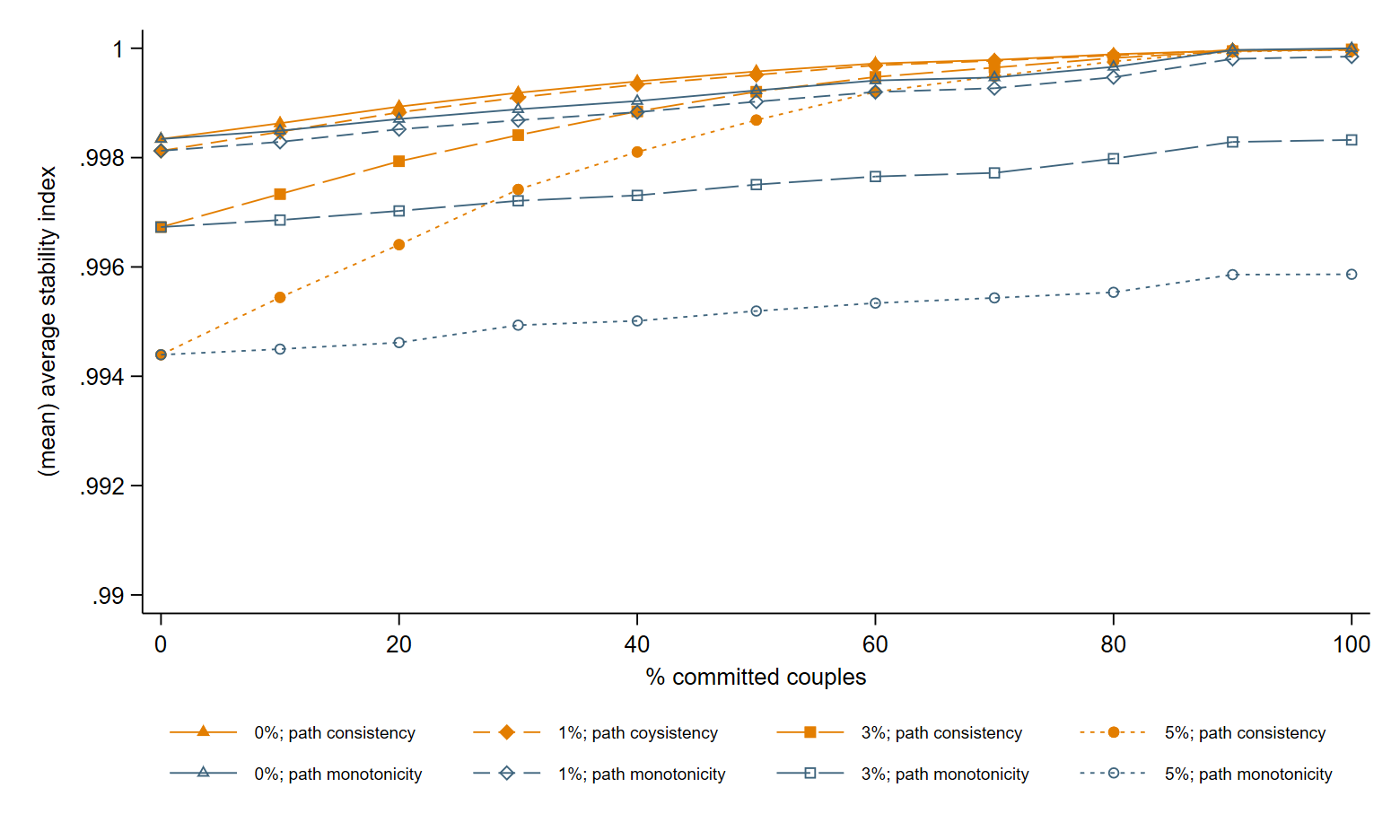} }}%
\end{figure}

\begin{figure}[htbp]
    \centering
    \caption{Relative difference in bounds for committed couples with singles excluded}%
    \label{fig_shares_withoutsingle}%
    \subfloat[\centering Price variation]{{\includegraphics[scale=0.135]{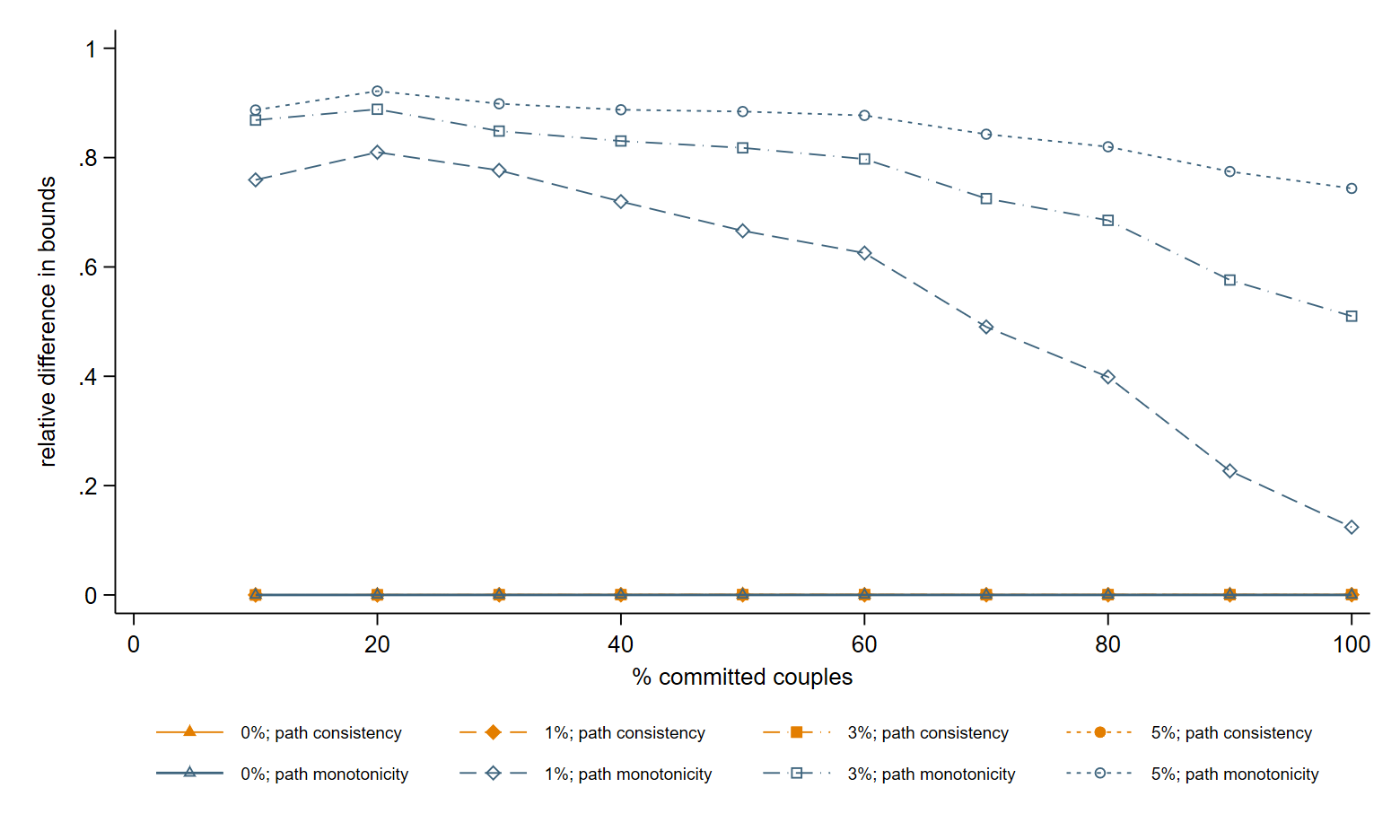} }}%
    \qquad
    \subfloat[\centering Income variation]{{\includegraphics[scale=0.135]{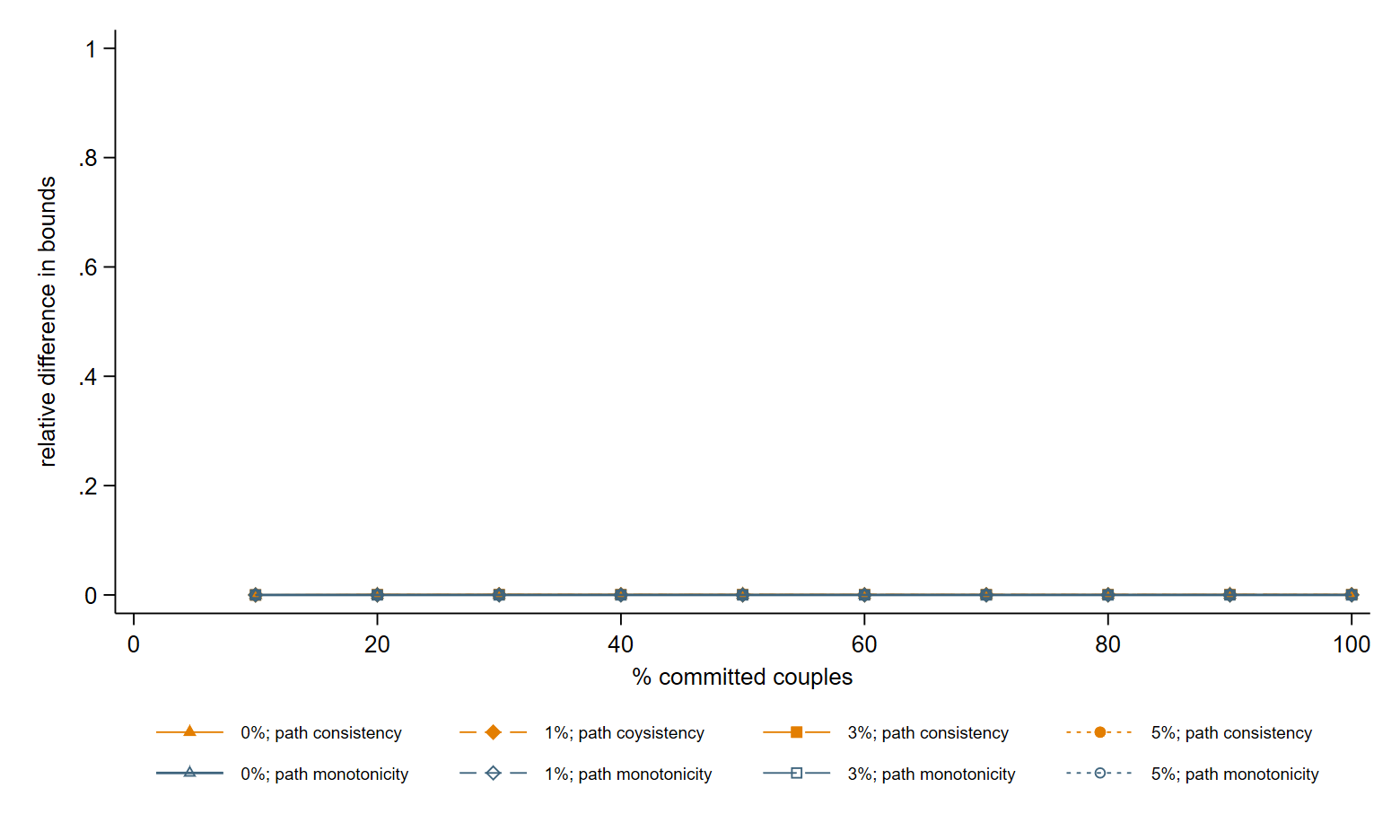} }}%
\end{figure}


\subsection{Non-committed Couples}
\label{app:simulation:noncommitted}
In the main text, we focused on the identifying power of the stability conditions for committed couples. Here, we shift our focus to non-committed couples. Figure \ref{fig_shares_noncommitted} shows the average relative difference between the stable and naive bounds for non-committed couples.
Sub-figures (a) and (b) present the results with variations in prices and income, respectively. In all simulation scenarios, we observe that the stable bounds are notably tighter than the naive bounds for non-committed couples, often leading to near-point identification.
As before, we find that the stable bounds based on path monotonicity are tighter than those derived from path consistency and the tightness of the bounds increases with greater variation in prices or income.

\begin{figure}[htbp]
    \centering
    \caption{Relative difference in bounds for non-committed couples}%
    \label{fig_shares_noncommitted}%
    \subfloat[\centering Price variation]{{\includegraphics[scale=0.135]{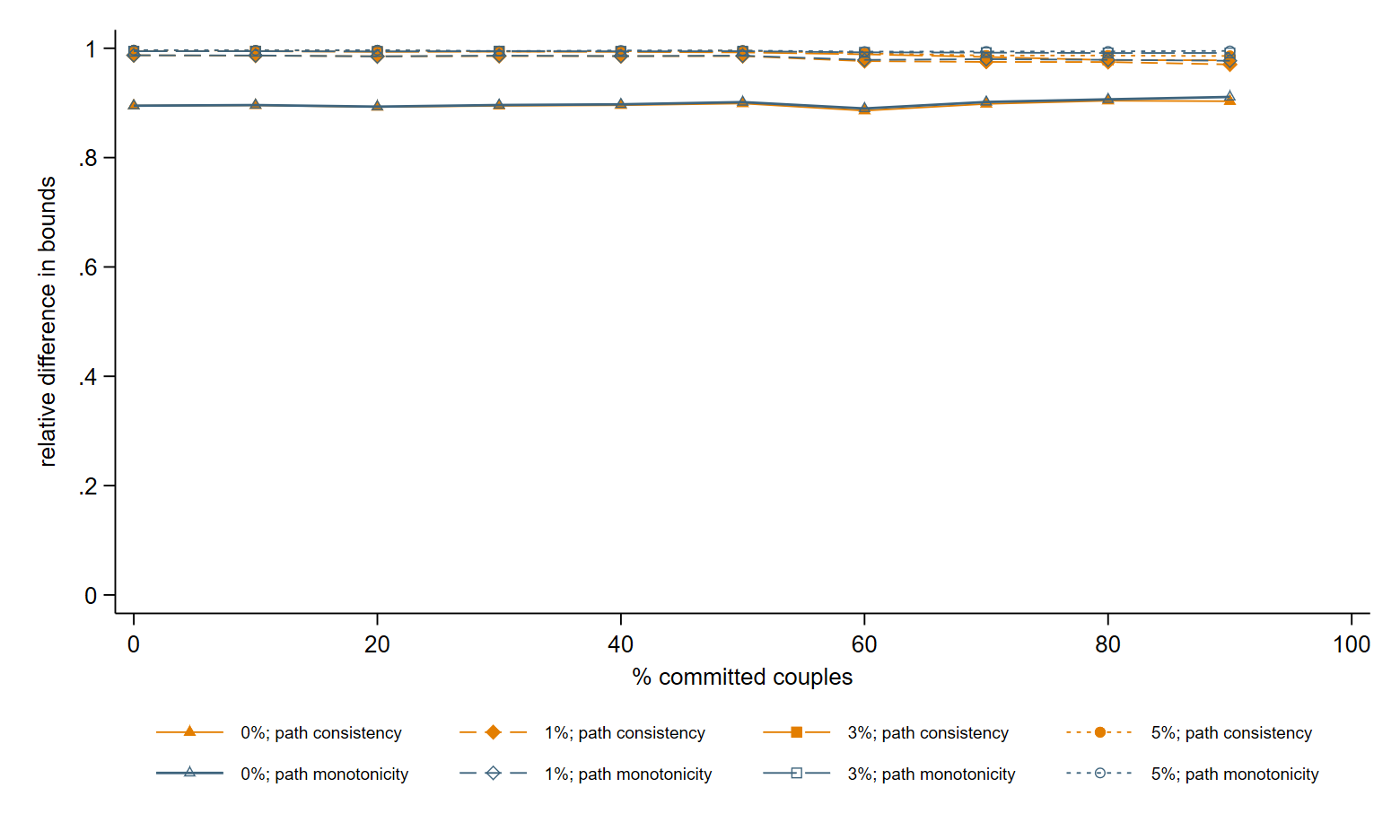} }}%
    \qquad
    \subfloat[\centering Income variation]{{\includegraphics[scale=0.135]{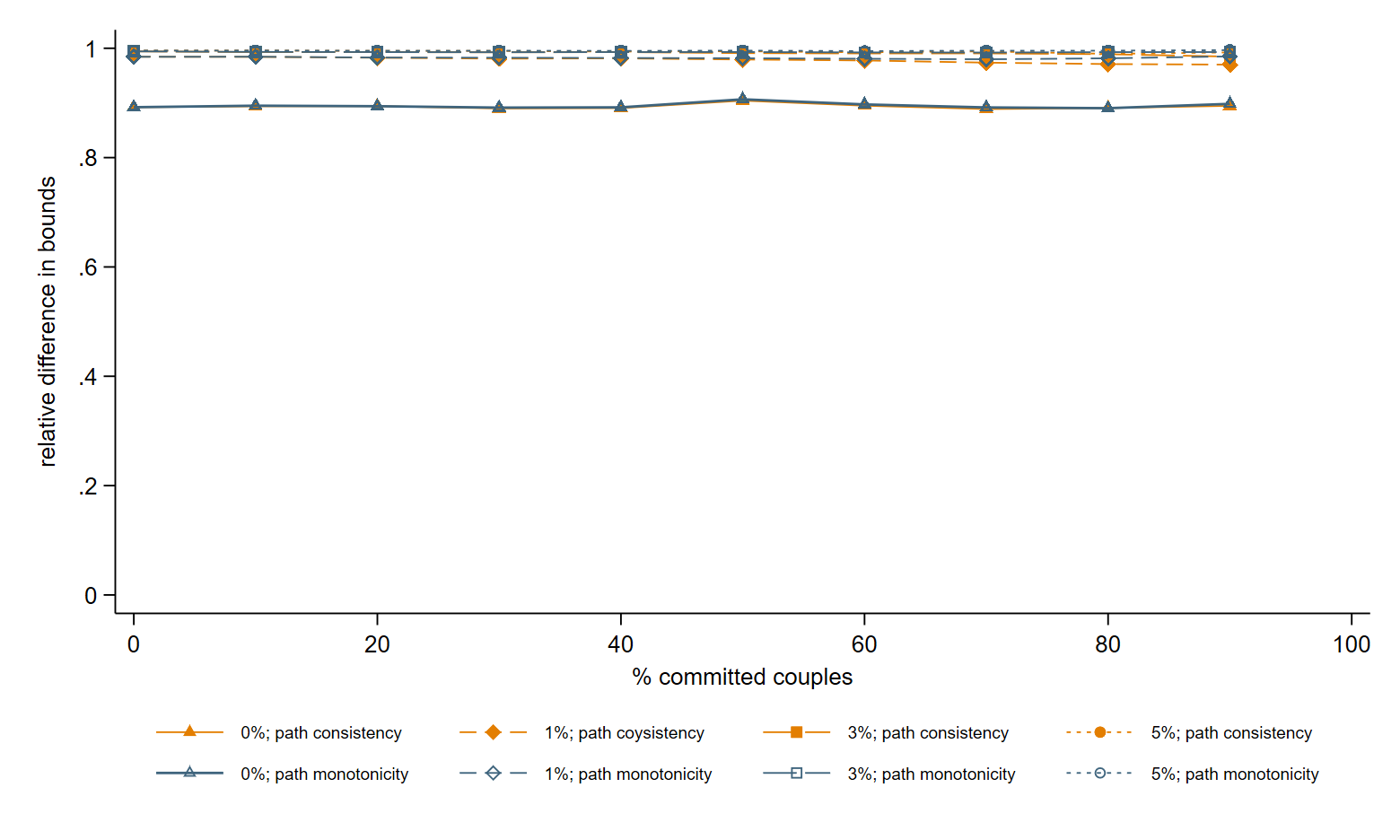} }}%
\end{figure}

\subsection{Subsample Size}
\label{app:simulation:robustness}

In our baseline setting, each subsample consisted of 50 randomly drawn households. As a robustness check, we consider scenarios with 70 and 100 randomly drawn households for each subsample. Figures \ref{fig_passrate_size70} and \ref{fig_shares_size70} show the average stability indices and relative difference in identified bounds, respectively, when using a subsample size of 70. Figures \ref{fig_passrate_size100} and \ref{fig_shares_size100} show the results when using a subsample size of 100.

\begin{figure}[htbp]
    \centering
    \caption{Goodness-of-fit with sample size 70}%
    \label{fig_passrate_size70}%
    \subfloat[\centering Price variation]{{\includegraphics[scale=0.135]{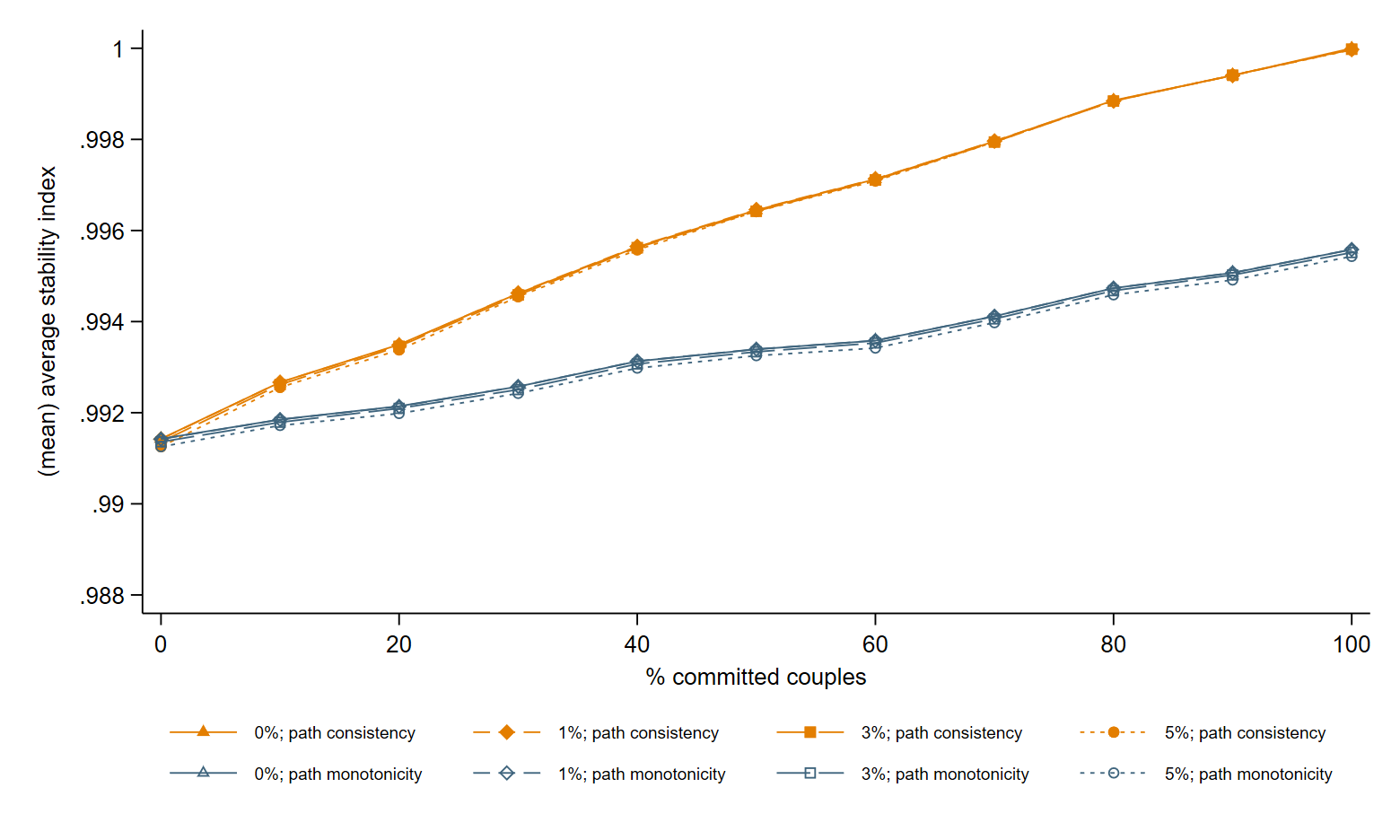} }}%
    \qquad
    \subfloat[\centering Income variation]{{\includegraphics[scale=0.135]{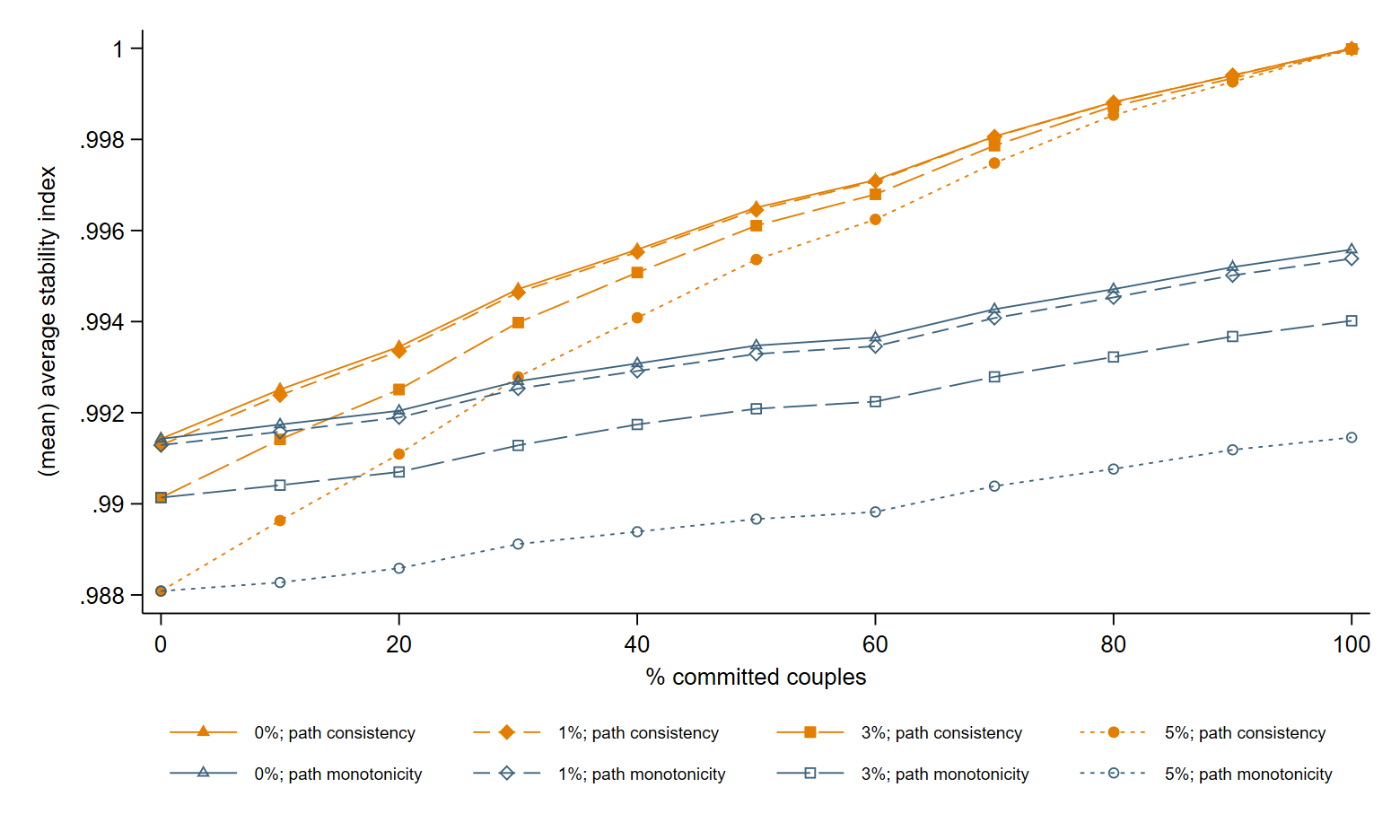} }}   
\end{figure}

\begin{figure}[htbp]
    \centering
    \caption{Relative difference in bounds for committed couples with sample size 70}%
    \label{fig_shares_size70}%
    \subfloat[\centering Price variation]{{\includegraphics[scale=0.135]{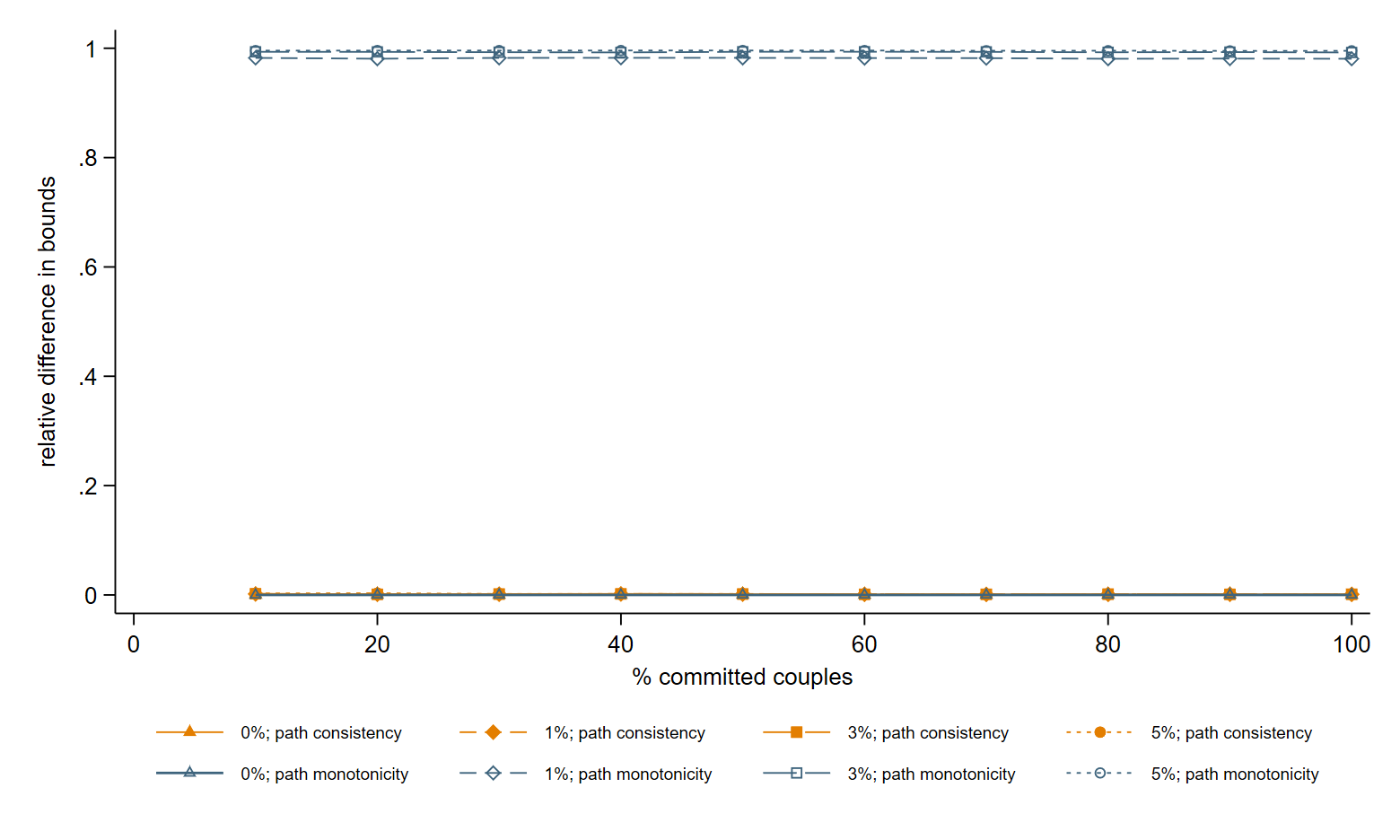} }}%
    \qquad
    \subfloat[\centering Income variation]{{\includegraphics[scale=0.135]{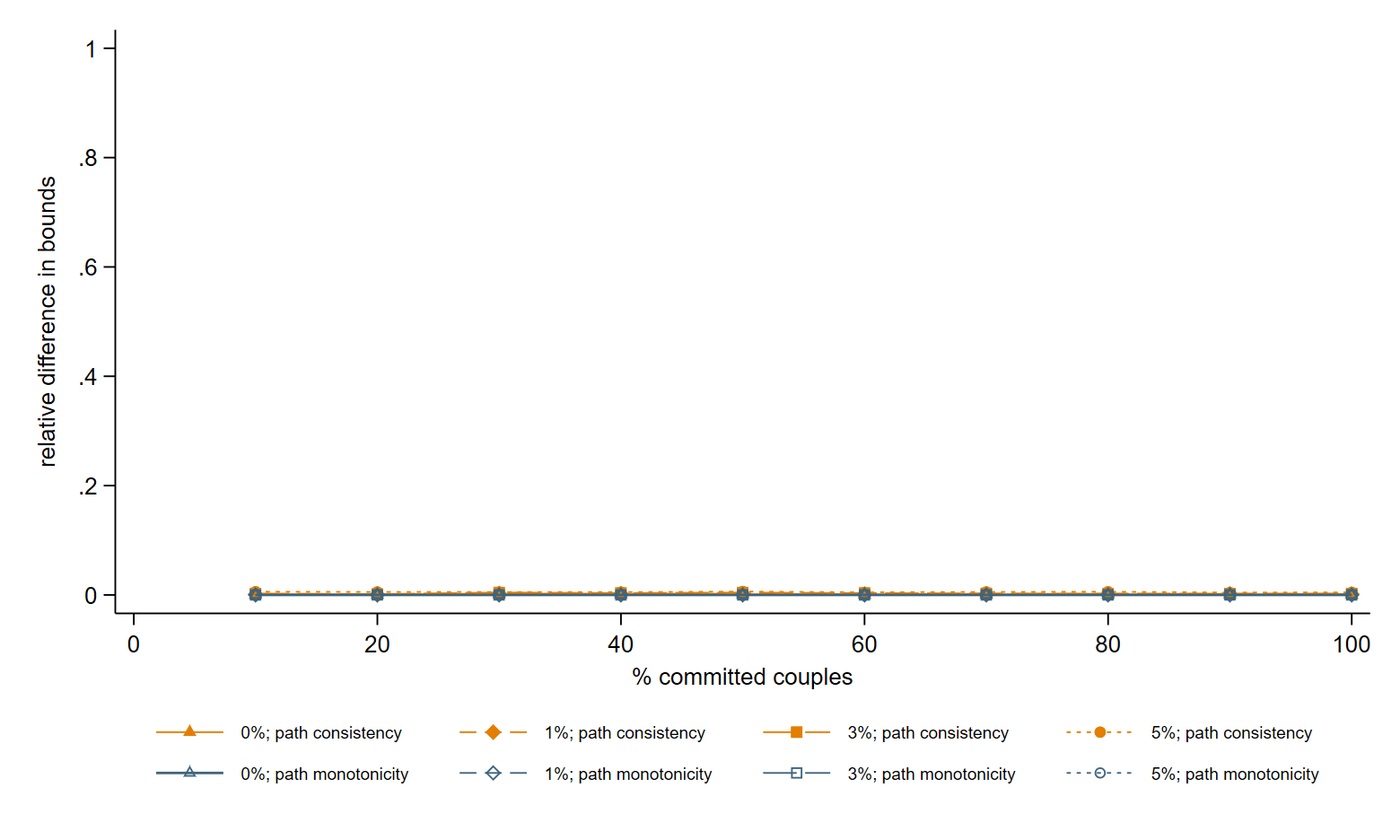} }}  
\end{figure}

\begin{figure}[htb]
    \centering
    \caption{Goodness-of-fit with sample size 100}%
    \label{fig_passrate_size100}%
    \subfloat[\centering Price variation]{{\includegraphics[scale=0.135]{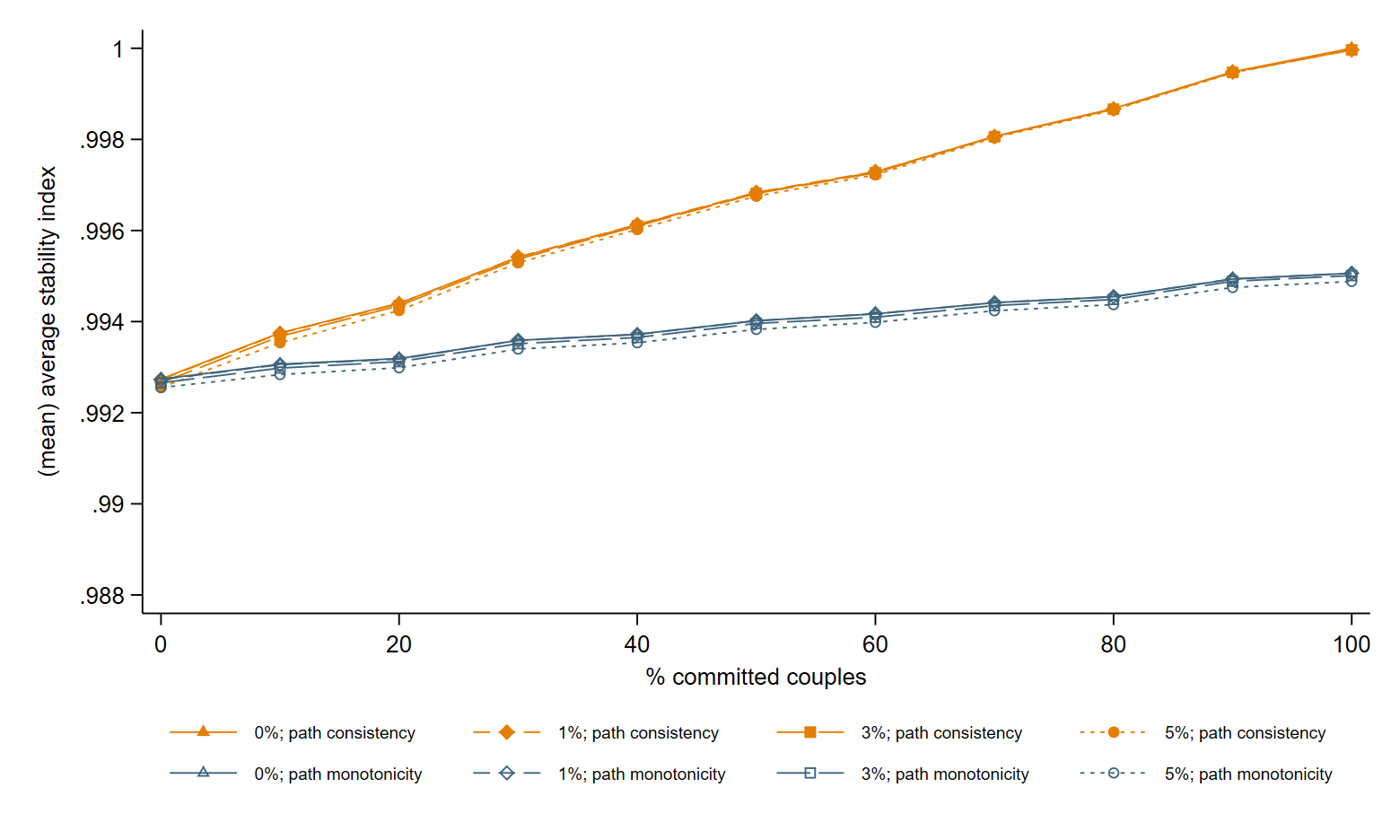} }}%
    \qquad
    \subfloat[\centering Income variation]{{\includegraphics[scale=0.135]{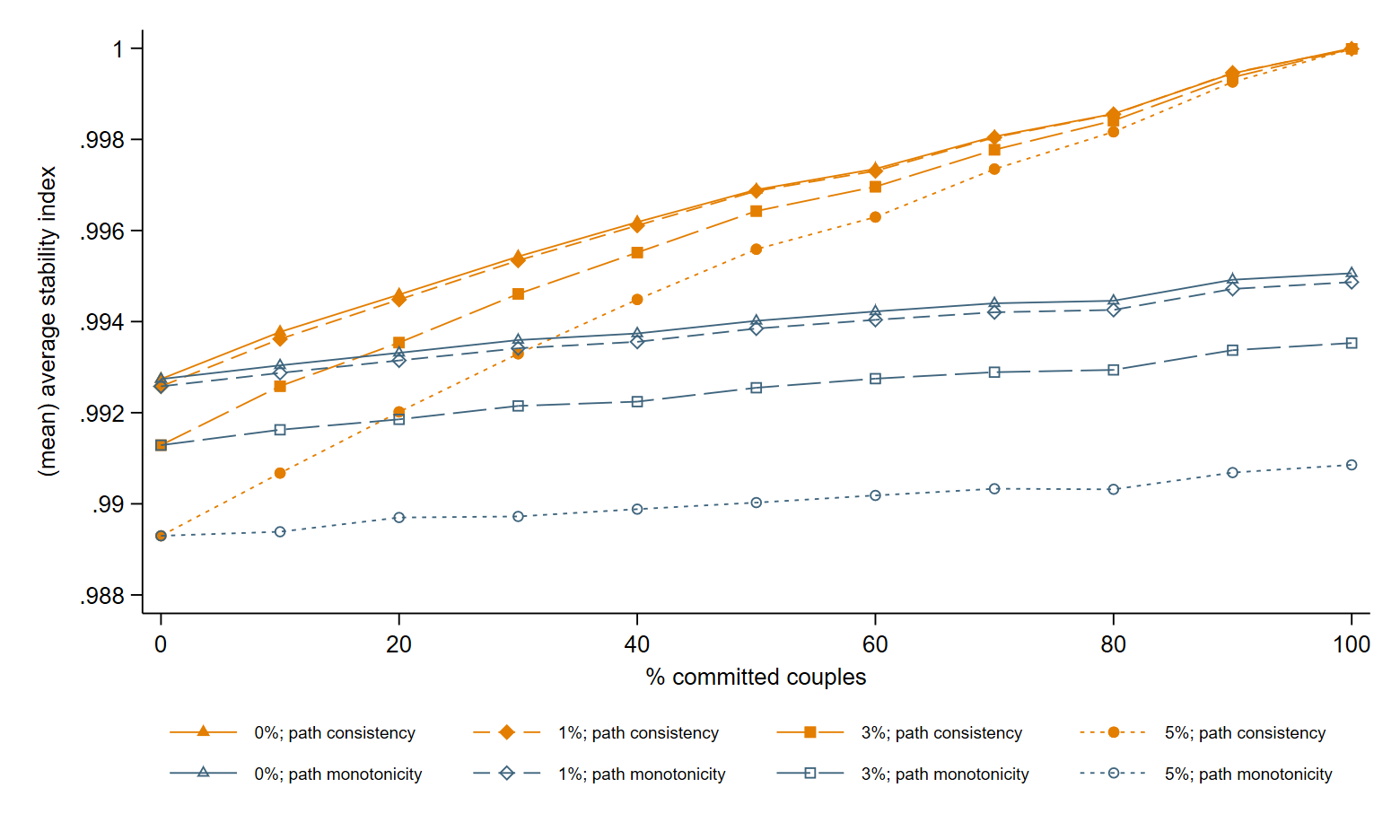} }}   
\end{figure}

\begin{figure}[htb]
    \centering
    \caption{Relative difference in bounds for committed couples with sample size 100}%
    \label{fig_shares_size100}%
    \subfloat[\centering Price variation]{{\includegraphics[scale=0.135]{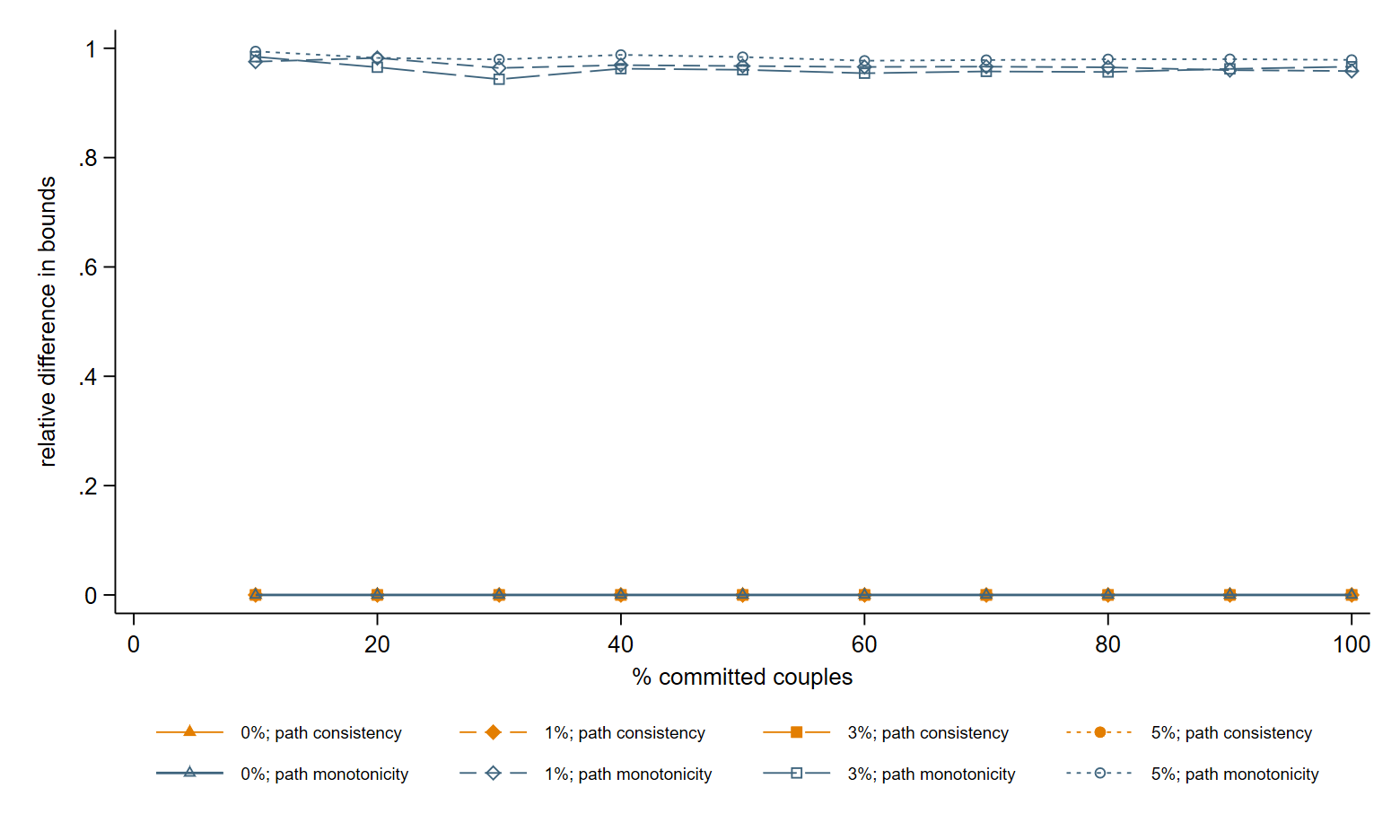} }}%
    \qquad
    \subfloat[\centering Income variation]{{\includegraphics[scale=0.135]{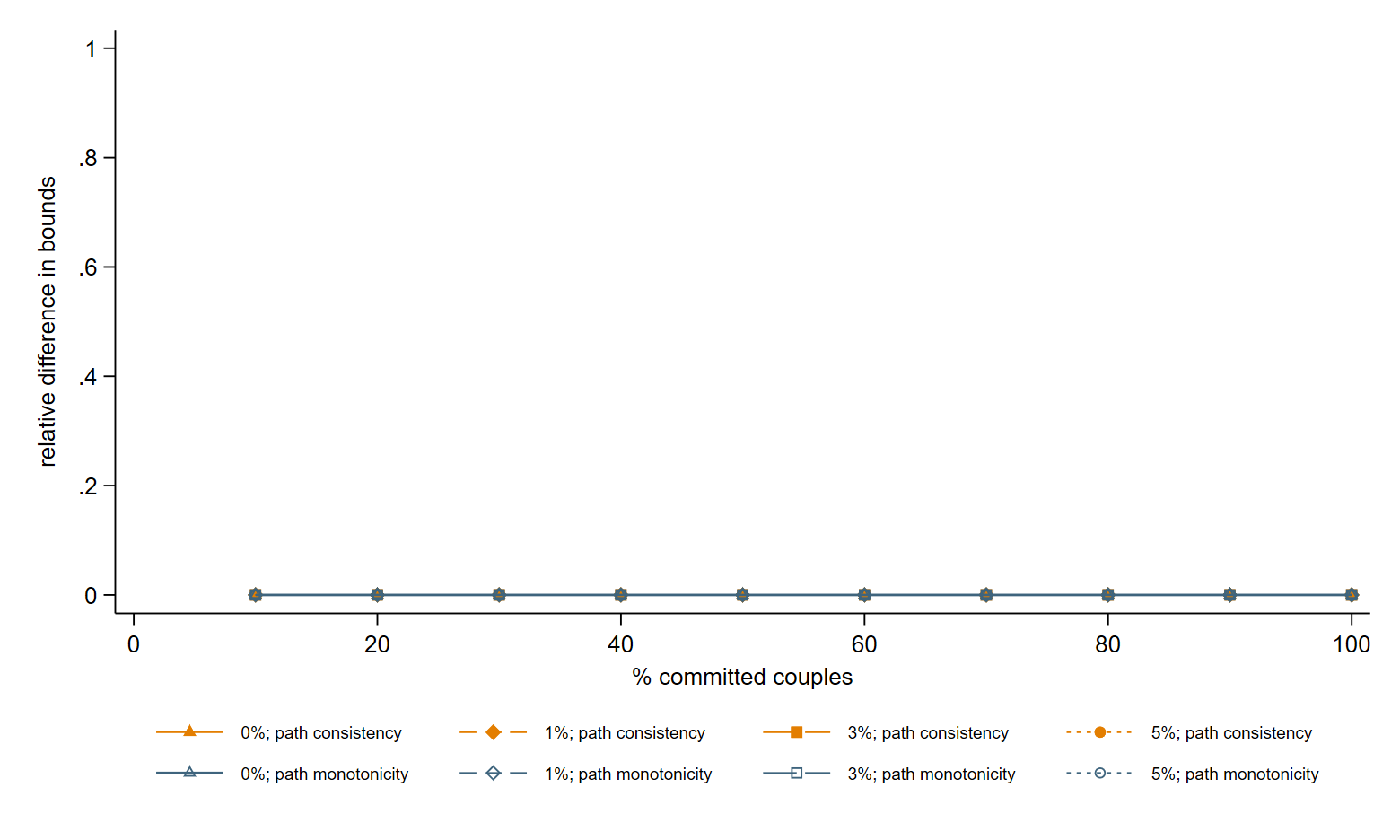} }}  
\end{figure}

We find that increasing the sample size generally results in lower stability indices and tighter bound estimates. Overall, the results closely resemble those presented in the main text. Notably, under path monotonicity, even a small degree of price variation ($\alpha = 1\%$) results in a significant decrease in the average stability index  and almost point identification of private consumption shares. However, even with a large market size of 100 and a high degree of price variation ($\alpha = 5\%$), we do not observe any identifying power from the path consistency condition. While path consistency exhibits empirical content when some non-committed couples are included, it does not lead to identification. Therefore, we conclude that only the implications of stability with transfers are likely to be useful in practical applications.

\clearpage
\bibliographystyle{te.bst}
\bibliography{refs}

\end{document}


\maketitle

\tableofcontents

\section{Empirical Illustration}
\label{section_empirical}
Through simulation exercises, we explore the empirical power of the three stability conditions corresponding to unilateral divorce and mutual consent divorce with and without transfers.
Our analysis is based on a subset of households drawn from the LISS panel, a representative survey of households in the Netherlands conducted by CentERdata.\footnote{
Although the Netherlands has a policy of allowing unilateral divorce, we choose to use this dataset for the illustration as our main aim is to explore the empirical tractability of the rationalizability conditions. Furthermore, this dataset has been utilized by numerous studies that employ the collective household model. In particular, \citet{cherchye2017household} use the same dataset to identify intrahousehold resource allocation in the context of unilateral divorce.} 
The survey collects rich data on economics and sociodemographic variables at both individual and household levels. 
For this application, we use a part of the sample analyzed by \citet{cherchye2017household}. 
Their sample selection criteria consider individuals with or without children, working at least 10 hours per week in the labor market, and aged between 25 and 65. 
The sample contains 832 individuals comprising of 239 couples, 166 single males, and 188 single females. 
For the sake of the illustration and due to computational limitations, we focus on  couples without children and use randomly selected 30 couples in each simulation. 

In terms of the setup, we consider a labor supply setting where individuals choose their leisure, private consumption, and public consumption.
In the data, we observe both aggregated household expenditure as well as some assignable expenditure.
Following \citet{cherchye2017household}, we use all expenditure information to form a Hicksian good with a price normalized to one.
We assume that the non-assignable expenditure is equally divided between public and private consumption.
For leisure, we take the price to be the individual's hourly wage.

\paragraph*{Stability Indices.}
The rationalizability conditions presented above are strict in nature. 
The observed household data would either satisfy the constraints or fail to find a feasible solution. 
In reality, household consumption behavior may not be exactly consistent with the model if, for example, the data contain measurement errors, there are frictions in the marriage market, or other factors, such as match quality, affect marital behavior.
In such cases, it is useful to quantify the deviations of observed data from exactly rationalizable behavior.
Following \citet{cherchye2017household}, we evaluate the goodness-of-fit of a model using stability indices. These indices allow us to quantify the degree to which the observed behavior is consistent with the exactly rationalizable behavior.

Formally, we include a stability index $s_{m,w}$ in the edge weight $a_{m,w}$ corresponding to the outside option $(m,w) \in  M \cup \{\emptyset \} \times  W \cup \{\emptyset \}$. 
In particular, we redefine $a_{m,w}$ as
$$
a_{m,w} = p_{m,w} ( q_{m,\sigma(m)}^m +  q_{\sigma(w),w}^w) + P^m_{m,w} Q_{m,\sigma(m)} + P^w_{m,w} Q_{\sigma(w),w} - y_{m,w} s_{m,w}.
$$
Further, we add the restriction $0 \leq s_{m,w} \leq 1$. Imposing $s_{m,w} = 1$ results in the original rationalizability restrictions, while imposing $s_{m,w}= 0$ would rationalize any data. 
Intuitively, the stability indices measure the loss of post-divorce income needed to represent the observed marriages as stable.
Generally, a lower stability index corresponds to a higher income loss associated with a particular outside option. 
This can be interpreted as a greater violation of the underlying model assumptions. 

In our simulations, we will use stability indices to evaluate the empirical content of the model. To identify the values of stability indices for a given data set $\mathcal{D}$, we compute
\begin{equation*}
    \max \sum_{m} \sum_{w} s_{m,w}
\end{equation*}
subject to the rationalizability conditions. This gives a stability index $s_{m,w}$ for each outside option.
If the original constraints are satisfied, then there is no need for adjustment, and all stability indices will be equal to one. 
Otherwise, a strictly smaller index will be required to rationalize the behavior.
In what follows, we will summarize the stability indices by computing the average value of all stability indices. 
Intuitively, the average stability index measures the average income loss required to make the market stable.

\paragraph{Simulation Setup.}
We consider three types of data sets. 
We randomly draw 30 couples to be used across all the scenarios.
In each scenario, we introduce variation in prices and/or income and draw 100 instances of marriage markets from the full sample. 

\begin{itemize}
    \item [(1)] {\bf Prices}. In the first scenario, we introduce variation in prices of private and public goods across potential matches. 
    Such differences could arise in real-life situations when individuals relocate for the purpose of marriage and encounter varying prices in different regions.
    To capture this, we assume that prices for both private and public Hicksian consumption are equal to one in existing marriages, while prices for potential marriages are randomly assigned as follows:
    \begin{equation*}
    \begin{split}
        p_{m,f} = 1 + \alpha \epsilon \; \text{where} \; \epsilon \sim U[-1,1] \; \text{and}  \; \alpha \in \lbrace 0\%, 5\%, 10\%, 15\%, 20\%, 25\% \rbrace, \\
        P_{m,f} = 1 +  \alpha \eta\; \text{where} \; \eta \sim U[-1,1] \; \text{and}  \; \alpha \in \lbrace 0\%, 5\%, 10\%, 15\%, 20\%, 25\% \rbrace.
        \end{split}
    \end{equation*}
 %
    \item [(2)] {\bf Income}. 
    In the second scenario, we introduce variation in household incomes, while prices of private and public goods are assumed to be the same.
    Empirically, such a variation could be present in the data if there is a marriage premium and individual incomes depend on the partners' characteristics \citep[see][]{korenman1991does}.
    To capture this, we assume that the household incomes in existing marriages are the sum of individual labor incomes, while household incomes in hypothetical marriages are determined as follows:
    \begin{equation*}
        y_{m,f} = (y_m  + y_f) (1 + \alpha \epsilon)  \quad \text{where,} \epsilon \sim U[-1,1] \; \text{and}  \; \alpha \in \lbrace 0\%, 5\%, 10\%, 15\%, 20\%, 25\% \rbrace.
    \end{equation*}
%
%
    \item [(3)] {\bf Prices and Income}. In the final scenario, we introduce variation in both prices and income. This combines the two variations shown above.
\end{itemize}

\subsection{Stability Indices}
We start by checking whether and to what extent the simulated household behavior in the three data scenarios would be consistent with the rationalizability conditions. We summarize the results by documenting the average stability indices across the 100 draws.
As discussed above, stability indices are important indicators of the empirical content of the models.
If a stability index is strictly below one, then the revealed preference conditions do have empirical content as formally the conditions imposed are violated.
As we will show later, the presence of empirical content is necessary (though not sufficient) for the identification of household allocation.

\begin{table}[htbp]
\centering
\begin{tabular}{lcHcHcHcHcHc}
\toprule
   & 0\%  & 5\%  & 5\%  & 15\%  & 10\%   & 25\%   & 15\% & 35\%  & 20\% & 45\%  & 25\% \\ \midrule
{\it Panel A: prices}    &  &   &  &  & & & & & &  & \\ 
\multicolumn{1}{l}{unilateral}                    
& \begin{tabular}[c]{@{}c@{}}.9995\\ \footnotesize .97  1  1\end{tabular} 
& \begin{tabular}[c]{@{}c@{}}.9994\\ \footnotesize .97  1  1\end{tabular} 
& \begin{tabular}[c]{@{}c@{}}.9992\\ \footnotesize .95  1  1\end{tabular} 
& \begin{tabular}[c]{@{}c@{}}.9990\\ \footnotesize .95 1 1 1 1\end{tabular} 
& \begin{tabular}[c]{@{}c@{}}.9985\\ \footnotesize .94 1 1\end{tabular} 
& \begin{tabular}[c]{@{}c@{}}.9981\\ \footnotesize .93 1 1 1 1\end{tabular} 
& \begin{tabular}[c]{@{}c@{}}.9975\\ \footnotesize .93 1 1\end{tabular} 
& \begin{tabular}[c]{@{}c@{}}.9969\\ \footnotesize .92 1 1 1 1\end{tabular} 
& \begin{tabular}[c]{@{}c@{}}.9963\\ \footnotesize .91 1 1\end{tabular}   
& \begin{tabular}[c]{@{}c@{}}.9955\\ \footnotesize .90 1 1 1 1\end{tabular}   
& \begin{tabular}[c]{@{}c@{}}.9947\\ \footnotesize .89 1 1\end{tabular}   \\

\multicolumn{1}{l}{mutual consent w/o transfers}                
& \begin{tabular}[c]{@{}c@{}}1\\ \footnotesize 1 1 1\end{tabular}       
& \begin{tabular}[c]{@{}c@{}}1\\ \footnotesize .97 1 1 1 1\end{tabular}       
& \begin{tabular}[c]{@{}c@{}}1\\ \footnotesize .98 1 1\end{tabular}       
& \begin{tabular}[c]{@{}c@{}}1\\ \footnotesize .98 1 1 1 1\end{tabular}       
& \begin{tabular}[c]{@{}c@{}}1\\ \footnotesize .97 1 1\end{tabular}       
& \begin{tabular}[c]{@{}c@{}}1\\ \footnotesize .95 1 1 1 1\end{tabular}       
& \begin{tabular}[c]{@{}c@{}}1\\ \footnotesize .97 1 1\end{tabular}       
& \begin{tabular}[c]{@{}c@{}}1\\ \footnotesize .98 1 1 1 1\end{tabular}       
& \begin{tabular}[c]{@{}c@{}}1\\ \footnotesize .98 1 1\end{tabular}         
& \begin{tabular}[c]{@{}c@{}}1\\ \footnotesize .98 1 1 1 1\end{tabular}         
& \begin{tabular}[c]{@{}c@{}}1\\ \footnotesize .97 1 1\end{tabular}         \\

\multicolumn{1}{l}{mutual consent with transfers} 
& \begin{tabular}[c]{@{}c@{}}1\\ \footnotesize 1 1 1\end{tabular}    
& \begin{tabular}[c]{@{}c@{}}1\\ \footnotesize .99 1 1 1 1\end{tabular}  
& \begin{tabular}[c]{@{}c@{}}.9999\\ \footnotesize .97 1 1\end{tabular} 
& \begin{tabular}[c]{@{}c@{}}.9996\\ \footnotesize .97 1 1 1 1\end{tabular} 
& \begin{tabular}[c]{@{}c@{}}.9993\\ \footnotesize .96 1 1\end{tabular} 
& \begin{tabular}[c]{@{}c@{}}.9989\\ \footnotesize .95 1 1 1 1\end{tabular} 
& \begin{tabular}[c]{@{}c@{}}.9984\\ \footnotesize  94 1 1\end{tabular} 
& \begin{tabular}[c]{@{}c@{}}.9979\\ \footnotesize .93 1 1 1 1\end{tabular} 
& \begin{tabular}[c]{@{}c@{}}.9971\\ \footnotesize .91 1 1\end{tabular} 
& \begin{tabular}[c]{@{}c@{}}.9965\\ \footnotesize .87 1 1 1 1\end{tabular}   
& \begin{tabular}[c]{@{}c@{}}.9958\\ \footnotesize .89 1 1\end{tabular}   \\ \midrule
                                                 {\it Panel B: income}    &  &   &  &  & & & & & &  & \\ 
\multicolumn{1}{l}{unilateral}                    
& \begin{tabular}[c]{@{}c@{}}.9995\\ \footnotesize .97 1 1\end{tabular} 
& \begin{tabular}[c]{@{}c@{}}.9994\\ \footnotesize .97 1 1 1 1\end{tabular} 
& \begin{tabular}[c]{@{}c@{}}.9956\\ \footnotesize .89 1 1\end{tabular} 
& \begin{tabular}[c]{@{}c@{}}.9988\\ \footnotesize .94 1 1 1 1\end{tabular} 
& \begin{tabular}[c]{@{}c@{}}.9876\\ \footnotesize .86 1 1\end{tabular} 
& \begin{tabular}[c]{@{}c@{}}.9978\\ \footnotesize .91 1 1 1 1\end{tabular} 
& \begin{tabular}[c]{@{}c@{}}.9787\\ \footnotesize .82 1 1\end{tabular} 
& \begin{tabular}[c]{@{}c@{}}.9964\\ \footnotesize .88 1 1 1 1\end{tabular} 
& \begin{tabular}[c]{@{}c@{}}.9698\\ \footnotesize .76 1 1\end{tabular}   
& \begin{tabular}[c]{@{}c@{}}.9949\\ \footnotesize .87 1 1 1 1\end{tabular}  
& \begin{tabular}[c]{@{}c@{}}.9610\\ \footnotesize .71 1 1\end{tabular} \\

\multicolumn{1}{l}{mutual consent w/o transfers}                
& \begin{tabular}[c]{@{}c@{}}1\\ \footnotesize 1 1 1\end{tabular}       
& \begin{tabular}[c]{@{}c@{}}1\\ \footnotesize .99 1 1 1 1\end{tabular}       
& \begin{tabular}[c]{@{}c@{}}1\\ \footnotesize 1 1 1\end{tabular}       
& \begin{tabular}[c]{@{}c@{}}1\\ \footnotesize .98 1 1 1 1\end{tabular}       
& \begin{tabular}[c]{@{}c@{}}1\\ \footnotesize .98 1 1\end{tabular}       
& \begin{tabular}[c]{@{}c@{}}1\\ \footnotesize .99 1 1 1 1\end{tabular}       
& \begin{tabular}[c]{@{}c@{}}1\\ \footnotesize .96 1 1\end{tabular}       
& \begin{tabular}[c]{@{}c@{}}1\\ \footnotesize .96 1 1 1 1\end{tabular}       
& \begin{tabular}[c]{@{}c@{}}1\\ \footnotesize .94 1 1\end{tabular}         
& \begin{tabular}[c]{@{}c@{}}1\\ \footnotesize .95 1 1 1 1\end{tabular}         
& \begin{tabular}[c]{@{}c@{}}1\\ \footnotesize .95 1 1\end{tabular}         \\

\multicolumn{1}{l}{mutual consent with transfers} 
& \begin{tabular}[c]{@{}c@{}}1\\ \footnotesize 1 1 1\end{tabular}   
& \begin{tabular}[c]{@{}c@{}}.9999\\ \footnotesize .99 1 1 1 1\end{tabular} 
& \begin{tabular}[c]{@{}c@{}}.9962\\ \footnotesize .93 1 1\end{tabular} 
& \begin{tabular}[c]{@{}c@{}}.9995\\ \footnotesize .96 1 1 1 1\end{tabular} 
& \begin{tabular}[c]{@{}c@{}}.9880\\ \footnotesize .88 1 1\end{tabular} 
& \begin{tabular}[c]{@{}c@{}}.9985\\ \footnotesize .94 1 1 1 1\end{tabular} 
& \begin{tabular}[c]{@{}c@{}}.9790\\ \footnotesize .82 1 1\end{tabular} 
& \begin{tabular}[c]{@{}c@{}}.9971\\ \footnotesize .92 1 1 1 1\end{tabular} 
& \begin{tabular}[c]{@{}c@{}}.9702\\ \footnotesize .76 1 1\end{tabular}  
& \begin{tabular}[c]{@{}c@{}}.9955\\ \footnotesize .91 1 1 1 1\end{tabular}  
& \begin{tabular}[c]{@{}c@{}}.9613\\ \footnotesize .73 1 1\end{tabular}   \\ \midrule
    {\it Panel C: prices and income}    &  &   &  &  & & & & & &  & \\ 
\multicolumn{1}{l}{unilateral}                    
& \begin{tabular}[c]{@{}c@{}}.9995\\ \footnotesize .97 1 1\end{tabular} 
& \begin{tabular}[c]{@{}c@{}}.9993\\ \footnotesize .96 1 1 1 1\end{tabular} 
& \begin{tabular}[c]{@{}c@{}}.9954\\ \footnotesize .90 1 1\end{tabular} 
& \begin{tabular}[c]{@{}c@{}}.9983\\ \footnotesize .94 1 1 1 1\end{tabular} 
& \begin{tabular}[c]{@{}c@{}}.9870\\ \footnotesize .86 1 1\end{tabular} 
& \begin{tabular}[c]{@{}c@{}}.9967\\ \footnotesize .90 1 1 1 1\end{tabular} 
& \begin{tabular}[c]{@{}c@{}}.9776\\ \footnotesize .80 1 1\end{tabular} 
& \begin{tabular}[c]{@{}c@{}}.9944\\ \footnotesize .86 1 1 1 1\end{tabular} 
& \begin{tabular}[c]{@{}c@{}}.9684\\ \footnotesize .74 1 1\end{tabular} 
& \begin{tabular}[c]{@{}c@{}}.9919\\ \footnotesize .81 .99 1 1 1\end{tabular} 
& \begin{tabular}[c]{@{}c@{}}.9591\\ \footnotesize .69 1 1\end{tabular} \\

\multicolumn{1}{l}{mutual consent w/o transfers}                
& \begin{tabular}[c]{@{}c@{}}1\\ \footnotesize  1 1 1\end{tabular}       
& \begin{tabular}[c]{@{}c@{}}1\\ \footnotesize  .98 1 1 1 1\end{tabular}       
& \begin{tabular}[c]{@{}c@{}}1\\ \footnotesize .95 1 1\end{tabular}       
& \begin{tabular}[c]{@{}c@{}}1\\ \footnotesize .98 1 1 1 1\end{tabular}       
& \begin{tabular}[c]{@{}c@{}}1\\ \footnotesize 1 1 1\end{tabular}       
& \begin{tabular}[c]{@{}c@{}}1\\ \footnotesize .96 1 1 1 1\end{tabular}       
& \begin{tabular}[c]{@{}c@{}}1\\ \footnotesize .96 1 1\end{tabular}       
& \begin{tabular}[c]{@{}c@{}}1\\ \footnotesize .95 1 1 1 1\end{tabular}       
& \begin{tabular}[c]{@{}c@{}}1\\ \footnotesize 1 1 1\end{tabular}         
& \begin{tabular}[c]{@{}c@{}}1\\ \footnotesize .94 1 1 1 1\end{tabular}         
& \begin{tabular}[c]{@{}c@{}}1\\ \footnotesize 1 1 1\end{tabular}         \\

\multicolumn{1}{l}{mutual consent with transfers} 
& \begin{tabular}[c]{@{}c@{}}1\\ \footnotesize 1 1 1\end{tabular}       
& \begin{tabular}[c]{@{}c@{}}.9999\\ \footnotesize .98 1 1 1 1\end{tabular} 
& \begin{tabular}[c]{@{}c@{}}.9961\\ \footnotesize .92 1 1\end{tabular} 
& \begin{tabular}[c]{@{}c@{}}.9991\\ \footnotesize .94 1 1 1 1\end{tabular} 
& \begin{tabular}[c]{@{}c@{}}.9877\\ \footnotesize .86 1 1\end{tabular} 
& \begin{tabular}[c]{@{}c@{}}.9975\\ \footnotesize .92 1 1 1 1\end{tabular} 
& \begin{tabular}[c]{@{}c@{}}.9783\\ \footnotesize .79 1 1\end{tabular} 
& \begin{tabular}[c]{@{}c@{}}.9954\\ \footnotesize .88 1 1 1 1\end{tabular} 
& \begin{tabular}[c]{@{}c@{}}.9693\\ \footnotesize .73 1 1\end{tabular}   
& \begin{tabular}[c]{@{}c@{}}.9930\\ \footnotesize .86 .99 1 1 1\end{tabular} 
& \begin{tabular}[c]{@{}c@{}}.9601\\ \footnotesize .67 1 1\end{tabular} \\ \bottomrule
\end{tabular}
{
\\[0.2em]
\raggedright  \scriptsize {\it Notes:} The number in each cell is the mean value of the average stability indices across all the simulations in the given scenario. The three numbers below correspond to the minimum, median and maximum of average stability indices across simulations.
\par}
\caption{Stability indices}
\label{tab:ResultsDivorceCosts}
\end{table}

Table \ref{tab:ResultsDivorceCosts} presents the results of this exercise.
Panels A, B, and C correspond to the scenarios where we introduce variation in prices, income, and both prices and income, respectively.
Each panel contains three rows for the three theories under consideration: unilateral divorce, mutual consent divorce without transfers, and mutual consent divorce with transfers.
There are six columns corresponding to the different scenarios from the default $\alpha = 0$\%, where the original prices and wages are used in the computation, to $\alpha = 25$\% allowing for a large variation in the prices and/or incomes across potential matches.

As expected, Table \ref{tab:ResultsDivorceCosts} shows that the stability conditions under unilateral divorce are the most restrictive. In all scenarios, the observed household behavior is not exactly consistent with the no negative edges condition. That is, we require strictly below one stability index to rationalize the observed consumption and marriages. This is true even when we assume that prices are the same across counterfactual matches and individual incomes inside and outside marriage are the same ($\alpha = 0$\%). By contrast, when prices and incomes are assumed to be the same across matches, the stability conditions under mutual consent divorce (with and without transfers) have no empirical power. In this setting, any household behavior is consistent with the cyclical consistency and cyclical monotonicity condition, as shown by the entire distribution of average stability indices being one. This observation is consistent with the theoretical result shown in Corollary 2.

When prices and/or incomes are allowed to vary across the outside options, the stability conditions under mutual consent divorce seem to have some empirical content. Moreover, the larger the price and income variation, the harder it becomes for the data to satisfy the exact conditions. This fact is indicated by the lower average values of stability indices along the rows. 
Finally, given the nested structure of the models, the average stability indices are generally lower under mutual consent divorce without transfers than under mutual consent divorce with transfers. While most of the simulated behavior is consistent with the stability conditions under mutual consent divorce without transfers, a few simulations do not satisfy the exact restrictions. This is indicated by the minimum values in the distribution of average stability indices, which are strictly less than one.


\subsection{Sharing Rule Identification}
Next, we consider whether the stability conditions can be used to identify the intrahousehold allocations. 
As the models are non-parametric in nature, the identification is partial.
Specifically, the identified set includes all possible values of the unobserved parameter that are consistent with the stability conditions.
We can use the stability conditions to identify the household sharing rule, which measures the proportion of household expenditure that is allocated to individual members. Formally, they are defined as
$$
\eta^w = \frac{p_{\sigma(w),w}q^w_{\sigma(w),w} + P^w_{\sigma(w),w}Q_{\sigma(w),w}}{y_{\sigma(w),w}}, \text{ and }
\eta^{\sigma(w)} = \frac{p_{\sigma(w),w}q^{\sigma(w)}_{\sigma(w),w} + P^{\sigma(w)}_{\sigma(w),w}Q_{\sigma(w),w}}{y_{\sigma(w),w}}.
$$
According to the definition, $\eta^w + \eta^{\sigma(w)} = 1$. Therefore, identifying the female sharing rule is sufficient. The bounds for the female sharing rule can be found by maximizing or minimizing the linear function ($\eta^w$) subject to the stability conditions. 
As noted by \citet{cherchye2017household}, a static stable marriage model cannot identify the Lindahl prices ($P^w_{\sigma(w),w}, P^{\sigma(w)}_{\sigma(w),w} $) for public goods within the households. This is because the Lindahl prices for the observed couples do not enter any of the inequalities.\footnote{
In order to identify these prices, one can combine the stability conditions with collective rationality restrictions using a panel dataset of household consumption.
} 
Thus, any identification in the sharing rule will be driven by the identification of the private consumption ($q^w_{\sigma(w),w}, q^{\sigma(w)}_{\sigma(w),w} $).
For the sake of this illustration, in what follows, we will focus on the identification of the females' share of private consumption.

We will summarize the identifying power of the stability conditions by focusing on the difference between the upper and lower bounds. This will indicate the tightness of the identified set. We expect that a model with high identifying power will lead to tighter identification.
The LISS panel data provide information on assignable private consumption, which we can call the naive bounds for private consumption shares.
These bounds are not based on any model and can be directly inferred from the data.
f the model fails to improve upon these naive bounds, then it indicates that the model has no identifying power.

\begin{table}[htb]
\centering
\resizebox{\textwidth}{!}{
\begin{tabular}{lcHcHcHcHcHc}
\toprule
                                       & 0\%                                                                               & 5\%                                                                               & 5\%                                                                              & 15\%                                                                              & 10\%                                                                              & 25\%                                                                              & 15\%                                                                              & 35\%                                                                              & 20\%                                                                              & 45\%                                                                              & 25\%                                                                                       \\\midrule
                                            {\it Panel A: prices}       &&&&&&&&&&&\\
\multicolumn{1}{l}{unilateral}                    
& \begin{tabular}[c]{@{}c@{}}.0869\\ \footnotesize 0 .06 .52\end{tabular}  
& \begin{tabular}[c]{@{}c@{}}.0036\\ \footnotesize 0 0 .54\end{tabular}        
& \begin{tabular}[c]{@{}c@{}}.0008\\ \footnotesize 0 0 .14\end{tabular}         
& \begin{tabular}[c]{@{}c@{}}.0002\\ \footnotesize 0 0 .11\end{tabular}         & \begin{tabular}[c]{@{}c@{}}.0003\\ \footnotesize 0 0 .53\end{tabular}         & \begin{tabular}[c]{@{}c@{}}.0000\\ \footnotesize 0 0 0 \end{tabular}           & \begin{tabular}[c]{@{}c@{}}.0001\\ \footnotesize 0 0 .36\end{tabular}         & \begin{tabular}[c]{@{}c@{}}.0000\\ \footnotesize 0 0 0 \end{tabular}           & \begin{tabular}[c]{@{}c@{}}.0001\\ \footnotesize 0 0 .24\end{tabular}         & \begin{tabular}[c]{@{}c@{}}.0000\\ \footnotesize 0 0 0 .01 .45\end{tabular}       & \begin{tabular}[c]{@{}c@{}}.0001\\ \footnotesize 0 0 .29\end{tabular}                  \\
\multicolumn{1}{l}{mutual consent w/o transfers}                & \begin{tabular}[c]{@{}c@{}}.6070\\ \footnotesize .24  .61 .84\end{tabular} & \begin{tabular}[c]{@{}c@{}}.6070\\ \footnotesize .24  .61 .84\end{tabular} & \begin{tabular}[c]{@{}c@{}}.6070\\ \footnotesize .24  .61 .84\end{tabular} & \begin{tabular}[c]{@{}c@{}}.6070\\ \footnotesize .24  .61 .84\end{tabular} & \begin{tabular}[c]{@{}c@{}}.6070\\ \footnotesize .24  .61 .84\end{tabular} & \begin{tabular}[c]{@{}c@{}}.6070\\ \footnotesize .24  .61 .84\end{tabular} & \begin{tabular}[c]{@{}c@{}}.6070\\ \footnotesize .24  .61 .84\end{tabular} & \begin{tabular}[c]{@{}c@{}}.6070\\ \footnotesize .24  .61 .84\end{tabular} & \begin{tabular}[c]{@{}c@{}}.6070\\ \footnotesize .24  .61 .84\end{tabular} & \begin{tabular}[c]{@{}c@{}}.6070\\ \footnotesize .24  .61 .84\end{tabular} & \begin{tabular}[c]{@{}c@{}}.6070\\ \footnotesize .24  .61 .84\end{tabular}          \\
\multicolumn{1}{l}{mutual consent with transfers} 
& \begin{tabular}[c]{@{}c@{}}.6070\\ \footnotesize .24  .61 .84\end{tabular} 
& \begin{tabular}[c]{@{}c@{}}.2643\\ \footnotesize 0 0 .09 .56 .84\end{tabular}     & \begin{tabular}[c]{@{}c@{}}.0985\\ \footnotesize 0 0 .84\end{tabular}         & \begin{tabular}[c]{@{}c@{}}.0642\\ \footnotesize 0 0 .83\end{tabular}         & \begin{tabular}[c]{@{}c@{}}.0384\\ \footnotesize 0 0 .63\end{tabular}         & \begin{tabular}[c]{@{}c@{}}.0272\\ \footnotesize 0 0 .63\end{tabular}         & \begin{tabular}[c]{@{}c@{}}.0165\\ \footnotesize 0 0 .63\end{tabular}         & \begin{tabular}[c]{@{}c@{}}.0108\\ \footnotesize 0 0 .56\end{tabular}         & \begin{tabular}[c]{@{}c@{}}.0085\\ \footnotesize 0 0 .56\end{tabular}         & \begin{tabular}[c]{@{}c@{}}.0051\\ \footnotesize 0 0 .49\end{tabular}         & \begin{tabular}[c]{@{}c@{}}.0036\\ \footnotesize 0 0 .56\end{tabular}                  \\ \midrule
 {\it Panel B: income}          &&&&&&&&&&&\\
\multicolumn{1}{l}{unilateral}                    
& \begin{tabular}[c]{@{}c@{}}.0869\\ \footnotesize 0 .06 .52\end{tabular}   
& \begin{tabular}[c]{@{}c@{}}.0802\\ \footnotesize 0 .01 .05 .11 .59\end{tabular}   
& \begin{tabular}[c]{@{}c@{}}.0905\\ \footnotesize 0 .07 .63\end{tabular}   
& \begin{tabular}[c]{@{}c@{}}.0781\\ \footnotesize 0 .05  .63\end{tabular}     
& \begin{tabular}[c]{@{}c@{}}.0898\\ \footnotesize 0 .07  .63\end{tabular}     
& \begin{tabular}[c]{@{}c@{}}.0787\\ \footnotesize 0 .05  .63\end{tabular}     
& \begin{tabular}[c]{@{}c@{}}.0886\\ \footnotesize 0 .06  .63\end{tabular}     
& \begin{tabular}[c]{@{}c@{}}.0793\\ \footnotesize 0 .05  .63\end{tabular}     
& \begin{tabular}[c]{@{}c@{}}.0854\\ \footnotesize 0 .06  .63\end{tabular}     
& \begin{tabular}[c]{@{}c@{}}.0798\\ \footnotesize 0 0 .06 .11 .63\end{tabular}     
& \begin{tabular}[c]{@{}c@{}}.0857\\ \footnotesize 0 .06  .63\end{tabular}              \\
\multicolumn{1}{l}{mutual consent w/o transfers}                & \begin{tabular}[c]{@{}c@{}}.6070\\ \footnotesize .24  .61 .84\end{tabular} & \begin{tabular}[c]{@{}c@{}}.6070\\ \footnotesize .24  .61 .84\end{tabular} & \begin{tabular}[c]{@{}c@{}}.6070\\ \footnotesize .24  .61 .84\end{tabular} & \begin{tabular}[c]{@{}c@{}}.6070\\ \footnotesize .24  .61 .84\end{tabular} & \begin{tabular}[c]{@{}c@{}}.6070\\ \footnotesize .24  .61 .84\end{tabular} & \begin{tabular}[c]{@{}c@{}}.6070\\ \footnotesize .24  .61 .84\end{tabular} & \begin{tabular}[c]{@{}c@{}}.6070\\ \footnotesize .24  .61 .84\end{tabular} & \begin{tabular}[c]{@{}c@{}}.6070\\ \footnotesize .24  .61 .84\end{tabular} & \begin{tabular}[c]{@{}c@{}}.6070\\ \footnotesize .24  .61 .84\end{tabular} & \begin{tabular}[c]{@{}c@{}}.6070\\ \footnotesize .24  .61 .84\end{tabular} & {\begin{tabular}[c]{@{}c@{}}.6070\\ \footnotesize .24  .61 .84\end{tabular}} \\
\multicolumn{1}{l}{mutual consent with transfers} & \begin{tabular}[c]{@{}c@{}}.6070\\ \footnotesize .24  .61 .84\end{tabular} & \begin{tabular}[c]{@{}c@{}}.6070\\ \footnotesize .24  .61 .84\end{tabular} & \begin{tabular}[c]{@{}c@{}}.6070\\ \footnotesize .24  .61 .84\end{tabular} & \begin{tabular}[c]{@{}c@{}}.6070\\ \footnotesize .24  .61 .84\end{tabular} & \begin{tabular}[c]{@{}c@{}}.6070\\ \footnotesize .24  .61 .84\end{tabular} & \begin{tabular}[c]{@{}c@{}}.6070\\ \footnotesize .24  .61 .84\end{tabular} & \begin{tabular}[c]{@{}c@{}}.6070\\ \footnotesize .24  .61 .84\end{tabular} & \begin{tabular}[c]{@{}c@{}}.6070\\ \footnotesize .24  .61 .84\end{tabular} & \begin{tabular}[c]{@{}c@{}}.6070\\ \footnotesize .24  .61 .84\end{tabular} & \begin{tabular}[c]{@{}c@{}}.6070\\ \footnotesize .24  .61 .84\end{tabular} & \begin{tabular}[c]{@{}c@{}}.6070\\ \footnotesize .24  .61 .84\end{tabular}          \\ \midrule
 {\it Panel C: prices and income}          &&&&&&&&&&&\\
\multicolumn{1}{l}{unilateral}                    
& \begin{tabular}[c]{@{}c@{}}.0869\\ \footnotesize 0 .06 .52\end{tabular}   
& \begin{tabular}[c]{@{}c@{}}.0000\\ \footnotesize 0 0 0\end{tabular}         
& \begin{tabular}[c]{@{}c@{}}.0000\\ \footnotesize 0 0 0\end{tabular}         
& \begin{tabular}[c]{@{}c@{}}.0000\\ \footnotesize 0 0 0\end{tabular}         
& \begin{tabular}[c]{@{}c@{}}.0000\\ \footnotesize 0 0 0\end{tabular}         
& \begin{tabular}[c]{@{}c@{}}.0000\\ \footnotesize 0 0 0\end{tabular}         
& \begin{tabular}[c]{@{}c@{}}.0000\\ \footnotesize 0 0 0\end{tabular}         
& \begin{tabular}[c]{@{}c@{}}.0000\\ \footnotesize 0 0 0\end{tabular}         
& \begin{tabular}[c]{@{}c@{}}.0000\\ \footnotesize 0 0 0\end{tabular}         
& \begin{tabular}[c]{@{}c@{}}.0000\\ \footnotesize 0 0 0\end{tabular}         
& \begin{tabular}[c]{@{}c@{}}.0000\\ \footnotesize 0 0 0\end{tabular}                  \\
\multicolumn{1}{l}{mutual consent w/o transfers}                & \begin{tabular}[c]{@{}c@{}}.6070\\ \footnotesize .24  .61 .84\end{tabular} & \begin{tabular}[c]{@{}c@{}}.6070\\ \footnotesize .24  .61 .84\end{tabular} & \begin{tabular}[c]{@{}c@{}}.6070\\ \footnotesize .24  .61 .84\end{tabular} & \begin{tabular}[c]{@{}c@{}}.6070\\ \footnotesize .24  .61 .84\end{tabular} & \begin{tabular}[c]{@{}c@{}}.6070\\ \footnotesize .24  .61 .84\end{tabular} & \begin{tabular}[c]{@{}c@{}}.6070\\ \footnotesize .24  .61 .84\end{tabular} & \begin{tabular}[c]{@{}c@{}}.6070\\ \footnotesize .24  .61 .84\end{tabular} & \begin{tabular}[c]{@{}c@{}}.6070\\ \footnotesize .24  .61 .84\end{tabular} & \begin{tabular}[c]{@{}c@{}}.6070\\ \footnotesize .24  .61 .84\end{tabular} & \begin{tabular}[c]{@{}c@{}}.6070\\ \footnotesize .24  .61 .84\end{tabular} & \begin{tabular}[c]{@{}c@{}}.6070\\ \footnotesize .24  .61 .84\end{tabular}          \\
\multicolumn{1}{l}{mutual consent with transfers} 
& \begin{tabular}[c]{@{}c@{}}.6070\\ \footnotesize .24  .61 .84\end{tabular} 
& \begin{tabular}[c]{@{}c@{}}.0000\\ \footnotesize 0 0 0 0 0\end{tabular}       
& \begin{tabular}[c]{@{}c@{}}.0000\\ \footnotesize 0 0 0\end{tabular}         
& \begin{tabular}[c]{@{}c@{}}.0000\\ \footnotesize 0 0 0\end{tabular}         
& \begin{tabular}[c]{@{}c@{}}.0000\\ \footnotesize 0 0 0\end{tabular}         
& \begin{tabular}[c]{@{}c@{}}.0000\\ \footnotesize 0 0 0\end{tabular}         
& \begin{tabular}[c]{@{}c@{}}.0000\\ \footnotesize 0 0 0\end{tabular}         
& \begin{tabular}[c]{@{}c@{}}.0000\\ \footnotesize 0 0 0 \end{tabular}           
& \begin{tabular}[c]{@{}c@{}}.0000\\ \footnotesize 0 0 0 \end{tabular}           
& \begin{tabular}[c]{@{}c@{}}.0000\\ \footnotesize 0 0 0 \end{tabular}           & \begin{tabular}[c]{@{}c@{}}.0000\\ \footnotesize 0 0 0 \end{tabular}       
\\ \bottomrule
\end{tabular}
}
{
\\[0.3em]
\raggedright  \scriptsize {\it Notes:} The number in each cell is the mean value of the difference between upper and lower bounds on female private consumption shares in the given scenario. The three numbers below correspond to the minimum, median and maximum of the distribution. 
\par}
\caption{Intrahousehold identification}
\label{tab:SharingRuleResults}
\end{table}

Table \ref{tab:SharingRuleResults} presents the results for the identification of the female shares of private consumption. The structure of the table is similar to Table \ref{tab:ResultsDivorceCosts}.
Panels A, B, and C correspond to the scenarios where we introduce variation in prices, income, and both prices and income, respectively.
Each panel contains three rows for the three theories under consideration: unilateral divorce, mutual consent divorce without transfers, and mutual consent divorce with transfers.
There are six columns corresponding to the different scenarios from the default $\alpha = 0$\%, where the original prices and incomes are used in the computation, to $\alpha = 25$\% allowing for a large variation in the prices and/or incomes across potential matches.

In line with the results in Table \ref{tab:ResultsDivorceCosts}, Table \ref{tab:SharingRuleResults} shows that the stability conditions under unilateral divorce exhibit the highest identifying power. The average width of the identified bounds is the smallest from the unilateral divorce model. 
This is expected as this model is nested within both the mutual consent models.
Even when prices and income are assumed to be the same across counterfactual matches, the stability conditions under unilateral divorce lead to tight identification with an average width of 8.69 percentage points. 
By contrast, the stability conditions under mutual consent divorce (with and without transfers) exhibit no identifying power if there is no variation in prices and incomes. The identified bounds are the same as the naive bounds.\footnote{
    As noted above, the naive bounds are identified using the available information on assignable private consumption. For the 100 randomly drawn markets, the average width of the naive bounds is 0.6070.
    Note that the naive bounds are the same for all the scenarios as we use the same marriage markets for all the simulations across the three scenarios.
}


When we introduce variation in prices and/or income, the identification power of the unilateral model improves. 
With regards to mutual consent divorce, the stability conditions exhibit some identifying power if there is variation in prices and post-divorce transfers are allowed (see Panels A and C in Table \ref{tab:SharingRuleResults}). For example, we can see a significant improvement in the model's identifying power when there is a small price variation ($\alpha$ = 5\%). 
The average width of the identified bounds reduces from 60.70 percentage points when $\alpha = 0$\% to 9.85 percentage points when $\alpha = 5$\%.
However, the model has no identifying power if there is income but no price variation. 
This is in line with our theoretical result in Corollary \ref{cor:MutualConsentTransfersNoIdentification}. 
Note that income variation does not significantly improve the identifying power even for the unilateral divorce model (see Panel B in Table \ref{tab:SharingRuleResults}).
Finally, having both price and income variation provides point identification for both unilateral divorce and mutual consent divorce with transfer models. This is true even when the degree of variation is small ($\alpha=5\%$; see Panel C in Table \ref{tab:SharingRuleResults}).


\section{Linear Programming Formulation}
In this section, we present linear programming formulations of the rationalizability conditions shown in the text.
While the combinatorial conditions presented in the main text are more elegant in a mathematical way, the linear programming formulations are more convenient for computational purposes.

\paragraph{Unilateral Divorce.}
The rationalizability condition for stability under unilateral divorce is straightforward. We can reformulate the no negative edges condition in terms of linear inequalities. Given a dataset $\mathcal{D}$, $A(\mathcal{D})$ has no negative edges if and only if there exist individual private consumption shares $q^m_{m,\sigma(m)}, q^{\sigma(m)}_{m,\sigma(m)} \in \R^n_+ $, Lindahl prices $P^m_{m,w}, P^w_{m,w} \in \R^N_+$ and edge weights $a_{m,w}$ such that, for all $m\in M\cup\{\emptyset \} \text{ and } w\in W\cup\{\emptyset \}$, 
\begin{equation*}
\begin{cases}
& q^m_{m,\sigma(m)}+ q^{\sigma(m)}_{m,\sigma(m)}  = q_{m,\sigma(m)} \\
& P^m_{m,w} + P^w_{m,w}  = P_{m,w} \\
& a_{m,w} = p_{m,w} q^m_{m,\sigma(m)} + P^m_{m,w} Q_{m,\sigma(m)} + p_{m,w} q^w_{\sigma(w),w}  + P^w_{m,w} Q_{\sigma(w),w} - y_{m,w} \\
& a_{m,w} \geq 0
\end{cases}
\end{equation*}

\paragraph{Mutual Consent Divorce without Transfers.}
The cyclical consistency condition is equivalent to GARP \citep[see][]{varian1982nonparametric}.
Thus, we can use the revealed preference framework to construct a mixed integer programming formulation for the cyclical consistency condition.\footnote{
    The equivalence between the mixed integer program and the GARP (and thus cyclical consistency) can be inferred from \citet{demuynck2020computing,freer2022welfare}.
}
Given a dataset $\mathcal{D}$, $A(\mathcal{D})$ satisfies cyclical consistency if and only if there exist individual private consumption shares $q^m_{m,\sigma(m)}, q^{\sigma(m)}_{m,\sigma(m)} \in \R^n_+ $, Lindahl prices $P^m_{m,w}, P^w_{m,w} \in \R^N_+$, edge weights $a_{m,w}$, numbers $z_m \in \R_+$ and binary variables $\delta_{m,w}$ such that, for all $m\in M\cup\{\emptyset \} \text{ and } w\in W\cup\{\emptyset \}$, 
\begin{equation*}
\begin{cases}
& q^m_{m,\sigma(m)}+ q^{\sigma(m)}_{m,\sigma(m)}  = q_{m,\sigma(m)} \\
& P^m_{m,w} + P^w_{m,w}  = P_{m,w} \\
& a_{m,w} = p_{m,w} q^m_{m,\sigma(m)} + P^m_{m,w} Q_{m,\sigma(m)} + p_{m,w} q^w_{\sigma(w),w}  + P^w_{m,w} Q_{\sigma(w),w} - y_{m,w} \\
& a_{m,w} + \delta_{m,w}M \ge 0 \\
& a_{m,w} - (1 - \delta_{m,w})M < 0 \\
& z_m + (1-\delta_{m,w})M > z_{\sigma(w)}
\end{cases}
\end{equation*}

\noindent
where $M \rightarrow \infty$ is a big-M so the $\delta M$ terms can be used to switch the inequalities ``on'' and ``off''. Using the big-M method, we have restated the cyclical consistency condition as a mixed integer program. Intuitively, the binary variables $\delta_{m,w}$ switch between the fourth and fifth inequalities. That is, $\delta_{m,w}=1$ if and only if $a_{m,w}<0$.
If $a_{m,w}<0$, then the fourth inequality requires that $\delta_{m,w}=1$. 
Conversely, if $\delta_{m,w} = 1$, then the fifth inequality is simplified to $a_{m,w}<0$.

\paragraph{Mutual Consent Divorce with Transfers.}
For mutual consent divorce with transfers, we can directly use the connection between cyclical monotonicity and rationalizability with quasilinear utility \citep[see][]{castillo2019general}.
Given a dataset $\mathcal{D}$, $A(\mathcal{D})$ satisfies cyclical monotonicity if and only if there exist individual private consumption shares $q^m_{m,\sigma(m)}, q^{\sigma(m)}_{m,\sigma(m)} \in \R^n_+ $, Lindahl prices $P^m_{m,w}, P^w_{m,w} \in \R^N_+$, edge weights $a_{m,w}$ and numbers $z_m \in \R_+$ such that, for all $m\in M\cup\{\emptyset \} \text{ and } w\in W\cup\{\emptyset \}$, 
\begin{equation*}
\begin{cases}
& q^m_{m,\sigma(m)}+ q^{\sigma(m)}_{m,\sigma(m)}  = q_{m,\sigma(m)} \\
& P^m_{m,w} + P^w_{m,w}  = P_{m,w} \\
& a_{m,w} = p_{m,w} q^m_{m,\sigma(m)} + P^m_{m,w} Q_{m,\sigma(m)} + p_{m,w} q^w_{\sigma(w),w}  + P^w_{m,w} Q_{\sigma(w),w} - y_{m,w} \\
& z_m - z_{\sigma(w)} \le a_{m,w}
\end{cases}
\end{equation*}
The sufficiency of this linear program can be derived straightforwardly. Given any cycle of remarriages, if we add up the last inequality along the cycle then the numbers on the left-hand side would cancel out. So we would end up with the requirement that the sum of $a_{m,w}$ along the cycle is non-negative, that is exactly the definition of cyclical monotonicity. For the proof of the necessity of this linear program, we refer the reader to \citet{castillo2019general} who have provided direct construction of numbers $z_m$ which would satisfy the above inequalities if the graph satisfies cyclical monotonicity.

\bibliographystyle{plainnat}
\bibliography{refs}